\documentclass[onecolumn, draftcls,11pt]{IEEEtran}
\usepackage{amsfonts}                                                          
\usepackage{epsfig}
\usepackage{amsthm}
\usepackage{graphicx}        
\usepackage{latexsym}
\usepackage{amssymb}
\usepackage{amsmath}
\usepackage{parskip}
\usepackage{multirow}
\usepackage{cite}
\usepackage{mathtools}
\usepackage{epstopdf}
\usepackage{etoolbox}
\usepackage{array}
\usepackage{setspace}
\usepackage{mathptmx}
\usepackage[colorlinks=true]{hyperref}
\usepackage[bottom]{footmisc}
\usepackage{tikz}

\usetikzlibrary{decorations.pathreplacing}

\usepackage{verbatim}
\usepackage{calrsfs}

\DeclareMathAlphabet{\pazocal}{OMS}{zplm}{m}{n}

\newcommand{\mb}{\mathbb}

\let\bbordermatrix\bordermatrix
\patchcmd{\bbordermatrix}{8.75}{4.75}{}{}
\patchcmd{\bbordermatrix}{\left(}{\left[}{}{}
\patchcmd{\bbordermatrix}{\right)}{\right]}{}{}

\newcommand{\sr}{\stackrel}

\newcommand{\rar}{\rightarrow}

\newcommand{\tri}{\sr{\triangle}{=}}

\newcommand{\be}{\begin{equation}}
\newcommand{\ee}{\end{equation}}
\newcommand{\bea}{\begin{eqnarray}}
\newcommand{\eea}{\end{eqnarray}}
\newcommand{\bes}{\begin{eqnarray*}}
\newcommand{\ees}{\end{eqnarray*}}
\newcommand{\bce}{\begin{center}}
\newcommand{\ece}{\end{center}}
\newcommand{\beae}{\begin{IEEEeqnarray}{rCl}}
\newcommand{\eeae}{\end{IEEEeqnarray}}

\def\VR{\kern-\arraycolsep\strut\vrule &\kern-\arraycolsep}
\def\vr{\kern-\arraycolsep & \kern-\arraycolsep}

\newcommand{\ben}{\begin{enumerate}}
\newcommand{\een}{\end{enumerate}}

\newcommand{\hso}{\hspace{.1in}}
\newcommand{\hst}{\hspace{.2in}}

\newcommand{\noi}{\noindent}

\newtheorem{theorem}{Theorem}[section]

\newtheorem{remark}{Remark}[section]
\newtheorem{corollary}{Corollary}[section]

\newtheorem{definition}{Definition}[section]
\newtheorem{lemma}{Lemma}[section]
\newtheorem{example}{Example}[section]

\newtheorem{conclusion}{Conclusion}[section]
\onehalfspace

\begin{document}

%\baselineskip=16pt
%\sloppy

%% Paper Title
%% You can use linebreaks \\ within to get better formatting as
%% desired.
\title{Information Structures of Maximizing Distributions of Feedback Capacity  for General  Channels with Memory \& Applications
}

%% Author names and affiliations:
%%
%% Avoiding spaces at the end of the author lines is not a problem with
%% conference papers because we don't use \thanks or \IEEEmembership.
%%
%% For several authors with only one affiliation:
%%
% \author{
%   \IEEEauthorblockN{Hui-Ting Chang and Stefan M.~Moser}
%   \IEEEauthorblockA{Department of Electrical and Computer Engineering\\
%     National Chiao Tung University (NCTU)\\
%     Hsinchu, Taiwan\\
%     Email: \{email-of-hui-ting,email-of-stefan\}@ieee.org}
% }
%%
%% For up to three affiliations:
%%

\author{Charalambos~D.~Charalambous and Christos~K.~Kourtellaris and Ioannis Tzortzis
%\thanks{Manuscript received October 9,~2012;~revised....}
\thanks{This work was financially supported by a medium size University of Cyprus grant entitled ``DIMITRIS".}
\thanks{The authors are with the Department of Electrical and Computer Engineering, University of Cyprus, 75 Kallipoleos Avenue, P.O. Box 20537, Nicosia, 1678, Cyprus, e-mail: $\{chadcha,kourtellaris.christos,tzortzis.ioannis@ucy.ac.cy\}$
}}

%%
%% For over three affiliations, or if they all won't fit within the width
%% of the page, use this alternative format:
%%
% \author{
%   \IEEEauthorblockN{
%     Michael Shell\IEEEauthorrefmark{1},
%     Homer Simpson\IEEEauthorrefmark{2},
%     James Kirk\IEEEauthorrefmark{3},
%     Montgomery Scott\IEEEauthorrefmark{3} and
%     Eldon Tyrell\IEEEauthorrefmark{4}}
%   \IEEEauthorblockA{
%     \IEEEauthorrefmark{1}School of Electrical and Computer Engineering\\
%     Georgia Institute of Technology, Atlanta, Georgia 30332--0250\\
%     Email: see http://www.michaelshell.org/contact.html}
%   \IEEEauthorblockA{
%     \IEEEauthorrefmark{2}Twentieth Century Fox, Springfield, USA\\
%     Email: homer@thesimpsons.com}
%   \IEEEauthorblockA{
%     \IEEEauthorrefmark{3}Starfleet Academy, San Francisco, California 96678-2391\\
%     Telephone: (800) 555--1212, Fax: (888) 555--1212}
%   \IEEEauthorblockA{
%     \IEEEauthorrefmark{4}Tyrell Inc., 123 Replicant Street, Los Angeles, California 90210--4321}
% }

\maketitle

\newpage

\tableofcontents
\newpage

\begin{abstract}
For any  class of  channel conditional distributions, with finite memory dependence on channel input RVs $A^n \tri \{A_i: i=0, \ldots, n\}$ or channel output RVs $B^n \tri \{B_i: i=0, \ldots, n\}$ or both,  
 we characterize the sets of channel input distributions, which maximize  directed information defined by
\begin{align}
 I(A^n \rar B^n) \tri \sum_{i=0}^n I(A^i;B_i|B^{i-1})  \nonumber %  \label{main_p}
\end{align}
and we derive the corresponding expressions, called  ``characterizations of Finite Transmission Feedback Information (FTFI) capacity''. The main theorems state that optimal channel input distributions occur in subsets ${\cal P}_{[0,n]}^{CI}\subseteq {\cal  P}_{[0,n]} \tri \big\{  {\bf P}_{A_i|A^{i-1}, B^{i-1}}:  i=0, \ldots, n\big\}$,  which satisfy conditional independence on past information. We derive similar characterizations, when general  transmission cost constraints are imposed. Moreover, we also show that the structural properties apply to general nonlinear and linear autoregressive channel models defined by discrete-time recursions on general alphabet spaces, and driven by arbitrary distributed noise processes.

We derive these structural properties  by invoking  stochastic optimal control theory and  variational equalities of directed information,  to identify  tight upper bounds on $I(A^n \rar B^n)$,  which are   achievable over subsets  of conditional  distributions  ${\cal P}_{[0,n]}^{CI} \subseteq {\cal P}_{[0,n]}$, which satisfy conditional independence and they are  specified by the dependence of  channel distributions and transmission cost functions on inputs and output symbols. 

We apply the characterizations to recursive Multiple Input Multiple Output Gaussian Linear Channel Models with  limited  memory on channel input and output sequences, and we show a separation principle between the computation of the elements of the optimal strategies.

The structural properties of optimal channel input distributions,  generalize the structural properties of Memoryless Channels with feedback, expressed in terms of conditional independence,  to any channel distribution with memory, and settle various long standing open problems in information theory. 

%,  under  general transmission cost constraints.   

%Encoder strategies which achieve the FTFI capacity and Feedback Capacity are constructed by utilizing the information structures of the optimal channel input distributions.

%Sufficient conditions are identified  for  feedback not to increase capacity of channels with memory. 

\end{abstract}

\section{Introduction}
\label{introduction}
Shannon's  mathematical model of a communication channel with feedback is defined  by $$\Big(\big\{ {\mb A}_i: i=-\infty, \ldots, n\big\},\big\{ {\mb B}_i: i=-\infty, \ldots, n\big\}, \big\{ {\bf P}_{A_i|A^{i-1}, B^{i-1}}:i=0,  \ldots, n\big\}, \big\{ {\bf P}_{B_i|B^{i-1}, A^{i}}:i=0,  \ldots, n\big\} \Big)$$ 
where $a^n\tri \{\ldots, a_{-1}, a_0, a_1, \ldots, a_n\}\in \times_{i=-\infty}^n {\mb A}_i$ are the    channel input symbols, $b^n \tri \{\ldots, b_{-1}, b_0, b_1, \ldots, b_n\}\in \times_{i=-\infty}^n {\mb B}_i$ are the channel output symbols, 
 ${\cal P}_{[0,n]} \tri \big\{ {\bf P}_{A_i|A^{i-1}, B^{i-1}}:i=0, 1, \ldots, n\big\}$ is the sequence of channel input conditional distributions with feedback, ${\cal C}_{[0,n]}\tri \big\{ {\bf P}_{B_i|B^{i-1}, A^{i}}:i=0, 1, \ldots, n\big\}$ is the sequence of channel conditional  distributions, and the initial distribution ${\bf P}_{A^{-1}, B^{-1}}\equiv \nu(da^{-1},db^{-1})$ is fixed. \\
 Shannon's operational definion for reliable communication of information over the channel is described via a sequence of feedback codes  $\{(n, { M}_n, \epsilon_n):n=0, 1, \dots\}$, which consist of the following elements.  \\
(a)  A set of uniformly distributed messages ${\cal M}_n \tri \{ 1,  \ldots, M_n\}$ and a set of encoding strategies,  mapping messages  into channel inputs of block length $(n+1)$, defined by\footnote{The superscript on expectation, i.e., ${\bf P}^g$ indicates the dependence of the distribution on the encoding strategies.} 
\begin{align}
{\cal E}_{[0,n]}^{FB}(\kappa) \triangleq & \Big\{g_i: {\cal M}_n \times {\mathbb A}^{i-1} \times {\mb B}^{i-1}  \longmapsto {\mb A}_i, \hso  a_0=g_0(w), a_1=g_1(w,a_0,b_0),\ldots, a_n=g_n(w, a^{n-1}, b^{n-1}), \nonumber \\
&\hso  w\in {\cal M}_n: \hso  \frac{1}{n+1} {\bf E}^g\Big(c_{0,n}(A^n,B^{n-1})\Big)\leq \kappa  \Big\}, \hso n=0, 1, \ldots. \label{block-code-nf-non}
\end{align}
The codeword for any $w \in {\cal M}_n$  is $u_w\in{\mb A}^n$, $u_w=(g_0(w), g_1(w, a_0, b_0),\dots,g_n(w, a^{n-1}, b^{n-1}))$, and ${\cal C}_n=( u_1,u_2,\dots,u_{{M}_n})$ is  the code for the message set ${\cal M}_n$, and $\{A^{-1}, B^{-1}\}=\{\emptyset\}$.  In general, the code  may depend on the initial data, depending on the convention, i.e.,  $(A^{-1}, B^{-1})=(a^{-1}, b^{-1})$, which are known to the encoder and decoder (unless specified otherwise). 
%(unless  it can be shown that in  the limit, as $n \longrightarrow \infty$, the induced channel output process  has a unique invariant distribution).  
\\
(b)  Decoder measurable mappings $d_{0,n}:{\mb B}^n\longmapsto {\cal M}_n$,  such that the average
probability of decoding error satisfies
\begin{align}
{\bf P}_e^{(n)} \triangleq \frac{1}{M_n} \sum_{w \in {\cal M}_n} {\bf  P}^g \Big\{d_{0,n}(B^{n}) \neq w |  W=w\Big\}\equiv {\bf  P}^g\Big\{d_{0,n}(B^n) \neq W \Big\} \leq \epsilon_n\nonumber
\end{align}
and the decoder may also assume knowledge of the initial data.\\
The coding rate or transmission rate over the channel is defined by  $r_n\triangleq \frac{1}{n+1} \log M_n$.
A rate $R$ is said to be an achievable rate, if there exists  a  code sequence satisfying
$\lim_{n\longrightarrow\infty} {\epsilon}_n=0$ and $\liminf_{n \longrightarrow\infty}\frac{1}{n+1}\log{{M}_n}\geq R$. \\The operational definition of feedback capacity of the channel is the supremum of all achievable rates, i.e., $C\triangleq \sup \{R: R \: \: \mbox{is achievable}\}$.

 Given a source process $\big\{X_i: i=0,1, \ldots, \big\}$ with finite entropy rate, which is mapped into messages to be encoded and transmitted over the channel, and  satisfies  conditional independence \cite{massey1990} 
\begin{align}
{\bf P}_{B_i|B^{i-1}, A^i, X^k}={\bf P}_{B_i|B^{i-1}, A^i} \hso   \forall k \in \{0,1, \ldots, n\},\hso i=0, \ldots, n \label{CI_Massey_N} 
\end{align}
 under appropriate conditions, it is shown in  \cite{kramer2003,tatikonda-mitter2009,permuter-weissman-goldsmith2009}, using tools from \cite{dobrushin1959,pinsker1964,gallager1968,blahut1987,cover-thomas2006,ihara1993,verdu-han1994,han2003},  that the supremum of all achievable rates is characterized  by the information  quantity $C_{A^\infty \rar B^\infty}^{FB}$, defined by the extremum problem 
%Coding theorems for  channels with memory with and without feedback are developed extensively over the years, in an anthology of papers, such as,   \cite{dobrushin1959,pinsker1964,gallager1968,blahut1987,cover-thomas2006,ihara1993,verdu-han1994,han2003,kramer2003,tatikonda-mitter2009,permuter-weissman-goldsmith2009,gamal-kim2011}.\\ 
\begin{align}
C_{A^\infty\rar  B^\infty}^{FB} \tri  \liminf_{n \longrightarrow \infty} \frac{1}{n+1} C_{A^n \rar  B^n}^{FB}, \hst  C_{A^n \rar B^n}^{FB} \tri   \sup_{    {\cal P}_{[0,n]} }I(A^n\rar B^n) \label{cap_fb_1}
 \end{align}
where $I(A^n \rar B^n)$ is the directed information from $A^n$ to $B^n$,  defined by \cite{marko1973,massey1990}
 \begin{align}
 I(A^n\rar B^n) \tri & \sum_{i=0}^n I(A^i;B_i|B^{i-1})= \sum_{i=0}^n {\bf E}_{\nu} \Big\{ \log \Big( \frac{ d{\bf P}_{B_i|B^{i-1}, A^i}(\cdot|B^{i-1}, A^i)}{d{\bf P}_{B_i|B^{i-1}}(\cdot|B^{i-1})}(B_i)\Big)\Big\} \label{intro_fbc1a}
  \end{align}
Here, ${\bf E}_\nu\{\cdot\}$ denotes expectation with respect to the joint distribution induced   by the channel  input conditional distribution from  ${\cal P}_{[0,n]}$, the specific channel conditional distribution from ${\cal C}_{[0,n]}$,   and the initial distribution $\nu(da^{-1},db^{-1})$.\\
A fundamental problem in such extremum problems of directed information,  is to determine the information structures of optimal channel input conditional distributions
 ${\cal P}_{[0,n]} \tri \big\{ {\bf P}_{A_i|A^{i-1}, B^{i-1}}:i=0, 1, \ldots, n\big\}$, for any  class of channel distributions, which maximize $I(A^n\rar B^n)$, equivalently, to characterize the subsets of ${\cal P}_{[0,n]}$ which satisfy conditional independence and maximize $I(A^n \rar B^n)$.

Our interest in the structural properties of optimization problem $C_{A^n\rar  B^n}^{FB}$ is the following. From the   converse coding theorem \cite{massey1990,kramer1998,permuter-weissman-goldsmith2009}, in view of (\ref{CI_Massey_N}), if the supremum over channel input distributions in $C_{A^n\rar  B^n}^{FB}$ exists, and its per unit time limit  exists and it is finite, then $C_{A^\infty\rar  B^\infty}^{FB}$ is a non-trivial  upper bound on the supremum of all achievable rates of feedback codes-the feedback capacity, while under stationary ergodicity or Dobrushin's directed information stability \cite{dobrushin1959,pinsker1964,ihara1993,tatikonda-mitter2009,permuter-weissman-goldsmith2009}, then $C_{A^\infty\rar  B^\infty}^{FB}$ is indeed the feedback capacity. \\ 
 When transmission cost constraints are imposed of the form (or variants of them)
\begin{align}
{\cal P}_{[0,n]}(\kappa)\tri  \Big\{ {\bf P}_{A_i|A^{i-1}, B^{i-1}},  i=0, \ldots, n:  \frac{1}{n+1} {\bf E}_\nu\Big(\sum_{i=0}^n \gamma_i(T^i A^n, T^i B^{n})\Big) \leq \kappa\Big\}, \hso \kappa \in [0, \infty) \label{cap_fb_3}
\end{align}
 the optimization problem (\ref{cap_fb_1}) is replaced by  
\begin{align}
C_{A^\infty\rar  B^\infty}^{FB}(\kappa) \tri  \liminf_{n \longrightarrow \infty} \frac{1}{n+1} C_{A^n \rar  B^n}^{FB}(\kappa), \hst  C_{A^n \rar B^n}^{FB}(\kappa) \tri   \sup_{    {\cal P}_{[0,n]}(\kappa) }I(A^n\rar B^n) \label{cap_fb_1_TC}
\end{align}
where  for  each $i$, the dependence of transmission cost function $\big\{\gamma_i(\cdot, \cdot): i=0, \ldots,n \big\}$, on input and output symbols is specified by  $T^i a^n \subseteq \{a_0, a_1, \ldots, a_i\},  T^i b^{n} \subseteq \{b_0, b_1, \ldots, b_{i}\}$, and these are either fixed or nondecreasing with $i$, for        $i=0,1, \ldots, n$. 

Our main objective is the following. Given a specific channel distribution and transmission cost function, we wish to determine the subsets of optimal channel input  distributions $ {\cal P}_{[0,n]}^{CI} \subseteq {\cal P}_{[0,n]}$ and $ {\cal P}_{[0,n]}^{CI}(\kappa) \subseteq {\cal P}_{[0,n]}(\kappa)$, which satisfy conditional independence and correspond to the  maximizing subsets of the extremum problems $C_{A^n\rar  B^n}^{FB}$ and $C_{A^n\rar  B^n}^{FB}(\kappa)$, respectively. Then to determine the corresponding characterizations, called Finite Transmission Feedback Information (FTFI) Capacity and Feedback capacity (i.e., their per unit time limiting versions), as it is done for Discrete Memoryless Channels (DMCs).  

\subsection{Literature Review}
\label{R-P-L}
Shannon and subsequently Dobrushin \cite{dobrushin1958} characterized the capacity of DMCs  (and  memoryless channels with continuous alphabets, subject to transmission cost $\int |a|^2 {\bf P}_A(da) \leq \kappa$), with and without feedback, and obtained the well-known two-letter expression 
\bea
C\tri \max_{{\bf P}_A} I(A;B). \label{cap_fb_c}
\eea
For memoryless channels without feedback, this characterization is obtained from the  upper bound
\bea
C_{A^n ; B^n}\tri \max_{{\bf P}_{A^n}} I(A^n;B^n) \leq \; \max_{{\bf P}_{A_i}, i=0, \ldots, n} \sum_{i=0}^n I(A_i;B_i)\leq (n+1) C. \label{cap_nf_c}
\eea
since this bound is achievable, when the  channel input distribution  satisfies conditional independence $
{\bf P}_{A_i|A^{i-1}}(da_i|a^{i-1})={\bf P}_{A_i}(da_i),         i=0, 1, \ldots, n$, and   $\{A_i:i =0, 1, \ldots, \}$ is identically distributed, which then implies  the  joint process $\{(A_i, B_i): i=0,1, \ldots, \}$ is independent and identically distributed. \\
For memoryless channels  with feedback, (\ref{cap_fb_c}) is often obtained by first applying  the converse to the coding theorem,  to show that feedback does not increase capacity \cite{cover-pombra1989}, which then implies
%provided it is  shown 
%which then implies
\begin{align}
 {\bf P}_{A_i|A^{i-1}, B^{i-1}}(da_i|a^{i-1}, b^{i-1})={\bf P}_{A_i}(da_i),        \hso  i=0, 1, \ldots, n \label{CI_DMC}
\end{align} 
   and $C$ is obtained if  $\{A_i:i =0, 1, \ldots, \}$ is identically distributed. That is, since feedback does not increase capacity, then mutual information and directed information are identical, in view of (\ref{CI_DMC}). However, as pointed out elegantly by Massey \cite{massey1990}, for channels with feedback it will be a mistake to use the same arguments as  in (\ref{cap_nf_c}). 
    The  conditional independence conditions imply  that the {\it Information Structure} of the maximizing channel input distributions is  the {\it Null Set}. \\  
In Section~\ref{randomized}, we develop a  methodology for directed information, which in principle, repeats the above steps,  to show that for many classes of channel distribution with memory subject to transmission cost constraints, the optimal channel input distributions occur in subsets, characterized by conditional independence. However,  each of the steps is  more involved due to the memory of the channels, and hence new tools are introduced  to established these achievable upper bounds.

     Cover and Pombra \cite{cover-pombra1989} (see also  \cite{ebert1970,ihara1993}) characterized the feedback capacity of  non-stationary  non-ergodic  Additive Gaussian Noise (AGN) channels  with memory,  
defined by 
\begin{align}
 B_i=A_i+V_i, \hst i=0,1, \ldots, n, \hst \frac{1}{n+1} \sum_{i=0}^{n+1} {\bf E} \Big\{ |A_i|^2\Big\} \leq \kappa, \hso \kappa \in [0,\infty) \label{c-p1989}    
\end{align}
where $\{V_i : i=0,1, \ldots, n\}$ is a real-valued jointly non-stationary Gaussian process $N(\mu_{V^n}, K_{V^n})$, under the  assumption that ``$A^n$ is causally related to $V^n$'' defined by\footnote{\cite{cover-pombra1989}, page 39, above Lemma~5.}  
\bea
 {\bf P}_{A^n, V^n}(da^n, dv^n)=\Big( \otimes_{i=0}^n  {\bf P}_{A_i|A^{i-1}, V^{i-1}}(da_i|a^{i-1}, v^{i-1})\Big) \otimes   {\bf P}_{V^n}(dv^n). \label{CR_1}
\eea 
In  \cite{cover-pombra1989},  the authors characterized feedback  capacity, via the maximization of mutual information between uniformly distributed messages and the channel output process, denoted by $I(W, B^n)$, and obtained the following  characterization \footnote{The methodology in \cite{cover-pombra1989} utilizes the converse coding theorem to obtain an upper bound on the entropy $H(B^n)$, by restricting  $\{A_i: i=0, \ldots, n\}$ to a Gaussian process.}.
\begin{align}
 C_{W; B^n}^{FB, CP}(\kappa) \tri& \max_{ \Big\{\frac{1}{n+1} \sum_{i=0}^n {\bf E} |A_i|^2 \leq \: \kappa \Big\} } I(W, B^n)=    \max_{ \Big\{\frac{1}{n+1} \sum_{i=0}^n {\bf E} |A_i|^2 \leq \: \kappa: \hso A_i=\sum_{j=0}^{i-1} \overline{\gamma}_{i,j}V_{j} + \overline{Z}_i : \hso i=0, 1, \ldots, n \Big\} }  H(B^n) - H(V^n ) \label{cp1989_a} \\
 =& \sup_{\big(\overline{\Gamma}^n, K_{\overline{Z}^n}\big): \frac{1}{n+1}Tr\big(\overline{\Gamma}^n K_{{V}^n} \overline{\Gamma}^T+ K_{\overline{Z}^n}\big)\leq \kappa}  \frac{1}{2}\log \frac{\Big|\big(\overline{\Gamma}^n+I\big)K_{V^n}(\overline{\Gamma}^n+I)^T+K_{\overline{Z}^n}\Big|}{\Big|K_{V^n}\Big|}  \label{cp1989}
  \end{align}
 where $\overline{Z}^n \tri \{\overline{Z}_i: i=0, 1, \ldots,n\}$ is a Gaussian process $N(0, K_{{\overline{Z}}^n})$,  orthogonal to $V^{n}\tri \{V_i: i=0, \ldots, n\}$, and   $\{\overline{\gamma}_{i,j}: i, j=0, \ldots, n\}$ are deterministic functions, which constitute  the entries of the lower diagonal matrix $\overline{\Gamma}^n$.  The feedback capacity is shown to be $C_{W; B^\infty}^{FB, CP}(\kappa) \tri \lim_{n \longrightarrow \infty} \frac{1}{n+1} C_{W; B^n}^{FB, CP}(\kappa)$.
%  Alajaji  \cite{alajaji1995} investigated the finite alphabet  version of \cite{cover-pombra1989}, defined by $B_i=A_i \oplus Z_i, i=0, \ldots,n,$ with alphabet spaces spaces ${\mathbb A}_i={\mathbb Z}_i=\big\{0,1,\ldots, q\big\}, i=0, \ldots, n$, without transmission cost,  and showed that  feedback does not increase the generalized capacity of such nonstationary non-ergodic channels \cite{verdu-han1994}, and that the maximizing channel input distribution is uniform.  The methodology adopted in \cite{alajaji1995} appears difficult to extend to channels with different noise and channel input alphabet spaces, and or  channels with transmission cost constraints, because it  is based on a converse to the coding theorem from  \cite{verdu-han1994}, which replaces  the distribution of $A^n$ by a uniform distribution.
Based on the characterization derived in \cite{cover-pombra1989}, several investigations of versions of the  Cover and Pombra \cite{cover-pombra1989} AGN channel are found  in the literature,  such as, \cite{ihara1993,yang-kavcic-tatikonda2007,kim2010}. Specifically,  in \cite{kim2010}, the  stationary ergodic version of Cover and Pombra \cite{cover-pombra1989} AGN channel, is revisited by utilizing characterization (\ref{cp1989}) to derive expressions for feedback capacity, $C_{W; B^\infty}^{FB, CP}(\kappa)$, using frequency domain methods, when  the noise power spectral density corresponds to a stationary Gaussian autoregressive moving-average model with finite memory. 
% of order $K$, then the optimal channel input distribution is  also of order $K$. 
% The methodology adopted in \cite{kim2010} is based  on power spectral densities, and hence it does  not extend  beyond stationary AGN   channels.\\  
For finite alphabet channels with memory and feedback,  expressions of feedback capacity are derived for certain  channels with symmetry,  in  \cite{permuter-cuff-roy-weissman2010,elishco-permuter2014,permuter-asnani-weissman2013,kourtellaris-charalambous2015,kourtellaris-charalambous-boutros:2015}, while in   \cite{yang-kavcic-tatikonda2005} it is illustrated   that if the input to the channel and the channel state are related by  a one-to-one mapping,  and the channel distribution is  $\big\{{\bf P}_{B_i|A_{i}, A_{i-1}}: i=0, \ldots, n\big\}$,  then dynamic programming can be used, in such optimization problems. In \cite{tatikonda-mitter2009}, the general concepts of dynamic programming are related to the computation of feedback capacity for Markov Channels (Definition~6.1 in  \cite{tatikonda-mitter2009}). In   \cite{chen-berger2005} the unit memory channel output (UMCO) channel  $\big\{{\bf P}_{B_i|B_{i-1}, A_i}: i=0, \ldots, n\}$, is analyzed under the assumption that the optimal channel input distribution is $\big\{{\bf P}_{A_i|B_{i-1}}: i=0, \ldots, n\}$. 
%The authors in \cite{chen-berger2005} showed that the UMCO channel can be transformed to one with state information.\\
% and that under certain conditions on the channel and channel input distributions, dynamic programming can be used to compute feedback capacity. 
%Coding theorems for  channels with memory with and without feedback are developed in many papers and books   \cite{dobrushin1959,pinsker1964,gallager1968,blahut1987,cover-thomas2006,ihara1993,verdu-han1994,han2003,kramer2003,tatikonda-mitter2009,permuter-weissman-goldsmith2009,gamal-kim2011}.
%, in three direction,   specifically,  for  jointly stationary ergodic  processes, for information stable processes, and for arbitrary nonstationary and nonergodic processes. 

\subsection{Channel Models and Transmission Cost Functions: Motivation and Objectives} 
 In general, it is almost impossible,  to determine the information structures of optimal channel input distributions directly from  $C_{A^\infty\rar  B^\infty}^{FB}$ and $C_{A^\infty\rar  B^\infty}^{FB}(\kappa)$. 
Indeed,  in the related theory of infinite horizon   Markov Decision (MD),  the fundamental question, whether optimizing the expected value of a fixed  pay-off functional over all non-Markov strategies occurs in the subclass of Markov strategies, is addressed from its finite horizon version. Then by using the Markovian property of strategies, the   infinite horizon or per unit time limit (i.e., asymptotic limit) over Markov strategies is analyzed
\cite{kumar-varayia1986}.  \\
%For specific MD models the analysis of the asymptotic limit, reveals several hidden properties of the role of optimal strategies to affect the controlled process.\\
However, classical stochastic optimal control or MD theory, is not directly applicable to  extremum problems of directed information, such as, (\ref{intro_fbc1a}), because the pay-off functional  is the directed information density,
\begin{align}
\iota_{A^n \rar B^n}(A^n, B^n) \tri \sum_{i=0}^n\frac{ d{\bf P}_{B_i|B^{i-1}, A^i}(\cdot|B^{i-1}, A^i)}{d{\bf P}_{B_i|B^{i-1}}(\cdot|B^{i-1})}(B_i)
\end{align}
and this pay-off depends nonlinearly on the channel input conditional distribution $\{ {\bf P}_{A_i|A^{i-1}, B^{i-1}}: i=0, \ldots, n\}$ via the channel output conditional distribution $\{{\bf P}_{B_i|B^{i-1}}: i=0, \ldots, n\}$.  This means, for general  extremum problems of feedback capacity, the information structure of optimal channel input distribution needs to be identified, before any method can be applied to compute feedback capacity, such as, the identification of sufficient statistics and dynamic programming \cite{kumar-varayia1986,vanschuppen2010}.

In this paper, our main objective is to determine the information structures of optimal channel input distributions, by  characterizing  the subsets of channel input distributions ${\cal P}_{[0,n]}^{CI}\subseteq {\cal P}_{[0,n]}$ and  ${\cal P}_{[0,n]}^{CI}(\kappa)\subseteq {\cal P}_{[0,n]}(\kappa)$, which satisfy conditional independence, and give   tight upper bounds on directed information $I(A^n \rar B^n)$, which are achievable, called the  ``characterizations of Finite Transmission Feedback Information (FTFI) capacity''. \\
We derive  characterizations of FTFI capacity  for any  class of time-varying  channel distributions and transmission cost functions, of the following type.

{\bf Channel Distributions.}
\begin{align}
&\mbox{\bf Class A.} \hso {\bf P}_{B_i|B^{i-1}, A^{i}}(db_i|b^{i-1}, a^{i}) = {\bf P}_{B_i|B^{i-1}, A_{i-L}^i}(db_i|b^{i-1}, a_{i-L}^i), \hso  i=0, \ldots, n, \label{CD_C2} \\
&\mbox{\bf Class B.} \hso {\bf P}_{B_i|B^{i-1}, A^{i}}(db_i|b^{i-1}, a^{i}) 
   ={\bf P}_{B_i|B_{i-M}^{i-1}, A_i}(db_i|b_{i-M}^{i-1}, a^i),\hso
    i=0, \ldots,n, \label{CD_C4a}\\
 & \mbox{\bf Class C.} \hso {\bf P}_{B_i|B^{i-1}, A^{i}}(db_i|b^{i-1}, a^{i}) 
  ={\bf P}_{B_i|B_{i-M}^{i-1}, A_{i-L}^i}(db_i|b_{i-M}^{i-1}, a_{i-L}^i),
   \hso  i=0, \ldots,n. \label{CD_C5}
\end{align}
{\bf Transmission Cost Functions.}
\begin{align}   
&\mbox{\bf Class A.}  \hso \gamma_i(T^ia^n, T^ib^{n}) ={\gamma}_i^{A.N}(a_{i-N}^i, b^{i}), \hst i=0, \ldots, n, \label{TC_1} \\
&\mbox{\bf Class B.}  \hso \gamma_i(T^ia^n, T^ib^{n}) = {\gamma}_i^{B.K}(a^i, b_{i-K}^{i}), \hst i=0, \ldots, n, \label{TC_2} \\
&\mbox{\bf Class C.}  \hso \gamma_i(T^ia^n, T^ib^{n}) = {\gamma}_i^{C.N,K}(a_{i-N}^i, b_{i-K}^{i}), \hst i=0, \ldots. n, \label{TC_3}
\end{align}
Here, $ \{K, L, M, N\}$  are nonnegative finite integers and we use the following convention. 
\begin{align}
\mbox{If $M=0$ then}\hso   {\bf P}_{B_i|B_{i-M}^{i-1}, \overline{A}^i}(db_i|b_{i-M}^{i-1}, \overline{a}^i)|_{M=0} = {\bf P}_{B_i| \overline{A}^i}(db_i| \overline{a}^i), \hso \mbox{for any} \hso \overline{A}^i \in \{A^i, A_{i-L}^i\}, \hso i=0,1, \ldots, n. \nonumber
%&\mbox{If $K=0$ then} \hso  {\gamma}_i^{C.N,K}(a_{i-N}^i, b_{i-K}^{i-1})\Big|_{K=0} \equiv {\gamma}_i^{C.N, 0}(a_{i-N}^i), \hso i=0, \ldots, n. \nonumber 
\end{align}
For $M=L=0$,  the channel  is memoryless.  By invoking function restriction, if necessary, the above transmission cost functions  include, as degenerate cases, many others, such as, $\gamma_i(\cdot, T^ib^{n}) = {\gamma}_i(\cdot, b_{i-K}^{i-1}), i=0, \ldots, n$. In this paper we do not treat the case $L=N=0$,  because these are investigated in \cite{kourtellaris-charalambousIT2015_Part_1}.  However, we provide discussions on the fundamental differences of the information structures of optimal channel input distributions, when the channels and transmission cost functions depend on past channel inputs, compared to $L=N=0$.\\ 
Channel distributions of  Class A, B or C, i.e.,  (\ref{CD_C2})-(\ref{CD_C5}), are induced  by various nonlinear channel models (NCM) driven by  noise processes \cite{caines1988}. 
%However, without proper re-formulation, these include NCM driven by independent noise processes.   
% and include  nonlinear and linear time-varying  Autoregressive models,  nonlinear and linear  channel models expressed in state space form .

We also derive characterizations of FTFI capacity  for any channel distribution induced by recursive Nonlinear Channel Models (NCM) driven by arbitrary distributed noise process $\{V_i: i=0, \ldots, n\}$ with memory and arbitrary alphabet spaces $\{{\mb V}_i: i=0, \ldots, n\}$, of the following type.

{\bf Nonlinear Channel Models with Correlated Noise.} 
\begin{align}
& B_i   = h_i^D(\overline{B}^{i-1}, \overline{A}^i, V_i) \hst   \mbox{for any} \hso \overline{A}^i \in \{A^i, A_{i-L}^i\}, \hso \overline{B}^{i-1} \in \{B^{i-1}, B_{i-M}^{i-1}\},   \hso i=0, \ldots, n,  \label{NCM-C_D_C_IN} \\
 &\frac{1}{n+1}{\bf E}_\nu \Big\{ \sum_{i=0}^n  \gamma_i(T^iA^n, T^iB^{n}) \Big\}  \leq \kappa, \label{NCM-A.D-IC_C_IN}  \\ %hst I \tri \max\{L, N\}, \: J\tri \max\{M, L, N, K\}, 
&{\bf P}_{V_i|V^{i-1}, A^i}(dv_i|v^{i-1}, a^i)=  {\bf P}_{V_i|\overline{V}^{i-1}}(dv_i|\overline{v}^{i-1}), \hso \overline{v}^{i-1} \in \big\{v_{i-T}^{i-1}, v^{i-1}\big\},\hso i=0, \ldots, n \label{CI_M_IN}
\end{align} 
where $\{h_i^D(\cdot, \cdot, \cdot): i=0, \ldots, n\}$ are nonlinear mappings and  $B_{-M}^{-1}=b_{-M}^{-1}, A_{-L}^{-1}=a_{-L}^{-1}$ are the initial data.\\
% and  the noise process distribution   satisfies conditional independence 
%\bea
%{\bf P}_{V_i|V^{i-1}, A^i}(dv_i|v^{i-1}, a^i)=  {\bf P}_{V_i|\overline{V}^{i-1}}(dv_i|\overline{v}^{i-1}), \hso \overline{v}^{i-1} \in \big\{v_{i-T}^{i-1}, v^{i-1}\big\},\hso i=0, \ldots, n. \label{CI_M_IN}
%\eea
Specifically, we show that we can apply the main theorems of the characterizations of FTFI capacity for Class A, B, C channels and tranmsission cost functions, with slight modification, to derive the characterizations of FTFI capacity for NCMs with correlated noise.  

The channel distributions of Class A, B, C and the NCMs include  nonlinear and  linear time-varying  autoregressive models and linear  channel models expressed in state space form \cite{caines1988}. Our main theorems generalize  many existing results  found in the literature, for example, non-stationary and non-ergodic  Additive Gaussian Noise channels investigated by  Cover and Pombra \cite{cover-pombra1989} and  stationary deterministic channels  \cite{kim2008},  and finite alphabet channels with channel state information investigated in \cite{chen-berger2005,yang-kavcic-tatikonda2005,permuter-cuff-roy-weissman2010,permuter-asnani-weissman2013,elishco-permuter2014,kourtellaris-charalambous2015,kourtellaris-charalambous-boutros:2015}. 
However, the derivations of characterizations of FTFI capacity and realizations of optimal channel input distributions by random processes are fundamentally different from any of the above references.

\subsection{Methodology \& Main Results}
\label{meth}
The methodology we apply to derive the information structures of optimal channel input distributions and the corresponding characterizations of FTFI capacity, combines  stochastic optimal control theory \cite{hernandezlerma-lasserre1996} and variational equalities of directed information \cite{charalambous-stavrou2013aa}.  
This   method is applied in \cite{kourtellaris-charalambousIT2015_Part_1} to derive  characterizations of FTFI capacity for channel distributions of Class A and C, with $L=0$, with and without transmission cost functions of Class A or C, with $N=0$.  \\
In this paper, we apply the method with some variations,  to any combination of channel distributions and transmission cost functions of class A, B, C, and to NCMs with correlated noise, as follows. \\

\noi{\bf Class A, B, C Channel Distributions and Transmission Cost Functions.}\\
First, we identify the connection between stochastic optimal control theory and  extremum problems    $C_{A^n \rar B^n}^{FB}, C_{A^n \rar B^n}^{FB}(\kappa)$ (see also Figure~\ref{figurefig3}), as follows.

\begin{description}
\item[(i)] The information measure $I(A^n \rar B^n)$ is the pay-off; 

\item[(ii)]  the channel output process $\{B_i: i=0, 1, \ldots, n\}$ is the controlled process;

\item[(iii)] the channel input process $\{A_i: i=0,1, \ldots, n\}$ is the control process;

\item[(iv)]  the channel output process $\{B_i: i=0, 1, \ldots, n\}$   is controlled, by controlling the conditional  probability distribution  $\big\{ {\bf P}_{B_i|B^{i-1}}: i=0, \ldots, n\big\}$, via the choice of the   transition probability distribution $\big\{ {\bf P}_{A_i| A^{i-1}, B^{i-1}}: i=0, \ldots, n\big\}\in {\cal P}_{[0,n]}$ or ${\cal P}_{[0,n]}(\kappa)$ called the control object.
\end{description}
Second, we identify  variational equalities of directed information, which can be used to determine achievable upper bounds on directed information over  subsets of channel input conditional distributions, ${\cal P}_{[0,n]}^{CI}(\kappa)  \subseteq {\cal P}_{[0,n]}(\kappa)$, characterized by conditional independence.   
\begin{figure}
  \centering
    \includegraphics[width=0.75\textwidth]{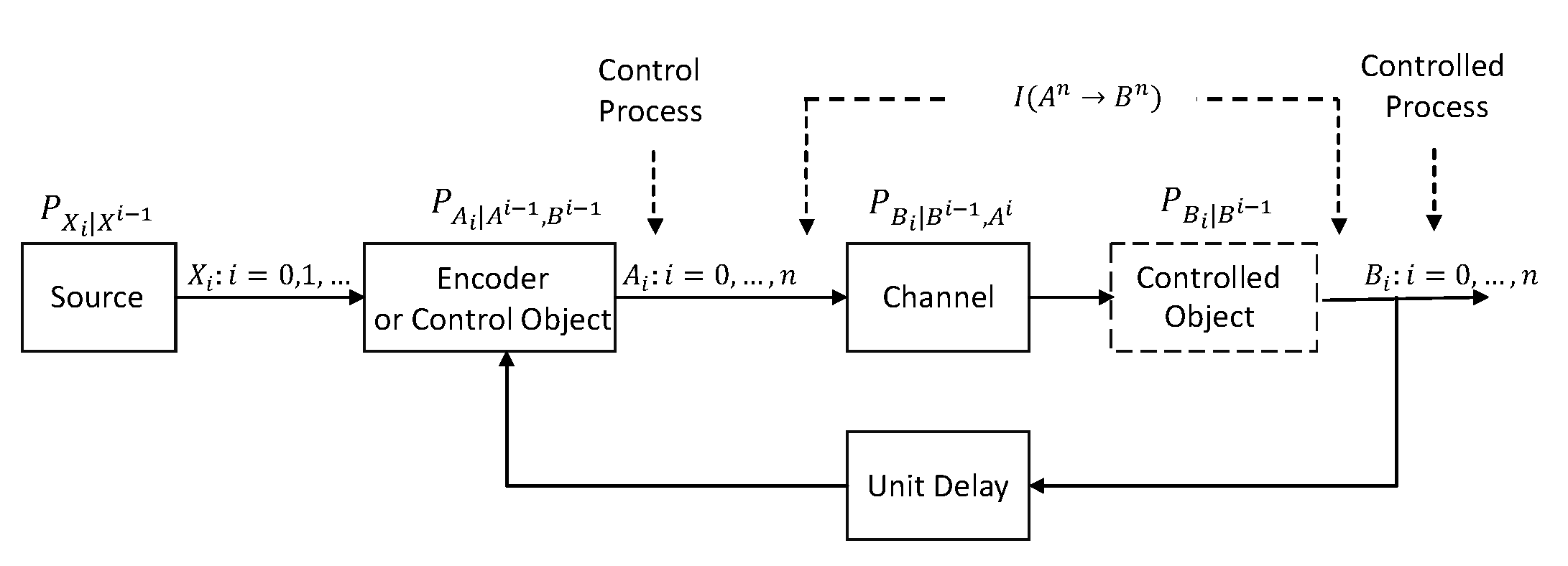}
    \label{figurefig3}
      \caption{Communication block diagram and its analogy to stochastic optimal control.}
\end{figure}
%----------------------------------------------------------
%-----------     E N D    F I G U R E ---------------------
%----------------------------------------------------------
%We apply the  analogy to stochastic optimal control theory and the variational equalities of directed information, to 
We show that for  any combination of channel distributions and transmission cost functions  of class A, B, or C, the  maximization of  $I(A^n \rar B^n)$ over ${\cal P}_{[0,n]}(\kappa)$, occurs in a subset  ${\cal P}_{[0,n]}^{CI}(\kappa)$,  which  satisfy  conditional independence, as follows.  
\begin{align}
&{\cal P}_{[0,n]}^{CI}(\kappa)  \tri \Big\{  {\bf P}_{A_i|A^{i-1}, B^{i-1}}(da_i|a^{i-1}, b^{i-1})={\pi }_i(da_i| {\cal  I}_i^{\bf P}) \equiv     {\mb P}\big\{A_i \in da_i| {\cal  I}_i^{\bf P}\big\}: i=0, \ldots, n\Big\} \bigcap {\cal P}_{[0,n]}(\kappa), \label{IS_int_1}\\
& {\cal  I}_i^{\bf P} \subseteq \big\{a^{i-1}, b^{i-1}\big\}, \hso i=0, \ldots, n, \\
&{\cal  I}_i^{\bf P} \tri \:  \mbox{Information Structure of optimal channel input distributions for $i=0,\ldots,n$}. \label{IS_int_3}
\end{align}
Further, we show that the information structure ${\cal I}_i^{{\bf P}}, i=0,1, \ldots, n$,  is specified by the memory of the channel conditional distribution, and the dependence of the transmission cost function on the channel input and output symbols. This procedure allows us to determine the dependence, of the joint distribution of $\{(A_i, B_i): i=0, \ldots, n\}$, and  the conditional  distribution $\{ {\bf P}_{B_i|B^{i-1}}: i=0, \ldots, n\}$  on  the control object, $\big\{{\pi }_i(da_i| {\cal  I}_i^{\bf P}): i=0, \ldots, \big\}$, and to determine the   characterizations of FTFI capacity. 

\noi {\bf NCMs with Correlated Noise.} \\
For any NCM defined by (\ref{NCM-C_D_C_IN})-(\ref{CI_M_IN}), with limited memory, i.e., $\big\{\overline{A}^i, \overline{B}^{i-1}, \overline{V}^i\big\}=\big\{A_{i-L}^i, B_{i-M}^{i-1}, V_{i-T}^i\big\}, \gamma_i(\cdot, \cdot)=\gamma_i^{C.L, M}(\cdot, \cdot), i=0, \ldots, n$, and under the assumption that the functions  mappings for fixed $(b_{i-M}^{i-1}, a_{i-L}^i)$ defined by 
\bea
h_i^D(b_{i-M}^{i-1}, a_{i-L}^i, \cdot): {\mathbb V}_i \longmapsto h_i^D(b_{i-M}^{i-1}, a_{i-L}^i, v_i), \hso i=0, \ldots, n \label{IV_IN}
\eea
 are invertible and measurable, with inverse $\overline{h}_i^D(b_i,b_{i-M}^{i-1}, a_{i-L}^i), i=0, \ldots, n$,  we first apply the converse to the coding theorem to derive the tight upper bound 
\begin{align} 
 R \leq   \liminf_{n \longrightarrow \infty}\frac{1}{n+1} {C}_{W ; B^n}^{FB,D}(\kappa) 
\end{align}
where 
\begin{align}
&{C}_{W ; B^n}^{FB,D}(\kappa)\tri  \sup_{{\cal P}_{[0,n]}^{D}(\kappa)}\sum_{i=0}^n  I(A_{i-L}^i, V_{i-T}^{i-1}; B_i|B^{i-1}), \label{CCIS_8_BC_D_n_INTRO}\\
&{\cal P}_{[0,n]}^{D}(\kappa) \tri \Big\{{\bf P}_{A_i| A^{i-1}, V^{i-1},  B^{i-1}}, i=0,\ldots,n:  \frac{1}{n+1} {\bf E}_\nu\Big( \sum_{i=0}^n \gamma_i^{C.L, M}(A_{i-L}^i, B_{i-M}^i) \leq \kappa  \Big)  \Big\}. 
\end{align} 
That is, ${C}_{W ; B^n}^{FB,D}(\kappa)$ is the analog of ${C}_{A^n \rar  B^n}^{FB}(\kappa)$. Then we show that  the methodology described above for Class A, B, C channels and transmission costs, with slight variations, is directly applicable, and we derive characterizations of the FTFI capacity, by  showing that the maximization in (\ref{CCIS_8_BC_D_n_INTRO}), occurs in subsets of  ${\cal P}_{[0,n]}^{D}(\kappa)$, which satisfy conditional independence.  

We emphasize that our objective and methodology  descibed above,  are fundamentally different from  any derivations given in  the literature, such as, 
%the  derivation proposed  for Lemma~2 in \cite{berger_shannon_lecture} corresponding to the unit memory channel output channel, 
  Theorem~1 in   \cite{yang-kavcic-tatikonda2005},   Theorem~1   in   \cite{yang-kavcic-tatikonda2007},    Theorem~1 in  \cite{permuter-cuff-roy-weissman2010},  and further adopted in subsequent work in   \cite{elishco-permuter2014,permuter-asnani-weissman2013}. Specifically, we show that the supremum of directed information over all channel input conditional distributions occurs in a smaller set, satisfying a conditional independence condition, which is  analogous to  (\ref{IS_int_1}). This is different from the derivations given  in  \cite{yang-kavcic-tatikonda2005,yang-kavcic-tatikonda2007,permuter-cuff-roy-weissman2010}. This point is further elaborated in  Section~\ref{randomized}.

In Section~\ref{VE_SE}, we introduce the notation and the variational equalities of directed information.\\
In Section~\ref{randomized}, we derive the information structures of optimal channel input distributions for any combination of channel distributions and transmission cost functions of Class A, B or C. \\
In section~\ref{exa_gen},  we consider the  application example of general Multiple-Input Multiple-Output (MIMO) Gaussian channels with memory on past channel input and output symbols, and quadratic cost constraint, i.e., class C, with $L=M=1$.  We show that the optimal channel input distribution corresponding to the characterization of the  FTFI capacity exhibits a  separation principle. We show this separation principle by using the orthogonal decomposition of  realizations of optimal channel input distributions.   \\
Via the separation principle, we derive an expression for the optimal channel input distribution, and we  relate the characterization of FTFI capacity to the so-called Linear-Quadratic-Gaussian partially observable stochastic optimal control problem \cite{kumar-varayia1986}.\\
In Section~\ref{DTRM}, we first derive a converse to the coding theorem for NCMs defined by  (\ref{NCM-C_D_C_IN})-(\ref{CI_M_IN}) and (\ref{IV_IN}) and then  we derive analogous information structures of optimal channel input distributions and corresponding characterizations of FTFI capacity. \\
Throughout the paper we relate the characterizations of FTFI capacity of various channels and the realizations of optimal channel input distributions 
to existing results given in the literature.

\section{Extremum problems of Directed Information and Variational Equalities}
\label{VE_SE}
In this section, we introduce  the basic notation,  the  precise definition of extremum problem of FTFI capacity (\ref{cap_fb_1}), the variational equalities of directed information \cite{charalambous-stavrou2012}, and some of their properties.

% and a basic lemma derived in \cite{gihman-skorohod1979} is introduced, which gives conditions for any conditional distribution to be realized  by deterministic measuurable functions driven by uniform random variables.

Throughout the paper we use the following notation.
\begin{align}
& {\mathbb R}: \hso  \mbox{set of  real numbers};  \nonumber \\
& {\mathbb Z}: \hso  \mbox{set of  integer};  \nonumber \\
& {\mathbb N}_0: \hso  \mbox{set of nonnegative integers} \hso \{0, 1,2,\dots\}; \nonumber \\
%& {\mathbb N}^n: \hso  \mbox{set of first $n+1$ %natural numbers} \hso \{0, 1,2,\dots, n\}; \nonumber \\
& {\mathbb R}^n: \hso  \mbox{set of  $n$ tuples of real  natural}; \nonumber \\
& {\mb S}_+^{p \times p}: \hso \mbox{set of symmetric positive semi-define}\ p\times p  \hso \mbox{matrices}\ A \in {\mathbb R}^{p \times p}; \nonumber   \\
&{\mb S}_{++}^{p \times p}: \hso \mbox{subset of positive definite matrices of the set}\hso  {\mb S}_+^{p \times p};  \nonumber \\
&\langle \cdot, \cdot \rangle: \hso \mbox{inner product of elements of vectors spaces;} \nonumber  \\
&(\Omega, {\cal F}, {\mathbb P}): \mbox{probability space, where ${\cal F}$ is the $\sigma-$algebra generated by subsets of $\Omega$}; \nonumber \\ 
& {\cal  B}({\mathbb  W}): \hso \mbox{Borel $\sigma-$algebra of a given topological space  ${\mathbb W}$};  \nonumber \\
&{\cal M}({\mathbb W}): \hso \mbox{set of all probability measures on ${\cal  B}({\mathbb W})$ of a Borel space ${\mathbb W}$}; \nonumber\\
&{\cal K}({\mathbb V}|{\mathbb W}): \hso \mbox{set of all stochastic kernels on  $({\mathbb V}, {\cal  B}({\mathbb V}))$ given $({\mathbb W}, {\cal  B}({\mathbb W}))$ of Borel spaces ${\mathbb W}, {\mathbb V}$}; \nonumber
%\\
%&X \leftrightarrow Y \leftrightarrow Z: \hso \mbox{Conditional independence of RVs $(X, Z)$ given RV $Y$}. \nonumber
\end{align}
All spaces (unless stated otherwise) are complete separable metric spaces also called  Polish spaces, i.e., Borel spaces. This generalization is adopted   to treat simultaneously discrete, finite alphabet,  real-valued ${\mathbb R}^k$ or complex-valued ${\mathbb C}^k$ random processes for any positive integer $k$, and general ${\mathbb R}^k-$valued  random processes with absolute summable $p$-moments, 
%characterized by $\ell^p({\mathbb N} \times \Omega, {\cal F}, {\mathbb P}; {\mathbb R}^n)$-spaces, 
$p =1, 2, \ldots$, (see \cite{dunford1988}) etc. 

%\subsection{Basic Notions of Probability}
%\label{subsec-prob}
Given two measurable spaces $({\mb X}, {\cal  B}({\mb X}))$,  $({\mb Y}, {\cal  B}({\mb Y}))$ then ${\mb X} \times {\mb Y} \tri \{(x,y):  x\in {\mb X}, y \in {\mb Y}\}$ is the cartesian product of ${\mb X}$ and ${\mb Y}$, and for $A \in  {\cal  B}({\mb X})$ and $B \in {\cal  B}({\mb Y})$ then $A \times B$ is called a measurable rectangle. The product measurable space of $({\mb X}, {\cal  B}({\mb X}))$ and $({\mb Y}, {\cal  B}({\mb Y}))$ is denoted by $({\mb X} \times {\mb Y}, {\cal  B}({\mb X})\otimes  {\cal  B}({\mb Y}))$, where  ${\cal  B}({\mb X})\otimes   {\cal  B}({\mb Y})$ is the product $\sigma-$algebra generated by $\{A \times B:  A \in {\cal  B}({\mb X}), B\in  {\cal  B}({\mb Y})\}$.  \\
A Random Variable (RV)  defined on a probability space $(\Omega, {\cal F}, {\mathbb P})$ by the mapping $X: (\Omega, {\cal F}) \longmapsto ({\mb X}, {\cal  B}({\mb X}))$  induces a probability distribution $ {\bf P}(\cdot) \equiv {\bf P}_X(\cdot)$ on  $({\mb X}, {\cal  B}({\mb X}))$ as follows\footnote{The subscript $X$ is often omitted.}.
\begin{align}
{\bf P}(A) \equiv  {\bf P}_X(A)  \tri {\mathbb P}\big\{ \omega \in \Omega: X(\omega)  \in A\big\},  \hso  \forall A \in {\cal  B}({\mb X}).
 \end{align}
A RV is called discrete if there exists a countable set ${\cal S}_X\tri \{x_i: i \in {\mathbb N}_0\}$ such that $\sum_{x_i \in {\cal S}_X} {\mathbb  P} \{ \omega \in \Omega : X(\omega)=x_i\}=1$. The probability distribution ${\bf P}_X(\cdot)$  is then concentrated on  points in ${\cal S}_X$, and it is defind by 
\bea
 {\bf P}_X(A)  \tri \sum_{x_i \in {\cal S}_X \bigcap A} {\mathbb P} \big\{ \omega \in \Omega : X(\omega)=x_i\big\}, \hso \forall A \in {\cal  B}({\mb X}). 
\eea 
%{\bf put material on integral and conditional distribution, like K-V page 14-18}
If the cardinality of ${\cal S}_X$ is finite then the RV is finite-vaued  and it is called a finite alphabet RV. \\
Given another RV $Y: (\Omega, {\cal F}) \longmapsto ({\mb Y}, {\cal  B}({\mb Y}))$,   ${\bf P}_{Y|X}(dy| X)(\omega)$ is called the conditional distribution of RV $Y$ given RV $X$. The conditional distribution of RV $Y$ given $X=x$ is denoted by ${\bf P}_{Y|X}(dy| X=x)  \equiv {\bf P}_{Y|X}(dy|x)$. Such conditional distributions are  equivalently  described   by  stochastic kernels or transition functions ${\bf K}(\cdot|\cdot)$ on $ {\cal  B}({\mb Y}) \times {\mathbb X}$, mapping ${\mb X}$ into ${\cal M}({\mathbb Y})$ (the space of probability measures on $({\mathbb Y}, ({\cal B}({\mathbb Y}))$, i.e., $x \in {\mathbb X}\longmapsto {\bf K}(\cdot|x)\in{\cal M}({\mathbb Y})$,  such that for every $F \in {\cal  B}({\mb Y})$, the function ${\bf K}(F|\cdot)$ is ${\cal  B}({\mb X})$-measurable.\\
The family of such probability distributions on $({\mb Y}, {\cal B}({\mb Y}))$ parametrized by $x \in {\mb X}$, is defined by ${\cal K}({\mb Y}| {\mb X})\tri \Big\{{\bf K}(\cdot| x) \in {\cal M}({\mathbb Y}): \hso x \in {\mathbb X}\Big\}$.

 \subsection{FTFI Capacity and Variational Equalities}
\label{def-sub2}

The  communication block diagram is shown in  Figure~\ref{figurefig3}.
The channel input and  channel output alphabets are  sequences of Polish measurable spaces (complete separable metric spaces) $\{({\mb A}_i,{\cal  B}({\mb A }_i)):i\in\mathbb{Z}\}$ and  $\{({\mb  B}_i,{\cal  B}({\mb  B}_i)):i\in\mathbb{Z}\}$, respectively, and 
their history spaces  are the product spaces ${\mb A}^{\mathbb{Z}}\tri {{\times}_{i\in\mathbb{Z}}}{\mb A}_i,$ ${\mb  B}^{\mathbb{Z}}\tri {\times_{i\in\mathbb{Z}}}{\mb  B}_i$.  These spaces are endowed with their respective product topologies, and  ${\cal  B}({\Sigma}^{\mathbb{Z}})\tri \otimes_{i\in\mathbb{Z}}{\cal  B}({\Sigma }_i)$,  denotes the $\sigma-$algebra on ${\Sigma }^{\mathbb{Z}}$, where ${\Sigma}_i \in  \big\{{\mb A}_i, {\mb  B}_i\big\}$,  ${\Sigma}^{\mathbb{Z}} \in  \big\{{\mb A}^{\mathbb Z}, {\mb  B}^{\mathbb Z}\big\}$,  generated by cylinder sets. Points in ${\Sigma }_k^m \tri \times_{j=k}^m {\Sigma}_j$ are denoted by $z_{k}^m \tri \{z_k, z_{k+1}, \ldots, z_m\} \in {\Sigma}_k^m$,   $(k, m)\in   {\mathbb Z} \times {\mathbb Z}$.
% We often restrict  ${\mb Z}$ to  ${\mb N}_0$. 

Next, we introduce the various distributions.

\noi{\bf Channel Distributions with Memory.}  A sequence of stochastic kernels or distributions defined by 
\begin{align}
{\cal C}_{[0,n]} \tri \Big\{Q_i(db_i|b^{i-1},a^{i})= {\bf P}_{B_i|B^{i-1}, A^i}  \in {\cal K}({\mb  B}_i| {\mb  B}^{i-1} \times {\mb A}^i) :  i=0,1, \ldots, n \Big\}. \label{channel1}
\end{align}
 At each time instant $i$ the conditional distribution of the channel is affected causally by  past channel output symbols $b^{i-1} \in {\mb B}^{i-1}$ and current and past channel input symbols $a^{i} \in {\mb A}^i, i=0,1, \ldots, n$. 
% The distribution at time $t=0$ is either fixed or the conditioning information is fixed, depending to the convention used. 

\noi{\bf Channel Input Distributions with Feedback.}  A  sequence of stochastic kernels defined by 
\bea
{\cal P}_{[0,n]} \tri  \Big\{  P_i(da_i|a^{i-1},b^{i-1})={\bf P}_{A_i|A^{i-1}, B^{i-1}}  \in  {\cal K}({\mb A}_i| {\mb A}^{i-1} \times {\mb  B}^{i-1}):   i=0,1, \ldots, n \Big\}. \label{rancodedF}
\eea
At each time instant $i$ the conditional channel input distribution with feedback is affected causally by past  channel inputs and  output symbols  $\{a^{i-1}, b^{i-1}\} \in {\mb A}^{i-1} \times {\mb B}^{i-1}, i=0,1, \ldots, n$.

{\bf Admissible Histories.} For each $i=-1, 0, \ldots, n$, we introduce the space ${\mb G}^{i}$ of admissible histories of channel input and output symbols,   as follows. Define
\begin{IEEEeqnarray}{rCl}
 {\mb  G}^i\triangleq {\mb A}^{-1} \times  {\mb B}^{-1} \times  \mathbb{A}_0\times \mathbb{B}_0\times \hdots \times \mathbb{A}_{i-1}\times\mathbb{B}_{i-1}\times  \mathbb{A}_i\times {\mb B}_i,
 \;  i=0, \ldots,n,  \; {\mb G}^{-1}= {\mb A}^{-1} \times {\mb B}^{-1}.
\end{IEEEeqnarray}
A typical element of ${\mb G}^i$ is a sequence of the form $(a^{-1}, b^{-1},a_0,b_0,\hdots, a_{i},b_i)$.  We equip the space 
%${\mb G}^i$
  ${\mb G}^i$  with the natural $\sigma$-algebra 
%  ${\cal B}({\mb G}^i)$ and 
  ${\cal B}({\mb G}^i)$, for $i=-1,0,\ldots, n$. Hence, for each $i$, the information structure of the channel input distribution  is
\begin{align}
{\cal I}_i^P\tri \Big\{A^{-1}, B^{-1}, A_0, B_0, \ldots, A_{i-1}, B_{i-1}\Big\},  \; i=0,1, \ldots, n, \hst {\cal I}_0^P \tri  \big\{A^{-1}, B^{-1}\big\}.
\end{align}
This implies at time $i=0$, the initial distribution is $P_0(da_0|a^{-1}, b^{-1})=P_0(da_0| {\cal I}_0^P)=P_0(da_0|a^{-1}, b^{-1})$. However, we can modify  ${\cal I}_0^P$ to  consider an alternative convention such as ${\cal I}_0^P =\{\emptyset\}$ or ${\cal I}_0^P=\{b^{-1}\}$, etc..

{\bf  Joint and Marginal Distributions.}  Given any channel input distribution $\big\{{ P}_i(da_i|a^{i-1}, b^{i-1}): i=0,1, \ldots, n\big\} \in {\cal P}_{[0,n]}$,  the channel distribution $\big\{Q(db_i| b^{i-1}, a^{i-1}): i=0,1, \ldots, n\big\}\in {\cal C}_{[0,n]}$,  and the initial probability distribution ${\bf P}(da^{-1},db^{-1})\equiv \nu(da^{-1},db^{-1})\in {\cal M}({\mb G}^{-1})$, then we can  uniquely define the  induced joint distribution  ${\bf P}_\nu^{P}(da^n, db^n)$ on  the  canonical space $\Big({\mb  G}^n, {\cal  B}({\mb G}^n)\Big)$, and we can construct a probability space $\Big(\Omega, {\cal F}, {\mathbb P}\Big)$ carrying the sequence of RVs $\{(A_i, B_i): i=\ldots, -1, 0, \ldots, n\}$, as follows.
\begin{align}
{\bf P}_\nu^P(da^n, db^n)\equiv & {\bf P}^P_{\nu}( da^{-1}, db^{-1},da_0,db_0,\hdots,da_{n},db_n) \label{JOINT_1} \\
 =& \nu(da^{-1}, db^{-1})\otimes P_0(da_0|a^{-1},b^{-1})  \otimes Q_0(db_0|b^{-1},a^{-1},a_0) \otimes 
 P_1(da_1|a^{-1}, a_0, b^{-1},b_0) \nonumber \\
 &\otimes\hdots\otimes Q_{n-1}(db_{n-1}|b^{n-2},a^{n-1}) \otimes  P_{n}(da_{n}|b^{n-1},a^{n-1})\otimes Q_n(db_n|b^{n-1},a^{n})\label{upm1}\\
 \equiv & \nu(da^{-1}, db^{-1})   \otimes_{j=0}^n \Big(Q_j(db_j|b^{j-1}, a^j)\otimes P_j(da_j|a^{j-1}, b^{j-1})\Big) \label{CIS_2gg_new} 
\end{align}
such that for $j=0, \ldots, n$, 
\begin{align}
&{\mb P}\big\{(A^{-1}, B^{-1}) \in C\big\}= {\bf P}^P_{\nu}( C)=\nu(C), \hst   C\in{\cal B}({\mb G}^{-1})\\
&{\mb P}\big\{A_j \in D|A^{j-1}= a^{j-1}, B^{j-1}=b^{j-1}\} =  {\bf P}^P_{\nu}(D|a^{j-1}, b^{j-1})  =P_j(D|a^{j-1}, b^{j-1}), \hst  D\in{\cal B}({\mb A}_j)\\
&{\mb P}\big\{B_j \in E|B^{j-1}= b^{j-1}, A^{j}=a^{j}\} =  {\bf P}^P_{\nu}(E|b^{j-1}, a^{j})=Q_j(E|b^{j-1},a^j),  \hst  E\in {\cal B}({\mb B}_{j}).
\end{align}
%
%\begin{align}
% {\mathbb P}\big\{A^n \in d{a}^n, B^n \in d{b}^n\big\}  \tri\;  &
%{\bf P}^P(da^n, db^n), \hso n \in {\mathbb N} \nonumber \\
%=\; & \otimes_{j=0}^n \Big( { \bf P}(db_j|b^{j-1}, a^j) \otimes {\bf P}(da_j| a^{j-1}, b^{j-1})\Big) \label{CIS_2gde2new} \\
%=&\otimes_{j=0}^n \Big(Q_j(db_j|b^{j-1}, a^j)\otimes P_j(da_j|a^{j-1}, b^{j-1})\Big). \label{CIS_2gg_new} 
%\end{align}
Further, we define the joint distribution of $\big\{B^{-1}, B_0, \ldots, B_n\big\}$ and the conditional  probability  distribution of $B_i$ given $B^{i-1}$  by\footnote{Throughout the paper the superscript notation ${\bf P}^P(\cdot), \Pi_{0,n}^P(\cdot), etc., $ indicates the dependence of the distributions on the channel input conditional distribution.}
\begin{align}
{\mathbb  P}\big\{B^n \in db^n\big\} \tri \; & {\bf P}_\nu^{P}(db^n) =  \int_{{\mb A}^n}  {\bf P}_\nu^{ P}(da^n, db^n) ,   \hso  n \in {\mathbb N},  \label{CIS_3g}\\
\equiv \; & \nu(db^{-1})\otimes \overrightarrow{\Pi}_{0,n}^{P}(db^n), \hso  \overrightarrow{\Pi}_{0,n}^{P}(db^n)\tri  \otimes  \Pi_0^{ P}(db_0|b^{-1})\otimes \ldots \otimes \Pi_n^P(db_n|b^{n-1}) \label{MARGINAL} \\
\Pi_i^{ P}(db_i|b^{i-1})= \; &  \int_{{\mb A}^i} Q_i(db_i|b^{i-1}, a^i)\otimes P_i(da_i|a^{i-1}, b^{i-1}) \otimes {\bf P}^{P}(da^{i-1}|b^{i-1}), \hso  i=0, \ldots, n . \label{CIS_3a}
\end{align}
The above distributions are parametrized by the distribution ${\bf P}(da^{-1}, db^{-1})=\nu(da^{-1}, db^{-1})$ or ${\bf P}(db^{-1})=\nu(db^{-1})$.   We denote the expectation operator with respect to ${\bf P}_\nu^P(da^n, db^n)$ by ${\bf E}_\nu^P$. Moreover, if $\nu$ is concentrated at $(A^{-1},B^{-1})=(a^{-1},b^{-1})$ we write ${\bf P}_{a^{-1}, b^{-1}}^P$ and  ${\bf E}_{a^{-1}, b^{-1}}^P$; in this case, the above distributions are parametrized by $(A^{-1},B^{-1})=(a^{-1},b^{-1})$.
This notation is often omitted when it is clear from the context.

\noi{\bf Transmission Cost.}   The cost of transmitting and receiving  symbols  is a  measurable function $c_{0,n}:{\mb A}^n\times{\mb  B}^{n} \longmapsto [0,\infty)$. The average transmission cost  is defined by 
\begin{align}
\frac{1}{n+1} {\bf E}_\nu^P \Big\{ c_{0,n}(A^n, B^{n}) \Big\}\leq  \kappa, \; \hst  c_{0,n}(a^n, b^{n}) \tri \sum_{i=0}^n \gamma_i(T^ia^n, T^ib^{n}), \;      \kappa \in [0,\infty) \label{transmissioncost}
\end{align}
where the superscript notation ${\bf E}_\nu^{ P }\{\cdot\}$ denotes the dependence of the joint distribution on the choice of  conditional distribution $\{P_i(da_i|a^{i-1}, b^{i-1}) : i=0, \ldots, n\} \in {\cal P}_{[0,n]}$.
The set of  channel input distributions with feedback and transmission cost is defined by 
\begin{align}
{\cal P}_{[0,n]}(\kappa) \tri  \Big\{  P_i(da_i|a^{i-1}, b^{i-1}) \in {\cal M}({\mb A}_i),  i=0, \ldots, n: 
\frac{1}{n+1} {\bf E}_\nu^{ P} \Big( c_{0,n}(A^n, B^{n}) \Big)\leq  \kappa\Big\}  \subset {\cal P}_{[0,n]} . \label{rc1}
\end{align}

{\bf FTFI Capacity.} 
The pay-off or  directed  information  $I(A^n \rar B^n)$ is defined as follows. 
\begin{align}
I(A^n\rar B^n) \tri  &\sum_{i=0}^n {\bf E}_\nu^{{ P}} \Big\{  \log \Big( \frac{dQ_i(\cdot|B^{i-1},A^i) }{d\Pi_i^{{ P} }(\cdot|B^{i-1})}(B_i)\Big)\Big\}  \\
=& \sum_{i=0}^n \int_{{\mb G}^{i}  }^{}   \log \Big( \frac{ dQ_i(\cdot|b^{i-1}, a^i) }{d\Pi_i^{{ P}}(\cdot|b^{i-1})}(b_i)\Big) {\bf P}_\nu^{ P}( da^i, db^i)  
 \equiv   {\mathbb I}_{A^n\rightarrow B^n}^\nu({P}_i,{ Q}_i,  :~i=0,1,\ldots,n)\label{CIS_6} 
\end{align}
where the notation in the right hand side of (\ref{CIS_6}) illustrates that  $I(A^n \rar B^n)$ is  a functional of the two sequences of conditional  distributions, $\big\{{P}_i(da_i|a^{i-1}, b^{i-1}), { Q}_i(db_i|b^{i-1}, a^i): i=0,1, \ldots, n\big\}$ and the fixed distribution $\nu(\cdot)$.

Next, we introduce the definition of FTFI capacity $C_{A^n \rar B^n}^{FB}$,  for Class A, B, C channel distributions and transmission cost functions, using the above notation.  \\

\begin{definition} (Extremum problem with feedback)\\
\label{def-gsub2sc}
Given  any  channel  distribution from the class ${\cal C}_{[0,n]}$, and any initial distribution $(A^{-1},B^{-1}) \sim \nu(da^{-1}, db^{-1})\in {\cal M}({\mathbb A}^{-1} \times {\mathbb B}^{-1})$, 
find the {\it Information Structure } of the optimal channel input distribution $\big\{P_i(da_i|a^{i-1}, b^{i-1}): i=0, \ldots, n\big\}  \in  {\cal P}_{[0,n]} $ (assuming it exists) of the extremum problem defined by
\begin{align}
C_{A^n \rar B^n}^{FB} \tri \sup_{ \big\{P_i(da_i|a^{i-1}, b^{i-1}): i=0, \ldots,n \big\} \in  {\cal P}_{[0,n]} } I(A^n\rar B^n), \hst   I(A^n \rar B^n) =(\ref{CIS_6}). \label{prob2}
\end{align}
When an  transmission cost constraint is imposed the extremum problem is defined by 
 \begin{align}
 C_{A^n \rar B^n}^{FB}(\kappa) \tri \sup_{\big\{P_i(da_i|a^{i-1}, b^{i-1}): i=0,\ldots, n\big\} \in {\cal P}_{[0,n]}(\kappa) } I(A^n\rar B^n).  \label{prob2tc}
\end{align}
\end{definition}

Our first objective is to determine the information structures of optimal channel input distributions for any combination of channel distribution and transmission cost of class A, B, or C, as discussed by (\ref{IS_int_1})-(\ref{IS_int_3}).  Clearly, 
for each time $i$ the largest information structure of the channel input distributions of   problem $C_{A^n \rar B^n}^{FB}$ and  $C_{A^n \rar B^n}^{FB}(\kappa)$   is ${\cal I}_i^P \tri \{a^{i-1}, b^{i-1}\}, i=0,1, \ldots, n$.

{\bf Alternative Equivalent Representation of Directed Information.} 
Often, it is convenient to use  alternative equivalent representations of the sets ${\cal P}_{[0,n]}, {\cal C}_{[0,n]}$  and induced joint distribution, and marginal distribution, via the causally conditioned compound probability distributions, defined as follows. Introduce the distributions    $\overrightarrow{Q}_{0,n}(\cdot|a^{n})\in {\cal M}({\mb  B}_0^n)$ parametrized by $(a^n, b^{-1}) \in {\mb A}^n \times {\mb B}^{-1}$ and $\overleftarrow{P}_{0,n}(\cdot|b^{n-1})\in {\cal M}({\mb  A}_0^n)$ parametrized by $(a^{-1},b^{n-1}) \in {\mb A}^{-1} \times  {\mb B}^{n-1}$, and defined by   
\begin{align}
\overrightarrow{Q}_{0,n}(db^n|a^n)\tri \otimes_{i=0}^n Q_i(db_i|b^{i-1}, a^i), \hso \overleftarrow{P}_{0,n}(da^n|b^{n-1})\tri \otimes_{i=0}^n P_i(da_i|a^{i-1}, b^{i-1}).
\end{align}
For a fixed $(A^{-1}, B^{-1})=(a^{-1}, b^{-1})$ these compound distribution define uniquely  the following joint and marginal distributions.
\begin{align}
  {\bf P}^P_{a^{-1},b^{-1}}(da^n, db^n)= (\overleftarrow{P}_{0,n} \otimes \overrightarrow{Q}_{0,n})(da^n, db^{n}), \hso {\bf P}_{b^{-1}}^P(db^n) = \overrightarrow{\Pi}_{0,n}^{\overleftarrow{P}}(db^n)\tri \int_{{\mb A}_0^n} (\overleftarrow{P}_{0,n} \otimes \overrightarrow{Q}_{0,n})(da^n, db^{n}). \label{CIS_3a_new}
\end{align}
 It is shown in \cite{charalambous-stavrou2013aa}, that the set of distributions $\overrightarrow{Q}_{0,n}(\cdot|a^{n})\in {\cal M}({\mb  B}_0^n)$ and   $\overleftarrow{P}_{0,n}(\cdot|b^{n-1})\in {\cal M}({\mb  A}_0^n)$ are convex. Moreover,  given a fixed ${\bf P}(da^{-1},db^{-1})=\nu(da^{-1},db^{-1})$, directed  information  $I(A^n \rar B^n)$ is equivalently defined as follows. 
\begin{align}
I(A^n\rar B^n) \tri  &\sum_{i=0}^n {\bf E}_{\nu}^{{ P}} \Big\{  \log \Big( \frac{dQ_i(\cdot|B^{i-1},A^i) }{d\Pi_i^{{ P} }(\cdot|B^{i-1})}(B_i)\Big)\Big\}  \\
=&\int_{{\mb G}^{n}   }^{}   \log \Big( \frac{ d \overrightarrow{Q}_{0,n}(\cdot|a^i) }{d\overrightarrow{\Pi}_{0,n}^{ \overleftarrow{P}}(\cdot)}(b^n)\Big) (\overleftarrow{P}_{0,n} \otimes \overrightarrow{Q}_{0,n})(da^n, db^{n})\otimes \nu(da^{-1},db^{-1}) \\
 \equiv &  {\mathbb I}_{A^n\rightarrow B^n}^\nu(\overleftarrow{P}_{0,n}, \overrightarrow{ Q}_{0,n})\label{CIS_6_P}
\end{align}
where the notation in the right hand side of (\ref{CIS_6_P}) illustrates the   functional dependence  on $\{\overleftarrow{P}_{0,n}(da^n|b^{n-1})$, $\overrightarrow{ Q}_{0,n}(db^n|a^n)\}$ and  the fixed distribution $\nu(da^{-1},db^{-1})$. These are equivalent representations  \cite{charalambous-stavrou2013aa}.  \\
Further, it is shown in \cite{charalambous-stavrou2013aa}, that for a fixed $\nu(\cdot)$, the functional ${\mathbb I}_{A^n\rightarrow B^n}^\nu(\overleftarrow{P}_{0,n}, \overrightarrow{ Q}_{0,n})$ is convex in  $\overrightarrow{Q}_{0,n}(\cdot|a^{n})\in {\cal M}({\mb  B}_0^n)$ for a fixed   $\overleftarrow{P}_{0,n}(\cdot|b^{n-1})\in {\cal M}({\mb  A}_0^n)$ and concave in $\overleftarrow{P}_{0,n}(\cdot|b^{n-1})\in {\cal M}({\mb  A}_0^n)$ for a fixed $\overrightarrow{Q}_{0,n}(\cdot|a^{n})\in {\cal M}({\mb  B}_0^n)$. These convexity and concavity properties imply that any extremum problem of feedback capacity is a convex optimization problem over appropriate sets of channel input distributions.

\begin{comment}

In communication applications, if the joint process $\{(A_i, B_i): i=0, 1, \ldots, \}$ is stationary ergodic or directed information stability holds, then the per unit time limiting versions of Definition~\ref{def-gsub2sc},  that is, $\liminf_{n \longrightarrow \infty} \frac{1}{n+1}C_{A^n \rar B^n}^{FB}$, has an operational meaning in terms of the supremum of all achievable rates.  However,  in general,   the optimal channel input distribution of the per unit time limiting versions of Definition~\ref{def-gsub2sc} depends on the initial data $b^{-1}$, and the supremum of all achievable rates is given by .....For indecomposable channels .... 

\end{comment}

{\bf Variational Equalities of Directed Information.} Next, we introduce the two variational equalities of directed information, derived in \cite{charalambous-stavrou2013aa},  which we  employ in many of the derivations. \\

\begin{theorem}(Variational Equalities)\\
\label{thm-var} 
Given a channel input distribution $\big\{P_i(da_i|a^{i-1},b^{i-1}): i=0,1, \ldots, n\big\} \in {\cal P}_{[0,n]}$ and channel distribution  $\big\{Q_i(db_i|b^{i-1},a^i) : i=0,1, \ldots, n\big\} \in {\cal C}_{[0,n]} $,  define the  corresponding  joint and  marginal distributions  ${\bf P}_{a^{-1}, b^{-1}}^P(da^n,db^n)$ and $\big\{ {\Pi }_i^P(db_i|b^{i-1}): i=0, \ldots, n\big\}$  by (\ref{JOINT_1})-(\ref{CIS_3a}).

(a)  Let   $\big\{V_i(db_i|b^{i-1}) \in {\cal M}({\mb B}_i): i=0, \ldots, n\}$ be an arbitrary distribution.   
Then the following variational equality holds.
\begin{align}
I(A^n\rightarrow B^n) =\inf_{  \big\{ V_i(db_i| b^{i-1} )\in {\cal M}({\mb  B}_{i}): i=0,1, \ldots, n\big\}}\sum_{i=0}^n  \int_{{\mb G}^i}   \log \Big( \frac{dQ_i(\cdot|b^{i-1},a^i) }{ dV_i(\cdot|b^{i-1})}(b_i)\Big)   {\bf P}_\nu^P(da^i,db^i)     \label{BAM52a}
\end{align}
and the infimum in  (\ref{BAM52a}) is achieved at 
\bea
V_i(db_i|b^{i-1})= {\Pi}_i^P(db_i|b^{i-1}),\hso  i=0, \ldots, n \hso  \mbox{ given by (\ref{CIS_3a})}.
\eea
(b) Let $\big\{S_i(db_i|b^{i-1},a^{i-1}) \in {\cal M}({\mb B}_i) : i=0, \ldots, n\big\}$ and $\big\{R_i(da_i|a^{i-1},b^i) \in {\cal M}({\mb A}_i): i=0,1,\ldots,n\big\}$ be arbitrary distributions and define the joint distribution on ${\cal M}({\mb A_0^n }\times {\mb B}_0^n)$ by $\otimes_{i=0}^n\Big(S_i(db_i|b^{i-1},a^{i-1})\otimes {R}_i(da_i|a^{i-1},b^i)\Big)$. 
%For a fixed $\overleftarrow{P}_{0,n}(da^n|b^{n-1}) \in {\cal M}({\mb A}^n)$ and   $\overrightarrow{Q}_{0,n}(db^n|a^n) \in {\cal M}({\mb B}^n)$,  define the following functional.
Then the following variational equality holds. 
\begin{align}
I(A^n\rightarrow{B}^n) =&\sup_{\substack{\big\{S_i(db_i|b^{i-1},a^{i-1})\otimes{R}_i(da_i|a^{i-1},b^{i})\in{\cal M}({\mb A}_i\times{\mb B}_i): i=0,1,\ldots,n\big\}\\\big\{S_i(db_i|b^{i-1},a^{i-1})\in{\cal M}({\mb B}_i),~R_i(da_i|a^{i-1},b^{i})\in{\cal M}({\mb A}_i)\big\}}}  \sum^n_{i=0}\int_{ {\mb G}^i}\log\Bigg(\frac{d{R}_i(\cdot|a^{i-1},b^{i})}{dP_i(\cdot|a^{i-1},b^{i-1})}(a_i)\nonumber \\
&.\frac{dS_i(\cdot|b^{i-1},a^{i-1})}{d\Pi_{i}^P(\cdot|b^{i-1})}(b_i)\Bigg) 
 {\bf P}_\nu^P(da^i, db^i)   \label{equation15a}
\end{align}
and the supremum in (\ref{equation15a}) is achieved when the following identity holds. 
\begin{align}
\frac{dP_i(\cdot|a^{i-1},b^{i-1})}{d{R}_i(\cdot|a^{i-1},b^i)}(a_i).\frac{d{Q}_i(\cdot|b^{i-1},a^{i})}{dS_i(\cdot|b^{i-1},a^{i-1})}(b_i)=1-a.a.(a^n,b^n),~i=0,1,\ldots,n.\label{equation102}
\end{align}
Equivalently, the supremum in (\ref{equation15a}) is achieved at 
\bea
\otimes_{i=0}^n\Big(S_i(db_i|b^{i-1},a^{i-1})\otimes {R}_i(da_i|a^{i-1},b^i)\Big)= {\bf P}_{a^{-1}, b^{-1}}^P(da^n, db^n).
\eea
\end{theorem}
\begin{proof} These are derived in  \cite{charalambous-stavrou2013aa}, Theorem~IV.1.  
\end{proof} 

We shall use the  variation equality in (a) to  identify upper bounds on directed information, which  are  achievable over specific subsets of the set of distributions ${\cal P}_{[0,n]}$ and ${\cal P}_{[0,n]}(\kappa)$, which depend on the properties of the channel distribution and the transmission cost function. This procedure is applied recently in \cite{kourtellaris-charalambousIT2015_Part_1} to derive the information structures of optimal channel input distributions for channel distributions and transmission cost functions corresponding to $L=N=0$. 
 We apply the second variation equality to  identify lower bounds on directed information, which are achievable over specific subsets of the set of distributions ${\cal P}_{0,n]}$ and ${\cal P}_{[0,n]}(\kappa)$. The first variational equality encompasses as a special case, the maximum entropy properties of joint and conditional distributions, such as,  the maximizing entropy property of Gaussian distributions. 
% Specifically, if the corresponding densities of the distributions exist, then directed information is expressed as 
%\begin{align} 
% I(A^n \rar B^n)=& H(B^n)- \sum_{i=0}^n H(B_i|B^{i-1}, A^i) \\
% &= \inf_{ V_{B^n}(b^n)} \Big\{-\int_{{\mathbb B}^n}  
%\end{align} 
%  the decomposition if  
%

%In addition to the above variational equalities, we utilize the sufficient conditions for continuity of  the functional ${\mathbb I}_{A^n\rightarrow B^n}(\overleftarrow{P}_{0,n}, \overrightarrow{ Q}_{0,n})$ in 
%
%
%
% can also give sufficient conditions for the following  identities to hold.
%(iii) If the RVs are finite-valued (i.e., $\{ {\mb A}_i, {\mb B}_i: i=0, \ldots, n\}$ are finite alphabet spaces) then the following holds.
%\begin{align}
%&\sup_{ \big\{P_i(da_i|a^{i-1}, b^{i-1}): i=0, \ldots,n \big\} \in  {\cal P}_{[0,n]} } I(A^n\rar B^n) = \sup_{   \overleftarrow{P}_{0,n}(da^n|b^{n-1})  \in  {\cal M}({\mb A}^n)  } \inf_{ V_{0,n}(db^n) \in {\cal M}({\mb B}^n)}  {\mb I}_{0,n}(V_{0,n}; \overleftarrow{P}_{0,n}, \overrightarrow{Q}_{0,n})\\
%& =\inf_{ V_{0,n}(db^n) \in {\cal M}({\mb B}^n)} \sup_{   \overleftarrow{P}_{0,n}(da^n|b^{n-1})  \in  {\cal M}({\mb A}^n)  } {\mb I}_{0,n}(V_{0,n}; \overleftarrow{P}_{0,n}, \overrightarrow{Q}_{0,n})\\
%&= \inf_{  \big\{ V_i(db_i| b^{i-1} )\in {\cal M}({\mb  B}_{i}): i=0,1, \ldots, n\big\}}\sup_{ \big\{P_i(da_i|a^{i-1}, b^{i-1}): i=0, \ldots,n \big\} \in  {\cal P}_{[0,n]} }  \sum_{i=0}^n  \int_{{\mb A}^i\times {\mb  B}^i}   \log \Big( \frac{dQ_i(\cdot|b^{i-1},a^i) }{ dV_i(\cdot|b^{i-1})}(b_i)\Big)   {\bf P}^P(da^i,db^i).
%\end{align}
%

Often, we use the following  alternative version of the variational given in Theorem~\ref{thm-var}, (a).  \\
Given a channel input distribution $\big\{P_i(da_i|a^{i-1},b^{i-1}): i=0,1, \ldots, n\big\} \in {\cal P}_{[0,n]}$ and channel distribution  $\big\{Q_i(db_i|b^{i-1},a^i) : i=0,1, \ldots, n\big\} \in {\cal C}_{[0,n]} $,  define the  corresponding  joint and  marginal distributions  ${\bf P}_{a^{-1}, b^{-1}}^P(da^n,db^n)\equiv (\overleftarrow{P}_{0,n} \otimes \overrightarrow{Q}_{0,n})(da^n, db^{n}) \in {\cal M}({\mb A}_0^n \times {\mb  B}_0^n)$, $\overrightarrow{\Pi }_{0,n}^P(db^n) = \otimes_{i=0}^n{\Pi }_i^P(db_i|b^{i-1}) \equiv \overrightarrow{\Pi}_{0,n}^{\overleftarrow{P}}(db^n) \in {\cal M}({\mb  B}_0^n)$ by  (\ref{CIS_3a_new}). 

(a)  Let $\overrightarrow{V}_{0,n}(db^n) \tri   \otimes_{i=0}^n V_i(db_i|b^{i-1}) \in{\cal M}({\mb  B}_0^n)$ be any arbitrary distribution on ${\mathbb B}_0^n$, for a fixed $ B^{-1}=b^{-1}$,  which is uniquely defined by   $\big\{V_i(db_i|b^{i-1}) \in {\cal M}({\mb B}_i): i=0, \ldots, n\}$   and vice-versa. \\
For a fixed  $\nu(da^{-1}, db^{-1})\in {\cal M}({\mb A}^{-1} \times {\mb B}^{-1})$,  $\overleftarrow{P}_{0,n}(da^n|b^{n-1}) \in {\cal M}({\mb A}_0^n)$ and   $\overrightarrow{Q}_{0,n}(db^n|a^n) \in {\cal M}({\mb B}_0^n)$,  define the following functional. 
\begin{align}
&{\mb I}_{0,n}^\nu(\cdot, \overleftarrow{P}_{0,n}, \overrightarrow{Q}_{0,n})  : {\cal M}({\mb B}_0^n) \longmapsto \big\{{\mathbb R}, +\infty\},  \hso \overrightarrow{V}_{0,n}(db^n) \longmapsto  {\mb I}_{0,n}^\nu(\overrightarrow{V}_{0,n},\overleftarrow{P}_{0,n}, \overrightarrow{Q}_{0,n} ), \\
 &{\mb I}_{0,n}^\nu(\overrightarrow{V}_{0,n}, \overleftarrow{P}_{0,n}, \overrightarrow{Q}_{0,n}))\tri    \int_{{\mb G}^n}   \log \Big( \frac{d\overrightarrow{Q}_{0,n}(\cdot|a^n) }{ d\overrightarrow{V}_{0,n}(\cdot)}(b^n)\Big)     (\overleftarrow{P}_{0,n} \otimes \overrightarrow{Q}_{0,n})(da^n, db^{n}) \otimes \nu(da^{-1}, db^{-1}) \label{BAM52a_new}
\end{align}
 Then the following hold.\\
 (i) The functional ${\mb I}_{0,n}^\nu(\overrightarrow{V}_{0,n}, \overleftarrow{P}_{0,n}, \overrightarrow{Q}_{0,n})$ is convex in  $\overrightarrow{V}_{0,n}(db^n) \in{\cal M}({\mb  B}_0^n)$ for  fixed $\overleftarrow{P}_{0,n}(da^n|b^{n-1}) \in {\cal M}({\mb A}_0^n)$,    $\overrightarrow{Q}_{0,n}(db^n|a^n) \in {\cal M}({\mb B}_0^n)$, and $\nu(da^{-1}, db^{-1})\in {\cal M}({\mb A}^{-1} \times {\mb B}^{-1})$.\\
 (ii) The following variational equality holds.
\bea
I(A^n \rar B^n) =\inf_{ \overrightarrow{V}_{0,n}(db^n) \in {\cal M}({\mb B}_0^n)}  {\mb I}_{0,n}^\nu(\overrightarrow{V}_{0,n}, \overleftarrow{P}_{0,n}, \overrightarrow{Q}_{0,n}) \label{BAM52a_new}
 \eea
 and the infimum in (\ref{BAM52a_new}) is achieved at
$\overrightarrow{V}_{0,n}(db^n)= \overrightarrow{\Pi}_{0,n}^{\overleftarrow{P}}(db^n)$ given by (\ref{CIS_3a_new}). 

Variational equality given in Theorem~\ref{thm-var}, (a) is often appropriate when it is applied together with dynamic programming, while the alternative one is appropriate to understand the convexity properties of ${\mathbb I}_{0,n}^\nu(\cdot, \cdot, \cdot)$ as a functional of causally conditioned compound distributions.

\section{Characterization of FTFI   Capacity }
 \label{randomized}
 In this section, we derive  the information structures of optimal channel input distributions, as described in Section~\ref{meth}. Using the established  notation, the channel output process $\{B_i: i=0, 1, \ldots, n\}$   is controlled, by controlling   $\big\{ {\bf P}(db_i|b^{i-1})\equiv \Pi_i^{ P}(db_i| b^{i-1})\in {\cal M}({\mb  B}_i): i=0, \ldots, n\big\}$  via the choice of the control object $\big\{ P_i(da_i| a^{i-1}, b^{i-1}): i=0, \ldots, n\big\}\in {\cal P}_{[0,n]}$.\\
We derive the characterizations of FTFI capacity in the following sequence.  

{\it Step 1-Channel Distributions and Transmission Cost Functions of Class A or B with $L\neq 0, N\neq 0$.}  Given a channel distribution of Class A or B, and  transmission cost functions of Class A or B, where $\{L, N\}$ are finite and different than zero, we   show via stochastic optimal control and variational equality (\ref{BAM52a}), that at each time instant $i$, the optimal  channel input distribution lies in a subset $\overline{\cal P}_{[0,n]} \subseteq {\cal P}_{[0,n]}$, which satisfies conditional independence and it is  of finite memory with respect to past channel input symbols, for $i=0, \ldots, n$. 
This impies for each $i$,  the information structure of the optimal channel input distribution lies in a subset ${\cal I}_i^P \subseteq \{a^{i-1}, b^{i-1}\}, i=0, 1, \ldots, n$.

{\it Step 2-Channel Distributions and Transmission Cost Functions of Class C with $L\neq 0, N\neq 0$.} Given a channel distribution of Class C, and  transmission cost functions of Class C, since these are special cases of the ones in Step 1, then the optimal channel input distributions lie in a subset $\overline{\cal P}_{[0,n]} \subseteq {\cal P}_{[0,n]}$, which satisfy conditional independence. 

{\it Step 3-Channel Distributions and Transmission Cost Functions of Class C with $L=N=0$.} Given a channel distribution of Class C, and  transmission cost functions of Class C with $L=N=0$,  we can further apply stochastic optimal control and the variational equality (\ref{BAM52a}),   to the resulting optimization problem  of Step 1, to obtain an  upper bound, which is achievable over smaller subsets of  conditional distributions $\sr{\circ}{\cal P}_{[0,n]} \subset \overline{\cal P}_{[0,n]}$, which satisfy conditional independence and  have  finite memory with respect to channel output symbols. 
However, since this is already shown in \cite{kourtellaris-charalambousIT2015_Part_1}, we will concentrate on  the fundamental differences of the information structures between $L=N=0$ and $L\neq 0$ and/or  $N\neq 0$, i.e., corresponding to the channels and transmission cost functions considered in steps 1, 2.

Although, in  Step 1 we invoke  generalizations of  methods often applied  in stochastic optimal  control problems to show  that optimizing a   pay-off \cite{vanschuppen2010,kumar-varayia1986} over all non-Markov policies or strategies,  occurs in the smaller set of Markov policies, there  are certain issues which should be treated with caution, because  extremum problems of information theory are  distinct from  any of the common pay-offs of stochastic optimal control. We discuss some of the fundamental  differences below to clarify subsequent derivations of information structures of optimal channel input distributions. 

{\bf Stochastic optimal control Theory versus Extremum Problems of Information Theory.} 
In classical  stochastic optimal control theory \cite{hernandezlerma-lasserre1996}, we are often given a controlled process $\{X_i: i=0, \ldots, n\}$,  called the state process taking values in $\big\{{\mathbb X}_i: i=0, \ldots, n\big\}$,  affected by  a control process $\{U_i: i=0, \ldots,n \}$ taking values in $\big\{{\mathbb U}_i: i=0, \ldots, n\big\}$, and the corresponding  control object distribution ${\cal P}_{[0,n]}^{CO}\tri   \big\{{\bf P}_{U_i|U^{i-1}, X^{i}}: i=0, \ldots, n\big\}$ and a general non-Markov  controlled object distribution ${\cal C}_{[0,n]}^{CO}\tri \big\{{\bf P}_{X_i|X^{i-1}, U^{i-1}}: i=0, \ldots, n\big\}$. \\
However, often the controlled object distribution is Markov  conditional on the past control values, that is,  ${\bf P}_{X_i|X^{i-1}, U^{i-1}}={\bf P}_{X_i|X_{i-1}, U_{i-1}}-a.a. (x^{i-1}, u^{i-1}), i=0, \ldots, n$. Such Markov controlled objects are often  induced by  discrete recursions 
\bea
X_{i+1}=f_i(X_{i}, U_i, V_{i}), \hso X_0=x_0, \hso i=0, \ldots, n \label{SOC_3}
\eea
 where $\{V_i: i=0, \ldots, n\}$ is an independent noise process taking values in $\big\{{\mathbb V}_i: i=0, \ldots, n\big\}$, independent of the initial state $X_0$. Let us denote the set of such Markov distributions  or controlled objects by  ${\cal C}_{[0,n]}^{CO-M}\tri \big\{{\bf P}_{X_i|X_{i-1}, U_{i-1}}: i=0, \ldots, n\big\}$.\\
In stochastic optimal control theory,  we are also  given a sample pay-off function to grade the behaviour of each of the strategies, often of additive form, defined by
\bea 
l: {\mathbb X}^n \times {\mathbb U}^n \longmapsto (-\infty, \infty], \hst  l(x^n, u^n)\tri  \sum_{i=0}^n\ell_i(u_i,x_{i}) \label{SOC_4}
\eea 
  where the functions $\big\{\ell_i(\cdot, \cdot): i=0, \ldots, n\}$ are fixed and independent of the control object $\big\{{\bf P}_{U_i|U^{i-1}, X^{i}}: i=0, \ldots, n\big\}$.\\
The main problem of stochastic optimal control is the following. Given a Markov controlled object distribution from the set ${\cal C}_{[0,n]}^{CO-M}$, determine the optimal strategy among all non-Markov strategies in ${\cal P}_{[0,n]}^{CO}$, which impacts the minimum  average of the sample path pay-off,  i.e.,
\bea
J_{0,n}^F( {\bf P}_{U_i|U^{i-1}, X^{i}}^*, {\bf P}_{X_i|X_{i-1}, U_{i-1}}: i=0, \ldots, n) \tri  \inf_{{\cal P}_{0,n]}^{CO}} {\bf E}\Big\{  \sum_{i=0}^n\ell(U_i,X_{i})\Big\}. \label{SOC_1}
\eea
Hence, for any non-Markov strategy from the set  ${\cal P}_{[0,n]}^{CO}$, the functional $J_{0,n}^F( {\bf P}_{U_i|U^{i-1}, X^{i}}, {\bf P}_{X_i|X_{i-1}, U_{i-1}}: i=0, \ldots, n)$ depends on a fixed and given controlled object distribution $ {\bf P}_{X_i|X_{i-1}, U_{i-1}}, i=0, \ldots, n.$
Next, we discuss two features of stochastic optimal control which are distinct from  any extremum problem of directed information.

{\bf Feature 1.} The definition of stochastic optimal control formulation (\ref{SOC_1})  pre-supposes the following. \\
(i) The controlled object distribution  is Markov, i.e., ${\cal C}_{[0,n]}^{CO-M}\tri \big\{{\bf P}_{X_i|X_{i-1}, U_{i-1}}: i=0, \ldots, n\big\}$;\\
(ii) at each $i$, the sample path pay-off is single letter, i.e., $\ell_i(u_i, x_i)$ for $i=0, \ldots, n$.\\
If (i) and/or (ii) do not hold, then prior to arriving to the formulation (\ref{SOC_1}),  additional state variables are introduced,  which constitute the complete state process $\{X_i: i=0, \ldots, n\}$ so that (i) and (ii) hold. This may be due to  a  noise process  $\{V_i: i=0, \ldots, n\}$ which is correlated, a dependence of the  discrete recursion on past information, and a dependence of  the sample pay-off function $\ell_i(\cdot, \cdot)$  at each $i$ on additional variables than single letters  $(x_i,u_i)$, for $i=0, \ldots, n$, and converted into the formulation (\ref{SOC_1}), satisfying (i) and (ii), by  state augmentation, so that  the controlled object is Markov, and the sample path pay-off is single letter. The procedure  is given in \cite{bertsekas1995} for deterministic or non-randomized strategies, defined by 
\bea
{\cal E}_{[0,n]}^{CO} \tri \big\{e_i: {\mathbb U}^{i-1}\times {\mathbb X}^i \longmapsto {\mathbb U}_i, \hso i=0, \ldots, n: \hso   u_i=e_i(u^{i-1}, x^i), i=0, \ldots, n\big\}.
\eea
In view of the Markovian property of the controlled object, i.e., satisfying  ${\bf P}_{X_i|X^{i-1}, U^{i-1}}={\bf P}_{X_i|X_{i-1}, U_{i-1}}, i=0, \ldots, n$, then it can be shown  that   the optimization in (\ref{SOC_1}) over all non-Markov strategies  reduces to an optimization problem over  Markov strategies, as follows \cite{kumar-varayia1986,hernandezlerma-lasserre1996}.
\begin{align}
J_{0,n}^F({\bf P}_{U_i|U^{i-1}, X^{i}}^*,{\bf P}_{X_i|X_{i-1}, U_{i-1}}: i=0, \ldots, n)=&J_{0,n}^M({\bf P}_{U_i|X_{i}}^*, {\bf P}_{X_i|X_{i-1}, U_{i-1}}: i=0, \ldots, n)  \\
\tri & \inf_{ {\bf P}_{U_i|X_{i}}, i=0, \ldots, n} {\bf E}\Big\{  \sum_{i=0}^n\ell(U_i,X_{i})\Big\}. \label{SOC_2}
\end{align}
This further implies that the control process $\{X_i: i=0, \ldots, n\}$ is Markov, i.e., it satisfies ${\bf P}_{X_i|X^{i-1}}= {\bf P}_{X_i|X_{i-1}}, i=0, \ldots, n$. On the other hand, if ${\bf P}_{X_i|X^{i-1}}= {\bf P}_{X_i|X_{i-1}}, i=0, \ldots, n$ then (\ref{SOC_2}) holds.

{\bf Feature 2.} Given a general non necessarily  Markov controlled object $\big\{{\bf P}_{X_i|X^{i-1}, U^{i-1}}: i=0, \ldots, n\big\}$, one of the fundamental results of classical stochastic optimal control is that the optimization of the average  pay-off ${\bf E}\Big\{  \sum_{i=0}^n\ell(U_i,X_{i})\Big\}$ over all non-Markov randomized strategies ${\cal P}_{[0,n]}^{CO}$ does not incur a better performance than optimizing it over non-Markov and non-randomized strategies ${\cal E}_{[0,n]}^{CO}$, i.e., the following holds. 
\begin{align}
J_{0,n}^F( {\bf P}_{U_i|U^{i-1}, X^{i}}^*,&{\bf P}_{X_i|X^{i-1}, U^{i-1}}: i=0, \ldots, n)=   \inf_{{\cal E}_{[0,n]}^{CO} } {\bf E}\Big\{  \sum_{i=0}^n\ell(U_i,X_{i})\Big\} \label{SOC_5}\\
=& \inf_{g_i(X_i):\; i=0, \ldots, n } {\bf E}^g\Big\{  \sum_{i=0}^n\ell(U_i,X_{i})\Big\} \hso \mbox{if} \hso  {\bf P}_{X_i|X^{i-1}, U^{i-1}}={\bf P}_{X_i|X_{i-1}, U_{i-1}}, i=0, \ldots, n. \label{SOC_6}
\end{align}
We note that in any extremum problem  of directed information, Features 1 and 2 above do not hold.\\
Specifically, the sample path pay-off is the directed information density, and this is a functional of the channel output conditional distribution, which depends on the channel input conditional distribution. Since the directed information density is not a fixed functional,  then  Feature 1 of stochastic optimal control formulation does not hold for extremum problems of directed information. The dependence of the directed information density or sample path pay-off on the channel output conditional distribution, which is induced by the channel distribution and the channel input conditional distribution makes extremum problems of directed information distinct compared to classical stochastic optimal control problems.  \\
Further, Feature 2 does not hold in extremum problems of directed information, because if the channel input distributions are replaced by non-randomized deterministic strategies, then directed information is zero. \\
In view of Features 1 and 2 of stochastic optimal control, any application of stochastic optimal control techniques to derive the information structures of optimal channel input distributions, which maximize directed information, needs to be treated with caution. Often, stochastic optimal control techniques might not be directly applicable and properties of optimal channel input distributions need to be derived from first principles. Also,  additional properties of directed information density might be needed, such as, the variational equality of directed information, to determine achievable upper bounds. 

%\end{remark}

%In this section, we show the following.
%
%\begin{description}
%\item[(i)] Channel distributions (\ref{CD_C2})-(\ref{CD_C5}), are induced  by various nonlinear channel models (NCM), driven by arbitrary distributed noise processes. These include  nonlinear and linear time-varying  Autoregressive models, and nonlinear and linear  channel models expressed in state space form \cite{caines1988}. 
% 
% \item[(ii)] The optimal channel input conditional  distributions of Multiple-Input Multiple Output (MIMO) Gaussian Linear Channel Models (G-LCM), driven by correlated Gaussian noise processes,  which  maximize directed information $I(A^n \rar B^n)$,   are Gaussian.   
%\end{description}
%Claim (i) illustrates that  many of the existing channels investigated in the literature, for example,  \cite{cover-pombra1989,kim2008,chen-berger2005,yang-kavcic-tatikonda2005,permuter-cuff-roy-weissman2010,permuter-asnani-weissman2013,elishco-permuter2014,kourtellaris-charalambous2015}, induce channel distributions of Class A, B or C.
%Claim  (ii) generalizes the Cover and Pombra \cite{cover-pombra1989} characterization (\ref{cp1989_a}) of feedback capacity of  nonnstationary nonergodic Additive Gaussian channels driven by correlated noise.
%

\subsection{Channel Class A or B and Transmission Cost Class A or B}
\label{class_A}
%One of the important applications of variational equalities of directed information, as well as, those of mutual information, is to apply them in extremum problems of information theory to identify lower or upper bounds which are achievable.

Throughout this section we use the following definition of channel input distributions satisfying conditional independence.\\

\begin{definition}(Conditional independence for class A channels and class B transmission cost functions)\\
\label{CI_CLASS_A_B}
Consider the class of channel input conditional distributions ${\cal P}_{[0,n]}$ and define the set of channel input conditional distributions for Class A channels and Class B transmission cost constraints  by 
\begin{align}
{\cal P}_{[0,n]}^{A}(\kappa) \tri     \Big\{   P_i(da_i | a^{i-1}, b^{i-1}), \hso  i=0, 1, \ldots, n:\hso    \frac{1}{n+1} {\bf E}^{P}\Big(\sum_{i=0}^n \gamma_i^{A.N}(A_{i-N}^i, B^{i})\Big) \leq \kappa  \Big\}\subset {\cal P}_{[0,n]}. \label{cor-ISR_29_cc_CC}
\end{align}
A subclass of channel input conditional  distributions from ${\cal P}_{[0,n]}$ for Class A channels,  which satisfy conditional independence  is defined by 
\begin{align}
\overline{\cal P}_{[0,n]}^{A.L} \tri     \Big\{   P_i(da_i | a^{i-1}, b^{i-1})= \pi_i^{A.L}(da_i|a_{i-L}^{i-1}, b^{i-1})-a.a. (a^{i-1}, b^{i-1}):\hso i=0, 1, \ldots, n\Big\} \subset {\cal P}_{[0,n]}  \label{FRD_1}
\end{align}
A subclass of channel input  conditional distributions from the set ${\cal P}_{[0,n]}^{A}(\kappa)$,  for Class A channels and Class B transmission cost constraints, which satisfy conditional  independence is defined by 
\begin{align}
\overline{\cal P}_{[0,n]}^{A.I}(\kappa) \tri     \Big\{   & P_i(da_i | a^{i-1}, b^{i-1})= \pi_i^{A.I}(da_i|a_{i-I}^{i-1}, b^{i-1})-a.a. (a^{i-1}, b^{i-1}), \hso  i=0, 1, \ldots, n:  \nonumber \\
& \frac{1}{n+1} {\bf E}^{\pi^{A.I}}\Big(\sum_{i=0}^n \gamma_i^{A.N}(A_{i-N}^i, B^{i})\Big) \leq \kappa  \Big\}\subset {\cal P}_{[0,n]}^A(\kappa), \hso I \tri \max\{L, N\}. \label{cor-ISR_29_cc}
\end{align}
\end{definition}

\subsubsection{\bf Channel Class A and Transmission Cost Class A} 
\label{Class_A_SOC} 
Given  the channel distribution (\ref{CD_C2}), the   joint distribution is defined  by\footnote{Often we do not indicate the dependence of the distributions ${\bf P}(\cdot)$  and expectation ${\bf E}\{\cdot\}$ on the initial data, $\nu(db^{-1})$ or $b^{-1}$, because  these are easily extracted from the definitions.}  
\bea
{\bf P}^P(da^i, db^i) \tri \otimes_{j=0}^i\Big( P_j(da_j|a^{j-1}, b^{j-1})\otimes Q_j(db_j|b^{j-1}, a_{j-L}^j)\Big), \hso i=0, \ldots, n. \label{CIS_9a_new_1}
\eea
  Consequently,  directed information is given by 
\begin{align}
I(A^n\rightarrow B^n)=& 
\sum_{i=0}^n {\bf E}^{ P}\Big\{
\log\Big(\frac{dQ_i(\cdot|B^{i-1},A_{i-L}^i)}{d\Pi_i^{ P}(\cdot|B^{i-1})}(B_i)\Big)
\Big\} \equiv \sum_{i=0}^{n}I(A_{i-L}^i;B_i|B^{i-1})  \label{cor-ISR_25a_new}\\
=&\sum_{i=0}^n {\bf E}^{ P}\Big\{ \ell_i^P(A_i, A_{i-L}^{i-1},B^{i-1})    \Big\}\label{cor-ISR_25a_new_new}\\
\equiv& \sum_{i=0}^n {\bf E}^{ P}\Big\{ \ell_i^P(A_i, A_{i-L}^{i-1},S_i)    \Big\}
\end{align}
where the sample path pay-off and the conditional distribution of the channel output obtained from   (\ref{CIS_3a}) are given by
 \begin{align}
 \ell_i^P(a_i, a_{i-L}^{i-1},b^{i-1}) \tri & \int_{{\mb B}_i}^{}
\log\Big(\frac{dQ_i(\cdot|B^{i-1},a_{i-L}^i)}{d\Pi_i^{ P}(\cdot|b^{i-1})}(b_i)\Big)Q_i(db_i|b^{i-1},a_{i-L}^i) \\
\equiv& \ell^P(a_i, a_{i-L}^{i-1}, s_i), \hst s_j\tri b^{j-1}, \; j=0, \ldots, n, \\
\Pi_i^{ P} (db_i| b^{i-1})=&  \int_{{\mb A}^i} Q_i(db_i|b^{i-1},  a_{i-L}^i)\otimes P_i(da_i|a^{i-1}, b^{i-1})\otimes {\bf P}^{ P}(da^{i-1}|b^{i-1}), \hso  i=0, \ldots, n. \label{CIS_9a_new}
%\\
% =& \int_{{\mb A}^i} Q_i(db_i|b^{i-1},  a_{i-L}^i)\otimes {\bf P}^{ P}(da_{i-L}^{i}|b^{i-1})
\end{align}
It is important to note that for each $i$, the \'a posteriory distribution ${\bf P}^{ P}(da^{i-1}|b^{i-1})$ in (\ref{CIS_9a_new}) depends on the channel input distribution $\{P_j(da_j|a^{j-1}, b^{j-1}): j=0, \ldots, i-1\}$, for $i=0,1,\ldots, n$.  

\noi {\bf Information Structures of Optimal Channel Input Distributions.}
First, we make the following observation. 
For each $i$, the pay-off in (\ref{cor-ISR_25a_new_new}), i.e., $
\ell_i^P(a_i, a_{i-L}^{i-1},b^{i-1})$
  depends on $(a_{i-L}^{i-1},b^{i-1})$ through the channel distribution dependence on these variables, and the control object   $ \big\{g_j(a^{j-1}, b^{j-1})\tri  { P}_j(da_j| a^{j-1}, b^{j-1}): j=0, \ldots, i\big\}$, via   $\big\{\xi_j^P(b^{j-1})\tri \Pi_i^{ P}(db_i|b^{i-1}) : j=0, \ldots, i\}$, defined by ({\ref{CIS_9a_new}), for $i=0, \ldots, n$.  Moreover, for each $i$,  $\xi_i^P(b^{i-1})$  depends   on $b^{i-1}$ through the channel distribution and the  control object $ \big\{g_j(a^{j-1}, b^{j-1}): j=0, \ldots, i\big\}$, for $i=0, \ldots, n$. Moreover, for each $i$, the common information  to the encoder (channel input distribution) and to the decoder is the process $S_i\tri B^{i-1}$, for $i=0, \ldots, n$.\\
  Next,  we show that $\{\xi_i^P(b^{i-1})\tri \Pi_i^P(db_i|b^{i-1}): i=0, \ldots, n\}$ is a functional of the  object  $\big\{{\bf P}(da_i|a_{i-L}^{i-1}, b^{i-1}): i=0, \ldots, n\big\}$, instead of  the control object   $ \big\{g_j(a^{j-1}, b^{j-1})\tri  { P}_j(da_j| a^{j-1}, b^{j-1}): j=0, \ldots, i\big\}$. First, we apply Bayes' theorem and we use the property of the channel distribution, to deduce the following conditional independence holds. 
\begin{align}
{\bf P}(ds_{i+1}|s^i, a^i)= {\bf P}(ds_{i+1}|s_i, a_{i-L}^i), \hso i=0, \ldots, n-1. \label{joint_p}
\end{align}
  It is easy to verify that   the process $\{S_i \tri B^{i-1}: i=0, \ldots, n\}$ is Markov, and satisfies the following identities.
\begin{align}
{\bf P}^P(ds_{i+1}|s^{i})={\bf P}^P(ds_{i+1}|s_{i})=\int_{{\mathbb A}_{i-L}^i}{\bf P}(ds_{i+1}|s_i, a_{i-L}^i)\otimes {\bf P}(da_i|a_{i-L}^{i-1}, s_i) \otimes {\bf P}^P(da_{i-L}^{i-1}|s_i), \hso i=0, \ldots, n-1 \label{MV}
\end{align}
where the second equality follows from (\ref{joint_p}).
Further, we show that  the \'a posteriori distribution ${\bf P}^P(da_{i-L}^{i-1}|s_i)\equiv {\bf P}^P(da_{i-L}^{i-1}|b^{i-1})$ appearing in (\ref{MV}) is a functional of  the conditional distribution \\ $\{{\bf P}(da_j|a_{j-L}^{j-1}, s_j): j=0, \ldots, i-1\}$ instead of the distribution $P_j(da_j|a^{j-1}, s_j): j=0, \ldots, i-1\}$, as follows. By applying Bayes' theorem we obtain the following recursion. 
\begin{align}
{\bf P}^P(da_{i-L}^{i-1}|s_{i}) =& \frac{\int_{{\mathbb A}_{i-1-L}} Q_{i-1}(db_{i-1}|s_{i-1}, a_{i-1-L}^{i-1})\otimes {\bf P}(da_{i-1}|a_{i-1-L}^{i-2}, s_{i-1}) \otimes {\bf P}^P(da_{i-1-L}^{i-2}|s_{i-1})}{{\bf P}^P(db_{i-1}|s_{i-1})},\; i=1, \ldots, n, \label{rec_g1} \\
{\bf P}^P(da_{-L}^{-1}|s_0)=&\mbox{given} \label{rec_g2} 
\end{align}
where the denominator is given by
\begin{align}
{\bf P}^P(db_{i-1}|s_{i-1}) =& \int_{{\mathbb A}_{i-1-L}^{i-1}} Q_{i-1}(db_{i-1}|s_{i-1}, a_{i-1-L}^{i-1})\otimes {\bf P}(da_{i-1}|a_{i-1-L}^{i-2}, s_{i-1}) \otimes {\bf P}^P(da_{i-1-L}^{i-2}|s_{i-1}). \label{rec_g3}
\end{align}
For a fixed channel distribution $Q_{i-1}(\cdot|\cdot, \cdot)$ define the operator appearing in the numerator of (\ref{rec_g1}) by 
\begin{align}
&\Big(b_{i-1}, s_{i-1}, {\bf P}(\cdot|\cdot, s_{i-1} ), {\bf P}^P(\cdot|s_{i-1})\Big) \longmapsto T_{i-1}\Big(b_{i-1}, s_{i-1}, {\bf P}(\cdot|\cdot, s_{i-1}), {\bf P}^P(\cdot|s_{i-1}) \Big)(\cdot), \hso i=1, \ldots, n,\\
 &T_{i-1}\Big(b_{i-1}, s_{i-1}, {\bf P}(\cdot|\cdot, s_{i-1}), {\bf P}^P(\cdot|s_{i-1}) \Big)(da_{i-L}^{i-1}) \tri \int_{{\mathbb A}_{i-1-L}} Q_{i-1}(db_{i-1}|s_{i-1}, a_{i-1-L}^{i-1})\otimes {\bf P}(da_{i-1}|a_{i-1-L}^{i-2}, s_{i-1}) \nonumber \\
 & \hst \hst \otimes {\bf P}^P(da_{i-1-L}^{i-2}|s_{i-1})
\end{align}
Then (\ref{rec_g1}) is expressed as follows.
\begin{align}
{\bf P}^P(da_{i-L}^{i-1}|s_{i}) =& \frac{ T_{i-1}\Big(b_{i-1}, s_{i-1}, {\bf P}(\cdot|\cdot, s_{i-1}), {\bf P}^P(\cdot|s_{i-1}) \Big)(da_{i-L}^{i-1}) }{ \int_{{\mathbb A}_{i-L}^{i-1}}T_{i-1}\Big(b_{i-1}, s_{i-1}, {\bf P}(\cdot|\cdot, s_{i-1}), {\bf P}^P(\cdot|s_{i-1}) \Big)(da_{i-L}^{i-1})},\hso i=1, \ldots, n \\
\equiv & \tilde{T}_{i-1}\Big(b_{i-1}, s_{i-1}, {\bf P}(\cdot|\cdot, s_{i-1}), {\bf P}^P(\cdot|s_{i-1}) \Big)(da_{i-L}^{i-1}), \label{A_POST_1}\\
{\bf P}^P(da_{-L}^{-1}|s_0)=&\mbox{given}.
\end{align}
By iterating (\ref{A_POST_1}), then we deduce that the conditional distribution
${\bf P}^P(da_{i-L}^{i-1}|s_{i})$ is a functional of the control object  $\big\{{g}_{j-1}^{A.L}(a_{j-1-L}^{j-2}, b^{j-2})\equiv  \pi_{j-1}^{A.L}(da_{j-1}|a_{j-1-L}^{j-2}, b^{j-2}): j=0, \ldots,i \big\}$ and $b^{i-1}$, i.e., for each $i$,  ${\bf P}^P(da_{i-L}^{i-1}|b^{i-1})\equiv {\bf P}^{\pi^L}(da_{i-L}^{i-1}|b^{i-1})$. \\
Thus,  from (\ref{MV}) we deduce that for each $i$, the conditional distribution ${\bf P}^P(ds_{i+1}|s^i) \equiv {\bf P}^P(ds_{i+1}|b^{i-1})\equiv {\bf P}^{\pi^L}(ds_{i+1}|b^{i-1})$ is a functional of the control object  $\big\{{g}_{j}^{A.L}(a_{j-L}^{j-1}, b^{j-1})\equiv  \pi_{j}^{A.L}(da_{j}|a_{j-L}^{j-1}, b^{j-1}): j=0, \ldots,i \big\}$ and $b^{i-1}$. This implies, for each $i$,  that  $\Pi_i^P(db_i|b^{i-1})\equiv     \Pi_i^{\pi^{A.L}}(db_i|b^{i-1})\equiv \xi_i^{\pi^{A.L}}(b^{i-1})$ is also a functional  of the control object  $\big\{{g}_{j}^{A.L}(a_{j-L}^{j-1}, b^{j-1})\equiv  \pi_{j}^{A.L}(da_{j}|a_{j-L}^{j-1}, b^{j-1}): j=0, \ldots,i \big\}$ and $b^{i-1}$. \\
Using the above facts, we can express the distribution of $\{S_i: i=0, \ldots, n\}$ as follows.
\begin{align}
&{\bf P}^P(ds_{i+1})= \int_{{\mathbb B}^{i-1}\times {\mathbb A}_{i-L}^i} {\bf P}(ds_{i+1}|s_i, a_{i-L}^{i})\otimes {\bf P}(da_{i-L}^i|s_i)\otimes {\bf P}^P(ds_i), \hso i=0, \ldots, n \label{MV_1_neq_1} \\
&  \Longrightarrow \hso  {\bf P}^{\pi^{A.L}}(ds_{i+1})= \int_{{\mathbb B}^{i-1}\times {\mathbb A}_{i-L}^i}  {\bf P}(ds_{i+1}|s_i, a_{i-L}^i)\otimes \pi_i^{A.L}(da_{i-L}^i|s_i)\otimes {\bf P}^{\pi^{A.L}}(ds_i) \label{MV_1_neq}
\end{align}
where (\ref{MV_1_neq}) is due to iterating (\ref{MV_1_neq_1}).
Finally, we can express  directed information (\ref{cor-ISR_25a_new_new}) as follows. 
\begin{align}
I(A^n\rightarrow B^n)=& \sum_{i=0}^n {\bf E}^{ P}\Big\{ \ell_i^P(A_i, A_{i-L}^{i-1},S^{i-1})\\
=& \sum_{i=0}^n \int_{{\mathbb B}^{i-1} \times {\mathbb A}_{i-L}^{i-1}}  \ell_i^{\pi^{A.L}}(a_i, a_{i-L}^{i-1},s^{i-1}) \pi_i^{A.L}(da_i|a_{i-L}^{i-1}, s_i)\otimes {\bf P}^{\pi^{A.L}}(da_{i-L}^{i-1}|s_i) \otimes {\bf P}^{\pi^{A.L}}(ds_i)\\
\equiv & 
\sum_{i=0}^n {\bf E}^{\pi^{A.L}}\Big\{ \ell_i^{\pi^{A.L}}(A_i, A_{i-L}^{i-1},B^{i-1})    \Big\}  \label{cor-ISR_25a_new_new_1}
\end{align}  
where the joint and marginal distributions are induced by the channel distribution $\{Q_i(db_i|b^{i-1}, a_{i-L}^i): i=0, \ldots, n\}$ and the control object $\{\pi_{i}^{A.L}(da_{i}|a_{i-L}^{i-1}, b^{i-1}): i=0, \ldots, n\}$. Clearly, from  (\ref{cor-ISR_25a_new_new_1}) we deduce that   the maximization of  directed information $
I(A^n\rightarrow B^n)$ defined  by  (\ref{cor-ISR_25a_new}) over  $\big\{ g_i(a^{i-1}, b^{i-1})\tri { P}_i(da_i| a^{i-1}, b^{i-1}): i=0, \ldots, n\big\}$, occurs in the subset, which satisfies conditional independence 
\bea
{ P}_i(da_i| a^{i-1}, b^{i-1})={\bf P}(da_i| a_{i-L}^{i-1}, b^{i-1})\equiv \pi_i^{A.L}(da_i|a_{i-L}^{i-1}, b^{i-1})-a.a.(a_{i-L}^{i-1}, b^{i-1}), \hso i=0, \ldots, n.
\eea 
Hence, $\{{S}_i: i=0, \ldots, n\}$ is    the controlled process controlled by the control process $\{A_i: i=0, \ldots, n\}$. However, the information structure of any candidate of optimal  channel input distribution (encoder) for each $i$, is $\{A_{i-L}^{i-1}, B^{i-1}\}$, while that of the decoder is $\{B^i\}$. Nevertheless, with the above simplification,  we can  further pursue the optimization of (\ref{cor-ISR_25a_new_new_1}) over the channel input distributions $\big\{\pi_i^{A.L}(da_i| a_{i-L}^{i-1}, b^{i-1}): i=0, \ldots, n\big\}$, with some variations from classical stochastic optimal control theory, as often  done in Markov decision theory, \cite{hernandezlerma-lasserre1996,kumar-varayia1986}.\\
We note that the above derivation is done from first principles, without utilizing any of the properties of Markov decision. In Theorem~\ref{cor-ISR} we provide an alternative derivation, which is based on identifying an augmented state process so that the above properties of optimal channel input distributions can be obtained, directly  from stochastic optimal control theory of Markov processes.

For the degenerate channel $Q_i(db_i|b^{i-1}, a_i, a_{i-1})=\overline{Q}_i(db_i|a_i, a_{i-1}), i=0, \ldots, n$, the derivation in \cite{yang-kavcic-tatikonda2005} (see Theorem~1) and  in   \cite{yang-kavcic-tatikonda2007},  should   be read with caution, because the authors do not show that the supremum of directed information over all channel input conditional distributions $\{P_i(da_i|a^{i-1}, b^{i-1}): i=0, \ldots, n\}$, occurs in a smaller set, satisfying a conditional independence $P_i(da_i|a^{i-1}, b^{i-1})=\overline{P}(da_i|a_{i-1}, b^{i-1}), i=0, \ldots, n$. Similarly, the derivation of Theorem~1, in  \cite{permuter-cuff-roy-weissman2010}, for the unifilar finite state channel of Figure 2 in \cite{permuter-cuff-roy-weissman2010}, i.e., $p(y_i|x_i, s_{i-1}), s_i=f(s_{i-1}, x_i,y_i)$,  where it is shown that directed information becomes $I(X^n\rar Y^n)=\sum_{i=0}^n I(X_i, S_{i-1}; Y_i|Y^{i-1})$, should be read with caution. Specifically,   the derivation given in  \cite{permuter-cuff-roy-weissman2010}, under ``Proof of Equality (18): It will suffice to prove
by induction that if we have two input distributions ...'', page 3154, is not equivalent to the statement that maximizing   $\sum_{i=0}^n I(X_i, S_{i-1}; Y_i|Y^{i-1})$ over ${\bf P}_{X_i|X^{i-1}, Y^{i-1}}:  i=1, \ldots, n$ occurs in the subset of distributions satisfying conditional independence $\big\{{\bf P}_{X_i|X^{i-1}, Y^{i-1}}={\bf P}_{X_i|S_{i-1}, Y^{i-1}}: i=1, \ldots, n\big\}$ (see feedback channels in \cite{cover-pombra1989}).
The derivations of Theorem~1 in \cite{yang-kavcic-tatikonda2005}, Theorem 1  in   \cite{yang-kavcic-tatikonda2007},  and  Theorem~1 in  \cite{permuter-cuff-roy-weissman2010}  need to incorporate the above steps, or variants of them (as done shortly), to fill the gaps of  showing   that the optimal channel input distributions for the specific channels considered by the authors, occur in subsets satisfying conditional independence.

{\bf Sufficient Statistic.} For Class A channel distribution, the resulting characterization of FTFI capacity  corresponds to the maximization problem 
\bea
\sup_{\pi_i(da_i|a_{i-L}^{i-1}, b^{i-1}): i=0, \ldots, n} \sum_{i=0}^n {\bf E}^{\pi^{A.L}}\Big\{ \ell_i^{\pi^{A.L}}(A_i, A_{i-L}^{i-1},B^{i-1})    \Big\}
\eea
 and this is expressed in terms of the \'a posteriori distribution  $\big\{{\bf P}^{\pi^{A.L}}(da_{i-L}^{i-1}|b^{i-1}): i=0, \ldots, n\big\}$   satisfying recursion (\ref{rec_g1}), (\ref{rec_g2}). However, it will be erroneous to assume that this  \'a posteriori distribution is a sufficient statistic for the channel input distribution $\big\{\pi_i^{A.L}(da_i|a_{i-L}^{i-1}, b^{i-1}): i=0, \ldots, n\big\}$, because it is not a Markov recursion \cite{kumar-varayia1986}. Rather, it is the joint process $\big\{({\bf P}^{\pi^{A.L}}(da_{i-L}^{i-1}|B^{i-1}), B^{i-1}): i=0, \ldots, n\big\}$  which is Markov.

Next,  we  provide an alternative more information theoretic  derivation based on the variational equalities of  Theorem~\ref{thm-var}, applied to the channel distribution (\ref{CD_C2}), that is, to (\ref{cor-ISR_25a_new}). For convenience of the reader  we introduce the following application of 
Theorem~\ref{thm-var} to channel distribution (\ref{cor-ISR_25a_new}).

\ \

\begin{theorem}
\label{cor-ISR_VE}
(Variational equalities for Class A  channels)\\
Consider the channel distribution of Class A (\ref{CD_C2}), i.e., $\{Q_i(db_i|b^{i-1}, A_{i-L}^i): i=0, \ldots, n\}$ and directed information $I(A^n \rar B^n)$ defined by (\ref{cor-ISR_25a_new}), via distributions (\ref{CIS_9a_new_1}), (\ref{CIS_9a_new}).\\
The following hold.\\
(a) Let ${\cal V}_{[0,n]} \tri \big\{ V_i(db_i|b^{i-1})\in {\cal M}({\mb B}_i): i=0, \ldots, n\big\}$ be an arbitrary set of distributions. Then 
\begin{align}
    & \sup_{ {\cal P}_{[0,n]}}  I(A^n\rightarrow B^n) = \sup_{ {\cal P}_{[0,n]}   }  \sum_{i=0}^n {\bf E}^{ P}\Big\{
\log\Big(\frac{dQ_i(\cdot|B^{i-1},A_{i-L}^i)}{d\Pi_i^{ P}(\cdot|B^{i-1})}(B_i)\Big)
\Big\}  \\
&    =\sup_{{\cal P}_{[0,n]}  } \inf_{ {\cal V}_{[0,n]}}  \sum_{i=0}^n {\bf E}^{ P}\Big\{
\log\Big(\frac{dQ_i(\cdot|B^{i-1},A_{i-L}^i)}{dV_i(\cdot|B^{i-1})}(B_i)\Big) \label{equation15a_VE_1}
\Big\}.
\end{align}
Moreover, the infimum over ${\cal V}_{[0,n]}$ is achieved at $V_i(db_i|b^{i-1})=\Pi_i^{ P}(db_i|b^{i-1}), i=0, \ldots, n$ given by (\ref{CIS_9a_new}).

(b) Let $\big\{S_i(db_i|b^{i-1},a^{i-1}) \in {\cal M}({\mb B}_i) : i=0, \ldots, n\big\}$ and $\big\{R_i(da_i|a^{i-1},b^i) \in {\cal M}({\mb A}_i): i=0,1,\ldots,n\big\}$ be arbitrary distributions and define the joint distribution on ${\cal M}({\mb A}_0^n \times {\mb B}_0^n)$ by $\otimes_{i=0}^n\Big(S_i(db_i|b^{i-1},a^{i-1})\otimes {R}_i(da_i|a^{i-1},b^i)\Big)$.  Then 
\begin{align}
&\sup_{ {\cal P}_{[0,n]} }  \sum_{i=0}^n {\bf E}^{ P}\Big\{
\log\Big(\frac{dQ_i(\cdot|B^{i-1},A_{i-L}^i)}{d\Pi_i^{ P}(\cdot|B^{i-1})}(B_i)\Big)
\Big\}   \\
 =& \sup_{ {\cal P}_{[0,n]} }  \sup_{\substack{\big\{S_i(db_i|b^{i-1},a^{i-1})\otimes{R}_i(da_i|a^{i-1},b^{i})\in{\cal M}({\mb A}_i\times{\mb B}_i): i=0,1,\ldots,n\big\}\\\big\{S_i(db_i|b^{i-1},a^{i-1})\in{\cal M}({\mb B}_i),~R_i(da_i|a^{i-1},b^{i})\in{\cal M}({\mb A}_i)\big\}}}  \sum^n_{i=0} {\bf E}^P\Big\{   \log\Big(\frac{d{R}_i(\cdot|A^{i-1},B^{i})}{dP_i(\cdot|A^{i-1},B^{i-1})}(A_i)\nonumber \\
&.\frac{dS_i(\cdot|B^{i-1},A^{i-1})}{d\Pi_{i}^P(\cdot|B^{i-1})}(B_i)\Big) \Big\}   \label{equation15a_VE_2}
\end{align}
Moreover, the supremum over $\big\{S_i(db_i|b^{i-1},a^{i-1})\otimes{R}_i(da_i|a^{i-1},b^{i})\in{\cal M}({\mb A}_i\times{\mb B}_i): i=0,1,\ldots,n\big\}$ is achieved when the following identity holds. 
\begin{align}
\frac{dP_i(\cdot|a^{i-1},b^{i-1})}{d{R}_i(\cdot|a^{i-1},b^i)}(a_i).\frac{d{Q}_i(\cdot|b^{i-1},a_{i-L}^{i})}{dS_i(\cdot|b^{i-1},a^{i-1})}(b_i)=1-a.a.(a^i,b^i),~i=0,1,\ldots,n.\label{equation102_VE_3}
\end{align}
Equivalently, the supremum  is achieved at 
\bea
\otimes_{i=0}^n\Big(S_i(db_i|b^{i-1},a^{i-1})\otimes {R}_i(da_i|a^{i-1},b^i)\Big)= \otimes_{i=0}^n\Big(P_i(da_i|a^{i-1},b^{i-1})\otimes {Q}_i(db_i|b^{i-1}, a_{i-L}^i)\Big).
\eea
\end{theorem}

\begin{proof} (a), (b) These are applications of 
Theorem~\ref{thm-var} to the specific channel, hence the derivations are omitted. 
\end{proof}

Next, we apply the variational equalities of Theorem~\ref{cor-ISR_VE} and  stochastic optimal control theory, to identify the information structure  of the optimal channel input conditional distribution, which maximizes (\ref{cor-ISR_25a_new}) over ${\cal P}_{[0,n]}$,  without and with a  transmission transmission cost constraint of Class A.\\

\begin{theorem}
\label{cor-ISR}
(Class A  channels and class A transmission cost functions)\\
Suppose the channel distribution is of Class A defined by (\ref{CD_C2}), i.e.,  
\bea
{\bf P}_{B_i|B^{i-1}, A^i}(db_i|b^{i-1}, a^i)=Q_i(db_i|b^{i-1}, a_{i-L}^i)-a.a. (b^{i-1}, a^i), \hso i=0, \ldots, n. \label{cor-ISR_29}
\eea
%Define the following restricted class of channel input distributions.
%\begin{align}
%\overline{\cal P}_{[0,n]}^{A.L} \tri     \Big\{   P_i(da_i | a^{i-1}, b^{i-1})= \pi_i^{A.L}(da_i|a_{i-L}^{i-1}, b^{i-1})-a.a. (a^{i-1}, b^{i-1}): i=0, 1, \ldots, n\Big\} \subset {\cal P}_{[0,n]}  \label{FRD_1}
%\end{align}
The following hold.\\
(a) Without Transmission Cost. The  maximization of $I(A^n \rar B^n)$ defined by (\ref{cor-ISR_25a_new}) over ${\cal P}_{[0,n]}$   occurs in $\overline{\cal P}_{[0,n]}^{A.L}$ defined by (\ref{FRD_1}) and the characterization of FTFI capacity is given by the following expression.
\begin{align}
{C}_{A^n \rar B^n}^{FB,A.L} = \sup_{\big\{\pi_i^{A.L}(da_i |a_{i-L}^{i-1}, b^{i-1}) \in {\cal M}({\mb A}_i) : i=0,\ldots,n\big\}} \sum_{i=0}^n {\bf E}^{ \pi^{A.L}}\Big\{
\log\Big(\frac{dQ_i(\cdot|B^{i-1},A_{i-L}^i)}{d\Pi_i^{ \pi^{A.L}}(\cdot|B^{i-1})}(B_i)\Big)
\Big\}  \label{cor-ISR_25a_a_c}
\end{align}
where 
\begin{align}
  \Pi_i^{\pi^{A.L}}(db_i | b^{i-1}) =& \int_{  {\mb A}_{i-L}^{i} }   Q_i(db_i |b^{i-1}, a_{i-L}^i) \otimes   {\pi}_i^{A.L}(da_i | a_{i-L}^{i-1}, b^{i-1})  \otimes {\bf P}^{\pi^{A.L}}(da_{i-L}^{i-1} | b^{i-1})  , \label{cor-ISR_31_c}\\
  {\bf P}^{\pi^{A.L}}(da^i,  d b^{i})=& \otimes_{i=0}^n \Big( Q_i(db_i|b^{i-1}, a_{i-L}^i) \otimes \pi_i^{A.L}(da_i|a_{i-L}^{i-1}, b^{i-1})\Big), \hso  i=0, \ldots, n \label{cor-ISR_32_c} \\
{\bf P}^{\pi^{A.L}}(da_{i-L}^{i-1}|b^{i-1}) =& \tilde{T}_{i-1}\Big(b_{i-1}, b^{i-2}, \pi_{i-1}^{A.L}(\cdot|\cdot, b^{i-2}), {\bf P}^{\pi^{A.L}}(\cdot|b^{i-2}) \Big)(da_{i-L}^{i-1}),\hso i=1, \ldots, n, \label{A_POST_1_n}\\
{\bf P}^{\pi^{A.L}}(da_{-L}^{-1}|b^{-1})=&\mbox{given}  \label{A_POST_1_nn}
\end{align}
 and the initial data ${\cal I}_0^P$  are specified by ${\cal I}_0^P\tri \{A_{-L}^{-1}, B^{-1}\}$ (or any other convention, ie.,  ${\cal I}_0^P\tri \{B^{-1}\}$). \\
(b) With Transmission Cost. Consider the average transmission cost constraint defined by (\ref{cor-ISR_29_cc_CC})   
%\begin{align}
%{\cal P}_{[0,n]}^{A}(\kappa) \tri     \Big\{   P_i(da_i | a^{i-1}, b^{i-1}),  i=0, 1, \ldots, n:   \frac{1}{n+1} {\bf E}^{P}\Big(\sum_{i=0}^n \gamma_i^{A.N}(A_{i-N}^i, B^{i})\Big) \leq \kappa  \Big\}\subset {\cal P}_{[0,n]} \label{cor-ISR_29_cc_CC}
%\end{align}
and suppose the following condition holds.
\begin{align}
%&(i) \hso  \gamma_i(T^ia^n,T^ib^{n-1})=\gamma_i^{A.N}(a_{i-N}^i, b^{i}), \hso i=0, \ldots, n; \label{cor-ISR_29_c} \\
\sup_{ {\cal P}_{[0,n]}^A(\kappa) }& I(A^{n}
\rar {B}^{n}) = \inf_{\lambda\geq 0} \sup_{ \big\{ P_i(da_i|a^{i-1}, b^{i-1}): i=0,\ldots, n\big\}\in {\cal P}_{[0,n]}   } \Big\{ I(A^{n}
\rar {B}^{n})  \nonumber \\
&\hst \hst - \lambda \Big\{ {\bf E}^P\Big(\sum_{i=0}^n \gamma_i^{A.N}(A_{i-N}^i,B^{i})\Big)-\kappa(n+1) \Big\}\Big\} \label{ISDS_6cc}
\end{align}  
where  $\lambda$ is the Lagrange multiplier associated with the transmission cost constraint.\\
The  maximization of $I(A^n \rar B^n)$ defined by (\ref{cor-ISR_25a_new}) over ${\cal P}_{[0,n]}^A(\kappa)$    occurs in the subset $\overline{\cal P}_{[0,n]}^{A.I}(\kappa)$ defined by (\ref{cor-ISR_29_cc}) 
%\begin{align}
%\overline{\cal P}_{[0,n]}^{A.I}(\kappa) \tri     \Big\{   & P_i(da_i | a^{i-1}, b^{i-1})= \pi_i^{A.I}(da_i|a_{i-I}^{i-1}, b^{i-1})-a.a. (a^{i-1}, b^{i-1}),  i=0, 1, \ldots, n:  \nonumber \\
%& \frac{1}{n+1} {\bf E}^{\pi^{A.I}}\Big(\sum_{i=0}^n \gamma_i^{A.N}(A_{i-N}^i, B^{i})\Big) \leq \kappa  \Big\}\subset {\cal P}_{[0,n]}^A(\kappa), \hso I \tri \max\{L, N\}. \label{cor-ISR_29_cc}
%\end{align}
and the characterization of FTFI capacity is given by the following expression.
\begin{align}
{C}_{A^n \rar B^n}^{FB,A.I}
= \sup_{\big\{\pi_i^{A.I}(da_i |a_{i-I}^{i-1}, b^{i-1}), i=0,\ldots,n: \frac{1}{n+1} {\bf E}^{\pi^{A.I}}\Big( \sum_{i=0}^n \gamma_{i}^{A.N}(A_{i-N}^i, B^{i})\Big) \leq \kappa   \big\}} \sum_{i=0}^n {\bf E}^{ \pi^{A.I}}\Big\{
\log\Big(\frac{dQ_i(\cdot|B^{i-1},A_{i-L}^i)}{d\Pi_i^{ \pi^{A.I}}(\cdot|B^{i-1})}(B_i)\Big)
\Big\}  \label{cor-ISR_25a_a_c_N}
\end{align}
 where the joint and marginal distributions  are 
% induced by  $\{Q_i(db_i|a_{i-L}^{i-1}, b^{i-1}), {\pi }_i^{A.I}(da_{i}|a_{i-I}^{i-1},b^{i-1}): i=0, \ldots, n\big\}$, as follows.
given by 
\begin{align}
  \Pi_i^{\pi^{A.I}}(db_i | b^{i-1}) =& \int_{  {\mb A}_{i-L}^{i} }   Q_i(db_i |b^{i-1}, a_{i-L}^i) \otimes   {\pi}_i^{A.I}(da_i | a_{i-I}^{i-1}, b^{i-1}) \otimes {\bf P}^{\pi^{A.I}}(da_{i-I}^{i-1}|  b^{i-1}), \label{cor-ISR_31_c_B}\\
  {\bf P}^{\pi^{A.I}}(da^i,  d b^{i})=& \otimes_{i=0}^n\Big( Q_i(db_i|b^{i-1}, a_{i-L}^i) \otimes \pi_i^{A.I}(da_i|a_{i-I}^{i-1}, b^{i-1})\Big), \hso  i=0, \ldots, n \label{cor-ISR_32_c_B} 
 \end{align} 
 and the \^a posteriori distribution satisfies a recursion similar to (\ref{A_POST_1_n}), (\ref{A_POST_1_nn}), and  the initial data are specified by the convention used.
\end{theorem}

\begin{proof} First, we show the pay-off is a functional of a certain process, called the state process and then we show that the state process is Markov given the past values of the state process and the past values of the channel inputs. Basically, we re-formulate the optimization problem so that the state and control processes satisfy  Feature 1, (i), (ii),  discussed below  (\ref{SOC_1}).  \\
(a) Recall (\ref{cor-ISR_25a_new}). By applying the  re-conditioning property of expectation, we obtain the following identities. 
 \begin{align}
I(A^n\rar B^n)  
 = &\sum_{i=0}^{n}  {\bf E}^{ P} \bigg\{    \log\bigg(\frac{dQ_i(\cdot| B^{i-1}, A_{i-L}^i)}
  {d\Pi_i^{ P}(\cdot | B^{i-1})}(B_i)\bigg)\bigg\}\label{CIS_6c_a}\\
 =& \sum_{i=0}^{n}  {\bf E}^{ P} \bigg\{  {\bf E}^{ P}\Big\{\log \Big(\frac{dQ_i(\cdot| B^{i-1},A_{i-L}^i)}
  {\Pi_i^{ P}(\cdot| B^{i-1})}(B_i)\Big)\bigg{|}A^i, B^{i-1} \Big\}\bigg\}\label{CIS_6d_a}   \\
  =& \sum_{i=0}^{n}  {\bf E}^{ P}\bigg\{   {\bf E}^{ P}\Big\{\log \Big(\frac{dQ_i(\cdot| B^{i-1},A_{i-L}^i)}
  {d\Pi_i^{ P}(\cdot| B^{i-1})}(B_i)\Big)\bigg{|}A_{i-L}^i, B^{i-1}\Big\}\bigg\}\label{CIS_6e_a}\\
=& \sum_{i=0}^{n} {\bf E}^{ P}\bigg\{  \ell_i^{ P}\big(A_i, \overline{S}_i\big)\bigg\}, \hst \overline{S}_i \tri (A_{i-L}^{i-1}, B^{i-1}), \hso S_i\tri B^{i-1}, \; i=0, \ldots, n,  \label{CIS_6f_a}\\
 \ell_i^{ P}\big(a_i, \overline{s}_i \big) \equiv & \ell_i^{ P}\big(a_i,a_{i-L}^{i-1}, s_{i} \big) = \int_{ {\mb  B}_i } \log \Big(\frac{Q_i(db_i| \overline{s}_i, a_i)}
  {\Pi_i^{ P}(db_i| s_{i})}\Big) Q_i(db_i| \overline{s}_i, a_i), \hso i=0, \ldots, n \label{BAM31_a}
  \end{align}
where (\ref{CIS_6e_a}) is due to the channel conditional independence property (\ref{cor-ISR_29}).
Hence, 
\begin{align}
\sup_{ {\cal P}_{[0,n]}}  \sum_{i=0}^n {\bf E}^{ P}\Big\{
\log\Big(\frac{dQ_i(\cdot|B^{i-1},A_{i-L}^i)}{d\Pi_i^{ P}(\cdot|B^{i-1})}(B_i)\Big)
\Big\}
 =\sup_{ {\cal P}_{[0,n]}    }\sum_{i=0}^{n} {\bf E}^{ P}\bigg\{  \ell_i^{P}\Big(A_i, \overline{S}_i\Big)\bigg\}. \label{SC_1A}
\end{align}
The pay-off functional $\big\{\ell_i^{ P}(a_i, \overline{s}_i): i=0, \ldots, n\big\}$ defined by (\ref{BAM31_a}) depends on $\big\{\overline{s}_i \tri (a_{i-L}^{i-1},b^{i-1}): i=0, 1, \ldots,n\big\}$, called the state process, via the channel distribution dependence on these variables, and the control object  $\big\{g_i(a^{i-1}, b^{i-1})\tri P_i(da_i|a^{i-1}, b^{i-1}): i=0, \ldots, n\big\}$, via  $\{\xi_i^P(b^{i-1})\tri \Pi_i^P(db_i|b^{i-1}): i=0, \ldots, n\}$. \\
Next, we give a different derivation than the one given earlier.
%, which was  based on the process $\{S_i: i=0, \ldots, n\}$.
For  each $i$, we can easily show, using Bayes' theorem, and  the property of the channel distribution, that  the conditional distribution of the state $\overline{S}_{i+1}$ given $\{\overline{S}^{i}, A^i\}$  is Markov, i.e.,  the following conditional independence holds. 
\begin{align}
{\bf P}(d\overline{s}_{i+1}|\overline{s}^i, a^i)= {\bf P}(d\overline{s}_{i+1}|\overline{s}_i, a_i), \hso i=0, \ldots, n-1. \label{ext_MP1}
\end{align}
Hence, $\Big\{{\bf P}(d\overline{s}_{i+1}|\overline{s}_i, a_i): i=0, \ldots, n-1\Big\}$ is the controlled object, i.e.,   $\{\overline{S}_i:  i=0, \ldots, n\}$ is the  controlled process,  control by the control process $\{A_i: i=0, \ldots, n\}$.\\
Note that if the pay-off function in  (\ref{SC_1A}), i.e., $\ell_i^{P}\Big(a_i, \overline{s}_i\Big)$, was fixed and independent of the channel input distribution $\{P_j(da_j|a^{j-1}, b^{j-1}): j=0, \ldots, i\}$, for $i=0, \ldots, n$,  then in view of the Markov property (\ref{ext_MP1}), it follows directly from  stochastic optimal control theory  \cite{hernandezlerma-lasserre1996} or \cite{kumar-varayia1986}, that the maximizing distribution occurs is the subset satisfying conditional independence $P_i(da_i|a^{i-1}, b^{i-1})={\bf P}(da_i|\overline{s}_i)\equiv \pi_i^{A.L}(da_i|a_{i-L}^{i-1}, b^{i-1})-a.a.(a^{i-1}, b^{i-1}), i=0, \ldots, n$. However, the dependence of the pay-off $\ell_i^{P}\Big(a_i, \overline{s}_i\Big)$ on $\{P_j(da_j|a^{j-1}, b^{j-1}): j=0, \ldots, i\}$ prevents us from using, directly stochastic optimal control theory, to establish this claim. \\
However,  we can bypass this technicality, by invoking the variational equalities of Theorem~\ref{cor-ISR_VE} to obtain  achievable upper bounds, when  the optimal control object satisfies $\big\{g_i(a^{i-1}, b^{i-1}) ={g}_i^{A.L}(a_{i-L}^{i-1}, b^{i-1})\equiv  \pi_i^{A.L}(da_i|a_{i-L}^{i-1}, b^{i-1}): i=0, \ldots, n\big\}$.
\begin{comment}
Using the definition of the  channel we obtain the following identity. 
\begin{align}
{\bf P}(ds_{i+1}, da_{i-L+1}^i|s^i, a_{i-L}^{i-1}) = {\bf P}(d s_{i+1}, da_{i-L+1}^i|s_i, a_{i-L}^{i-1}) , \hst  i=0, \ldots, n-1. \label{AN3_thm_9a}
\end{align}
\end{comment}
  Consider the set of arbitrary distributions  ${\cal V}_{[0,n]} \tri \big\{ V_i(db_i|b^{i-1})\in {\cal M}_i({\mb B}_i): i=0, \ldots, n\big\}$ and define the pay-off function 
\begin{align}
\ell_i^{V}\big(a_i, \overline{s}_i \big)  = \int_{ {\mb  B}_i } \log \Big(\frac{Q_i(db_i| \overline{s}_i, a_i)}
  {V_i(db_i| s_{i})}\Big) Q_i(db_i| \overline{s}_i, a_i), \hso i=0, \ldots, n.\label{equation15a_VE_300_b}
\end{align}
By virtue of  (\ref{equation15a_VE_1}), identity (\ref{SC_1A}), and inequality $\sup \inf \{ \cdot \} \leq \inf \sup \{\cdot\}$ we obtain the following upper bound.
\begin{align}
\sup_{ {\cal P}_{[0,n]}   }  \sum_{i=0}^n {\bf E}^{ P}\Big\{
\log\Big(\frac{dQ_i(\cdot|B^{i-1},A_{i-L}^i)}{d\Pi_i^{ P}(\cdot|B^{i-1})}&(B_i)\Big)
\Big\}   = \sup_{ {\cal P}_{[0,n]}    } \inf_{{\cal V}_{[0,n]}} \sum_{i=0}^{n} {\bf E}^{ P}\Big\{  \ell_i^{V}\Big(A_i, \overline{S}_i\Big)\Big\} \label{von_n_1}  \\
\leq & \inf_{{\cal V}_{[0,n]}}  \sup_{ {\cal P}_{[0,n]}    } \sum_{i=0}^{n} {\bf E}^{ P}\Big\{  \ell_i^{V}\Big(A_i, \overline{S}_i\Big)\Big\} \label{von_n_2} \\
\leq & \sup_{ {\cal P}_{[0,n]} } \sum_{i=0}^{n} {\bf E}^{ P}\Big\{  \ell_i^{V}\big(A_i, \overline{S}_i\big)\Big\}, \hso \forall \; V_i(db_i|b^{i-1})\in {\cal M}({\mb B}_i),  i=0, \ldots, n  \label{equation15a_VE_300}
\end{align}
Since the pay-off functions $\big\{\ell_i^{V}\big(a_i, \cdot\big): i=0, \ldots, n\big\}$ depend on $\big\{ \overline{s}_i=(a_{i-L}^{i-1}, b^{i-1}): i=0, \ldots, n\big\}$ (and  also $\zeta_i(b^{i-1})\tri V_i(db_i|b^{i-1})\in {\cal M}({\mb B}_i)$, whose information is already included in $\overline{s}_i$),  and the controlled object is Markov, i.e., (\ref{ext_MP1}) holds, 
then by  making use of recursions (\ref{rec_g1})-(\ref{rec_g3}) or applying the standard results of Markov decision of  stochastic optimal control theory  \cite{hernandezlerma-lasserre1996}, then the maximizing distribution in the right hand side of   (\ref{equation15a_VE_300}) occurs in the set $\overline{\cal P}_{[0,n]}^{A.L}$, 
 defined by (\ref{FRD_1}). Hence,  the following upper bound  is obtained.
\begin{align}
\sup_{ {\cal P}_{[0,n]}   }  \sum_{i=0}^n {\bf E}^{ P}\Big\{
\log\Big(\frac{dQ_i(\cdot|B^{i-1},A_{i-L}^i)}{d\Pi_i^{ P}(\cdot|B^{i-1})}(B_i)\Big)
\Big\} 
\leq & \sup_{\big\{\pi_i^{A.L}(da_i |a_{i-L}^{i-1}, b^{i-1}): i=0,\ldots,n\big\}} \sum_{i=0}^{n} {\bf E}^{ \pi^{A.L}}\Big\{  \ell_i^{V}\big(A_i, \overline{S}_i\big)\Big\},    \label{equation15a_VE_30_d}\\
\forall & \; V_i(db_i|b^{i-1})\in {\cal M}({\mb B}_i), i=0, \ldots, n \nonumber
\end{align}
where ${\bf E}^{\pi^{A.L}}$ means expectation with respect to joint distribution (\ref{cor-ISR_32_c}). 
Next, we evaluate the upper bound  (\ref{equation15a_VE_30_d}) at   $V_i(db_i|b^{i-1}) = \Pi_i^{\pi^{A.L}}(db_i | b^{i-1}), i=0, \ldots, n$,  defined by  (\ref{cor-ISR_31_c}), which implies 
\begin{align}
\ell_i^{V}\big(a_i, \overline{s}_i \big)\Big|_{V=\Pi^{\pi^{A.L}}} \equiv  \ell_i^{ \pi^{A.L}}\big(a_i,a_{i-L}^{i-1}, s_{i} \big) \tri \int_{ {\mb  B}_i } \log \Big(\frac{Q_i(db_i| \overline{s}_i, a_i)}
  {\Pi_i^{ \pi^{A.L}}(db_i| s_{i})}\Big) Q_i(db_i| \overline{s}_i, a_i), \hso i=0, \ldots, n \label{equation15a_VE_30_b}
\end{align}
to obtain the following upper bound.
\begin{align}
&\sup_{ {\cal P}_{[0,n]}   }  \sum_{i=0}^n {\bf E}^{ P}\Big\{
\log\Big(\frac{dQ_i(\cdot|B^{i-1},A_{i-L}^i)}{d\Pi_i^{ P}(\cdot|B^{i-1})}(B_i)\Big)
\Big\} 
\leq\sup_{\big\{\pi_i^{A.L}(da_i |a_{i-L}^{i-1}, b^{i-1}): i=0,\ldots,n\big\}      } \sum_{i=0}^{n} {\bf E}^{ \pi^{A.L}}\Big\{  \ell_i^{\pi^{A.L}}\big(A_i, \overline{S}_i\big)\Big\}   \label{equation15a_VE_30} \\
&  = \sup_{\big\{\pi_i^{A.L}(da_i |a_{i-L}^{i-1}, b^{i-1}): i=0,\ldots,n\big\}} \sum_{i=0}^n {\bf E}^{ \pi^{A.L}}\Big\{
\log\Big(\frac{dQ_i(\cdot|B^{i-1},A_{i-L}^i)}{d\Pi_i^{ \pi^{A.L}}(\cdot|B^{i-1})}(B_i)\Big)
\Big\} \equiv        {C}_{A^n \rar B^n}^{FB,A.L}. \label{equation102_VE_31}
\end{align}
Note that any other choice of $V_i(db_i|b^{i-1})$, other than $V_i(db_i|b^{i-1}) = \Pi_i^{\pi^{A.L}}(db_i | b^{i-1}), i=0, \ldots, n$ will not be consistent with the joint distribution induced by the channel distribution and $\{\pi_i^{A.L}(da_i |a_{i-L}^i, b^{i-1}): i=0, \ldots, n\}$, i.e., the distribution over which the expectation is taken in (\ref{equation15a_VE_30_d}). \\
The reverse inequality can be shown  by restricting the maximization in (\ref{SC_1A}) to the subset $\overline{\cal P}_{[0,n]}^{A.L} \subset {\cal P}_{[0,n]}$, which then implies the joint and transition probability distribution  of the channel output process are given by 
(\ref{cor-ISR_31_c}) and (\ref{cor-ISR_32_c}), and consequently the reverse inequality is obtained.\\
 We can also show the reverse inequality via an application of  variational equality (\ref{equation15a_VE_2}). We do so to illustrate the power of variational equalities. By virtue of (\ref{equation15a_VE_2}), and  by removing the supremum over $\big\{S_i(db_i|b^{i-1}, a^{i-1})\otimes R_i(da_i|a^{i-1}, b^i): i=0, \ldots, n\big\}$, and setting  
\begin{align}
S_i(db_i|b^{i-1},a^{i-1})\otimes& {R}_i(da_i|a^{i-1},b^i)= \frac{dP_i(\cdot|a^{i-1},b^{i-1})}{ d\pi^{A.L}(\cdot|a_{i-L}^{i-1}, b^{i-1})}(a_i) \frac{ d\Pi^P(\cdot|b^{i-1})}{d\Pi_i^{\pi^{A.L}}(\cdot|b^i)}(b_i) \nonumber \\
&.  \pi_i^{A.L}(da_i|a_{i-L}^{i-1}, b^{i-1}) \otimes dQ_i(db_i|b^{i-1},a_{i-L}^{i-1}), \hso a.s.,~i=0,1,\ldots,n.\label{equation102_VE_3}
\end{align}
where $\big\{\Pi_i^{\pi^{A.L}}(b_i|b^i): i=0, \ldots, n\big\}$ is given by  (\ref{cor-ISR_31_c}), then  the following  lower  bound is obtained.
\begin{align}
 \sup_{  {\cal P}_{[0,n]}   }  \sum_{i=0}^n {\bf E}^{ P}\Big\{
\log\Big(\frac{dQ_i(\cdot|B^{i-1},A_{i-L}^i)}{d\Pi_i^{ P}(\cdot|B^{i-1})}(B_i)\Big)
\Big\}  \geq \sup_{ {\cal P}_{[0,n]} } \sum_{i=0}^{n} {\bf E}^{ P}\Big\{  \ell_i^{\pi^{A.L}}\Big(A_i, \overline{S}_i\Big)\Big\}    .  \label{equation15a_VE_40}
\end{align}
Since for each $i$, the pay-off $\ell_i^{\pi^{A.L}}\Big(a_i,\cdot\Big)$ depends on $\overline{s}$, for $i=0, \ldots, n$, and the controlled object is Markov, i.e., (\ref{ext_MP1}) holds, then from Markov decision theory \cite{kumar-varayia1986}, or by  making use of recursions (\ref{rec_g1})-(\ref{rec_g3}), the supremum in the right hand side of (\ref{equation15a_VE_40}) occurs in $\overline{\cal P}_{[0,n]}^{A.L}$ and   the following lower bound is obtained. 
\begin{align}
\sup_{ {\cal P}_{[0,n]}}  \sum_{i=0}^n {\bf E}^{ P}\Big\{
\log\Big(\frac{dQ_i(\cdot|B^{i-1},A_{i-L}^i)}{d\Pi_i^{ P}(\cdot|B^{i-1})}(B_i)\Big)
\Big\} 
\geq {C}_{A^n \rar B^n}^{FB,A.L}. \label{equation102_VE_5}
\end{align}

 Combining (\ref{equation102_VE_31}) and (\ref{equation102_VE_5}) we establish the claims in (a). \\
(b)  Since by condition (\ref{ISDS_6cc}), the constraint problem is equivalent to an unconstraint problem,  we repeat the steps in (a), for   the augmented pay-off given by the following expression. 
\begin{align}
I(A^n \rar B^n)-\lambda  {\bf E}^P\Big(\sum_{i=0}^n \gamma_i^{A.N}(A_{i-N}^i,B^{i})\Big)= \sum_{i=0}^n {\bf E}^P \Big\{ 
\log\Big(\frac{dQ_i(\cdot|B^{i-1}, A_{i-L}^i)}{d\Pi_i^{ P}(\cdot|B^{i-1})}(B_i)\Big) -\lambda \gamma_i^{A.N}(A_{i-N}^i,B^{i}) \Big\}. \label{aug_CC_AA}
\end{align}
Note that the term $\lambda  (n+1)\kappa$ is not included, because it does not affect the derivation of information structures of optimal channel input conditional distribution. Similarly as in the unconstraint case, we have the following.
\begin{align}
&I(A^n \rar B^n)-\lambda  {\bf E}^P\Big(\sum_{i=0}^n \gamma_i^{A.N}(A_{i-N}^i,B^{i})\Big)\label{CIS_6d_a_CC_AAA}  \\
 =& \sum_{i=0}^{n}  {\bf E}^{ P} \bigg\{  {\bf E}^{ P}\Big\{\log \Big(\frac{dQ_i(\cdot| B^{i-1},A_{i-L}^i)}
  {\Pi_i^{ P}(\cdot| B^{i-1})}(B_i)\Big) -\lambda \gamma_i^{A.N}(A_{i-N}^i,B^{i})       \bigg{|}A^i, B^{i-1}    \Big\}\bigg\}\label{CIS_6d_a_CC_AA}   \\
%  =& \sum_{i=0}^{n}  {\bf E}^{ P}\bigg\{   {\bf E}^{ P}\Big\{\log \Big(\frac{dQ_i(\cdot| B^{i-1},A_{i-L}^i)}
%  {d\Pi_i^{ P}(\cdot| B^{i-1})}(B_i)\Big)\bigg{|}A_{i-L}^i, B^{i-1}\Big\}   -s  \gamma_i^{A.N}(A_{i-N}^i,B^{i})        \bigg\}\label{CIS_6e_a_CC_AA}\\
%=& \sum_{i=0}^{n} {\bf E}^{ P}\bigg\{  \ell_i^{ P}\big(A_i,A_{i-L}^{i-1}, B^{i-1}\big) -\lambda  \gamma_i^{A.N}(A_{i-N}^i,B^{i-1})        \bigg\} \label{CIS_6f_a_CC_AA}\\
=& \sum_{i=0}^{n} {\bf E}^{ P}\bigg\{  \overline{\ell}_i^{ P}\big(A_i,\widehat{S}_i\big) \bigg\}, \hst \widehat{S}_i \tri (A_{i-I}^{i-1}, B^{i-1}),\hso \overline{S}_i\tri (A_{i-L}^{i-1},B^{i-1}),  \hso S_i\tri B^{i-1}, \hso I\tri \max\{L, N\}\; i=0, \ldots, n,  \label{CIS_6f_a_CC_N_AA}
\end{align}
where 
\begin{align}
 \overline{\ell}_i^{ P}\big(a_i, \widehat{s}_i \big)  = \int_{ {\mb  B}_i } \Big[ \log \Big(\frac{Q_i(db_i|\overline{s}_i, a_{i})}
  {\Pi_i^{ P}(db_i|s_i)}\Big)- \lambda  \gamma_i^{A.N}(A_{i-N}^i,B^{i})    \Big]Q_i(db_i|\overline{s}_{i}, a_i), \hso i=0, \ldots, n. \label{BAM31_a_CCC_AA}
  \end{align}
Note that unlike  part (a), the augmented pay-off function $\big\{\overline{\ell}_i^{ P}\big(a_i, \cdot\big): i=0, \ldots, n\big\}$ defined by (\ref{BAM31_a_CCC_AA})  depends on  $\big\{\widehat{s}_i=(a_{i-I}^{i-1},b^{i-1}): i=0, \ldots, n\big\}$, via the channel and cost function,  and that if $I= L$ then $\hat{s}_i=\overline{s}_i, i=0, \ldots, n$ (same as in (a)).  \\
It is easy to verify that the variational equalities of Theorem~\ref{thm-var} (see Theorem~IV.1 in \cite{charalambous-stavrou2013aa}) are also valid, when transmission cost constraints are imposed. Thus, Theorem~\ref{cor-ISR_VE} holds, with the supremum over    ${\cal P}_{[0,n]}$ replaced by  ${\cal P}_{[0,n]}^A(\kappa)$. Note that if $I=L$, the optimal channel input distribution has exactly the same  form as in  part (a).\\
Keeping in mind the dependence, for each $i$, of  the unconstraint pay-off  function $\overline{\ell}_i^{ P}(a_i, \cdot)$  on  $\widehat{s}_i=(a_{i-I}^{i-1},b^{i-1})$, for $i=0, \ldots, n$,  we    repeat the derivation of the upper bound in (a), by invoking the first variational equality, and then  remove the infimum, to deduce
\begin{align}
\sup_{  {\cal P}_{[0,n]}^A(\kappa)}  \sum_{i=0}^n {\bf E}^{ P}\Big\{
\log\Big(\frac{dQ_i(\cdot|B^{i-1},A_{i-L}^i)}{d\Pi_i^{ P}(\cdot|B^{i-1})}(B_i)\Big)
\Big\}  \leq\sup_{  {\cal P}_{[0,n]}^A(\kappa)} \sum_{i=0}^{n} {\bf E}^{ P}\Big\{  \ell_i^{V}\big(A_i,\overline{S}_i\big)\Big\}  . \label{equation15a_VE_30dd_new}
\end{align}
Further, by setting  $V_i(db_i|b^{i-1}) = \Pi_i^{\pi^{A.I}}(db_i | b^{i-1}), i=0, \ldots, n$  defined by  (\ref{cor-ISR_31_c_B}),  
%then  the following analog of the upper bounds (\ref{equation15a_VE_30}), (\ref{equation102_VE_31}) are obtained. 
%\begin{align}
%\sup_{  {\cal P}_{[0,n]}^A(\kappa)}  \sum_{i=0}^n {\bf E}^{ P}\Big\{
%\log\Big(\frac{dQ_i(\cdot|B^{i-1},A_{i-L}^i)}{d\Pi_i^{ P}(\cdot|B^{i-1})}(B_i)\Big)
%\Big\}  \leq\sup_{  {\cal P}_{[0,n]}^A(\kappa)} \sum_{i=0}^{n} {\bf E}^{ P}\Big\{  \ell_i^{\pi^{A.I}}\big(A_i,\overline{S}_i\big)\Big\}   \label{equation15a_VE_30dd}
%\end{align}
%where 
%\begin{align}
%\ell_i^{\pi^{A.I}}\big(a_i, \overline{s}_i \big) = \int_{ {\mb  B}_i } \log \Big(\frac{Q_i(db_i|\overline{s}_i, a_i)}
%  {\Pi_i^{ \pi^{A.I}}(db_i| s_{i})}\Big) Q_i(db_i|\overline{s}_i, a_i), \hso  \Pi_i^{ \pi^{A.I}}(db_i| s_{i})=(\ref{cor-ISR_31_c_B}), \hso  i=0, \ldots, n.\label{equation15a_VE_30_bc}
%\end{align}
%Further, by stochastic optimal control theory 
the maximization of the right hand side of  (\ref{equation15a_VE_30dd_new})  occurs in the subset $\overline{\cal P}_{[0,n]}^{A.I}(\kappa) \subseteq {\cal P}_{[0,n]}^{A}(\kappa)$, hence the following upper bound. 
\begin{align}
\sup_{ {\cal P}_{[0,n]}^{A}(\kappa)   }  \sum_{i=0}^n {\bf E}^{ P}\Big\{
\log\Big(\frac{dQ_i(\cdot|B^{i-1},A_{i-L}^i)}{d\Pi_i^{ P}(\cdot|B^{i-1})}(B_i)\Big)
\Big\}  \leq {C}_{A^n \rar B^n}^{FB,A.I}(\kappa). \label{equation102_VE_31_dd}
\end{align}
From this point forward, by repeating the  derivation of part (a), if necessary, it is easy to deduce that  the information structure of the channel input distribution, which maximizes directed information, for each $i$, is ${\cal I}_i^P=\{a_{i- I}^{i-1}, b^{i-1}\}$, for $i=0,1, \ldots, n$, which then  implies (\ref{cor-ISR_25a_a_c_N})-(\ref{cor-ISR_32_c_B}).\\
We note that if $N >L$, the upper bound    corresponds to $\{\Pi_i^{\pi^{A.N}}(db_i | b^{i-1}): i=0, \ldots, n\}$ defined by (\ref{cor-ISR_31_c_B}), (with $I=N$), which  depends on the channel input conditional distribution $\{\pi_i^{A.N}(da_i|a_{i-N}^{i-1}, b^{i-1}): i=0, \ldots, n\}$. 
%However,  if instead another distribution is used, such as, $\{\Pi_i^{\pi^{A.L}}(db_i | b^{i-1}): i=0, \ldots, n\}$   or any other conditional distribution, which depends on  $A_{i-\overline{I}}^{i-1}$, for any $I < \overline{I} \leq i-1$, then the joint distribution of the corresponding upper bound  (\ref{equation102_VE_31_dd}) is (\ref{cor-ISR_32_c_B}), with $I=N$, in view of the transmission cost constraint, while  (\ref{cor-ISR_31_c_B}) will be inconsistent with the joint distribution.  This means the minimum information structure of the channel input distribution, which corresponds to the upper bound, for each $i$, is ${\cal I}_i^P=\{a_{i- I}^{i-1}, b^{i-1}\}$, for $i=0,1, \ldots, n$.  
This completes the prove.
\end{proof}

We make the following comments regarding the derivation of the theorem. \\

\begin{remark}(Comments on Theorem~\ref{cor-ISR})\\
%(a) As illustrated prior to the statement of the theorem and in the derivation, we do not need to apply the variational equality to determine the information structures of optimal channel input distributions for channel distribution of Class A  and transmission cost function of Class A. The conclusion follows directly from the standard theory of Markov Decision.  \\
(a) Recall the functional defined by (\ref{BAM52a_new}), below   Theorem~\ref{thm-var}, specialized to channel distribution Class A, with $\overrightarrow{Q}(db^n|a^n)$ replaced by  $\overrightarrow{Q}^A(db^n|a^n)\tri \otimes_{i=0}^n Q_i(db_i|b^{i-1}, a_{i-L}^i)$.   The pay-off functional in  (\ref{von_n_2}), is equivalent to 
\bea
{\mb I}(\overrightarrow{V}_{0,n}, \overleftarrow{P}_{0,n}, \overrightarrow{Q}_{0,n}^A)\tri \sum_{i=0}^{n} {\bf E}^{ P}\Big\{  \ell_i^{\overrightarrow{V}}(A_i, \overline{S}_i)\Big\}
\eea
 and this functional   is convex in $\overrightarrow{V}_{0,n}(db^n) \in {\cal M}({\mb B}_0^n)$ for fixed $\overleftarrow{P}_{0,n}(da^n|b^{n-1}) \in {\cal M}({\mb A}_0^n)$ (since the channel $\overrightarrow{Q}^A(db^n|a^n)\in {\cal }{\cal M}({\mb B}_0^n)$ is always fixed),  and concave in  $\overleftarrow{P}_{0,n}(da^n|b^{n-1})\in {\cal M}({\mb A}_0^n)$ for a fixed $\overrightarrow{V}_{0,n}(db^n) \in {\cal M}({\mb B}_0^n)$. Hence,  if we also impose sufficient conditions  so that ${\mb I}(\overrightarrow{V}_{0,n},  \overleftarrow{P}_{0,n}, \overrightarrow{Q}_{0,n}^A)$ is lower semicontinuous in  $\overrightarrow{V}_{0,n}(db^n) \in {\cal M}({\mb B}_0^n)$ and upper semicontinuous in  $\overleftarrow{P}_{0,n}(da^n|b^{n-1})\in {\cal M}({\mb A}_0^n)$, then the saddle point inequalities hold \cite{sion1958}, and we have   
\begin{align}
 \sup_{ {\cal P}_{[0,n]}    } \inf_{{\cal V}_{[0,n]}} \sum_{i=0}^{n} {\bf E}^{ P}\Big\{  \ell_i^{\overrightarrow{V}}(A_i, \overline{S}_i)\Big\}  = \inf_{{\cal V}_{[0,n]}}  \sup_{ {\cal P}_{[0,n]}    } \sum_{i=0}^{n} {\bf E}^{ P}\Big\{  \ell_i^{\overrightarrow{V}}(A_i, \overline{S}_i)\Big\}.  \label{von_n_3}
\end{align}
For the special case of finite alphabet spaces $\{({\mb A}_i, {\mb B}_i): i=\ldots, -1, 0, \ldots, n \}$, all conditions for validity of  (\ref{von_n_3}) hold. However, for countable or Borel spaces (i.e., continuous alphabet spaces) we need to impose conditions for upper and lower semicontinuity of the functional ${\mb I}(\overrightarrow{V}_{0,n},  \overleftarrow{P}_{0,n}, \overrightarrow{Q}_{0,n}^A)$. Such conditions are identified  in \cite{charalambous-stavrou2013aa} using the topology of weak convergence of probability distributions.
\end{remark}

In the next remark, we illustrate  that when the class A channel distributions  and transmission cost functions are specialized to  $L=N=0$,  then the last theorem gives as  degenerate case,  one of the information structures   derived \cite{kourtellaris-charalambousIT2015_Part_1}. Moreover,  for memoryless channels with feedback, we also illustrate that the  derivation based on variational equalities, gives an  alternative approach to derive the memoryless property of capacity achieving distribution, to the one given in \cite{cover-thomas2006},  which is based on first showing that feedback does not increase capacity. 

\ \

\begin{remark}(Fundamental differences between $I \neq 0$ and $I=0$)\\
\label{SC_REM}
(a) 
If $L=N=0$ then $I=0$,  which corresponds to a  channel distribution and transmission cost function,  that do not depend on past channel input symbols,  and  (\ref{cor-ISR_31_c_B}) and (\ref{cor-ISR_32_c_B}) are induced by the channel and channel input distribution $\{\pi_i^{A.0}(da_i|b^{i-1}): i=0, \ldots, n\big\}$, and all statements of Theorem~\ref{cor-ISR}, (b) specialize to one of the results derived in \cite{kourtellaris-charalambousIT2015_Part_1}, as follows. \\
The characterization of FTFI capacity is given by 
\begin{align}
{C}_{A^n \rar B^n}^{FB,A.0}
= \sup_{\overline{\cal P}_{[0,n]}^{A.0}(\kappa)} \sum_{i=0}^n {\bf E}^{ \pi^{A.0}}\Big\{
\log\Big(\frac{dQ_i(\cdot|B^{i-1},A_i)}{d\Pi_i^{ \pi^{A.0}}(\cdot|B^{i-1})}(B_i)\Big)
\Big\}  \label{cor-ISR_25a_a_c_N-SC}
\end{align}
where 
\begin{align}
\overline{\cal P}_{[0,n]}^{A.0}(\kappa) \tri \Big\{\pi_i^{A.0}(da_i |b^{i-1}), i=0,\ldots,n: \frac{1}{n+1} {\bf E}^{\pi^{A.0}}\Big( \sum_{i=0}^n \gamma_{i}^{A.0}(A_i, B^{i})\Big) \leq \kappa   \Big\}
\end{align}
and  the joint and marginal distributions  are 
% induced by  $\{Q_i(db_i|a_{i-L}^{i-1}, b^{i-1}), {\pi }_i^{A.I}(da_{i}|a_{i-I}^{i-1},b^{i-1}): i=0, \ldots, n\big\}$, as follows.
given by 
\begin{align}
  \Pi_i^{\pi^{A.0}}(db_i | b^{i-1}) =& \int_{  {\mb A}_{i} }   Q_i(db_i |b^{i-1}, a_i) \otimes   {\pi}_i^{A.0}(da_i | b^{i-1}), \label{cor-ISR_31_c_B_SC}\\
  {\bf P}^{\pi^{A.0}}(da^i,  d b^{i})=& \otimes_{i=0}^n\Big( Q_i(db_i|b^{i-1}, a_i) \otimes \pi_i^{A.0}(da_i|b^{i-1})\Big), \hso  i=0, \ldots, n. \label{cor-ISR_32_c_B_SC} 
 \end{align} 
In this case, for each $i$, the information structure of the channel input distribution is ${\cal I}_i^P =\{b^{i-1}\}$, and this information is also known at the decoder. In view of this, there is no need for the decoder to estimate any state variable using an \'a posteriori distribution, as in (\ref{cor-ISR_31_c_B}). \\
On the other hand, when $I \neq 0$  the optimal channel input distribution is of the form $\{\pi_i^{A.I}(da_i|a_{i-I}^{i-1}, b^{i-1}): i=0, \ldots, n\big\}$, and hence  for each $i$, the information structure is  ${\cal I}_i^P=\{a_{i-I}, \ldots, a_{i-1}, b^{i-1}\}$, is known to  the encoder, however,  the  additional variables $\{a_{i-I}, \ldots, a_{i-1}\}$ (i.e.,  state variables) are not known to the decoder. Hence, at each time $i$,  the additional  variables  $\{a_{i-I}, \ldots, a_{i-1}\}$  need to be estimated at the decoder, for $i=0, \ldots, n$. On the other hand, when $I=0$, since the  optimal channel input distribution is of the form $\{\pi_i^{A.0}(da_i|b^{i-1}): i=0, \ldots, n\big\}$, then there are  no additional state variables which need  to estimated at the decoder, because for each $i$, the decoder knows $B^{i-1}=b^{i-1}$, for $i=0, \ldots, n$. \\
The fundamental difference is that  the case $I \neq 0$ corresponds to an encoder or strategy  with memory or dynamics, in view of its dependence on past channel inputs, while the case $I=0$ corresponds to an encoder or strategy without memory or dynamics, since it does not depend on past channel input symbols. This fundamental difference needs to be accounted for, when attempting to optimize the characterizations of FTFI capacity. It is illustrated in Section~\ref{exa_gen}, for the application example of Multiple-Input Multiple Output (MIMO) Gaussian Recursive Linear Channel Models.

(b) Application of variational equalities to memoryless channels. If the channel is memoryless, i.e., ${\bf P}_{B_i|B^{i-1}, A^i}(db_i|b^{i-1}, a^i)=Q_i(db_i| a_i)- a.a. (b^{i-1}, a^i), i=0, \ldots, n$, then from Theorem~\ref{cor-ISR}, we obtain 
\begin{align}
\ell_i^{\pi^{A.0}}\big(a_i, {s}_i \big)  =&\int_{ {\mb  B}_i } \log \Big(\frac{Q_i(\cdot|a_i)}
  {\Pi_i^{\pi^{A.0}}(\cdot|b^{i-1})}(b_i)\Big) Q_i(db_i|a_i), \hso {s}\tri b^{i-1}, \hso  i=0, \ldots, n \label{BAM31_newa} \\
  \hso\Pi_i^{\pi^{A.0}} (db_i| b^{i-1})=&  \int_{{\mb A}_i} Q_i(db_i|a_i)\otimes {\pi_i^{A.0}}(da_i|b^{i-1})
\end{align}
 Thus,   for each $i$, the pay-off function $\ell_i^{ \pi^{A.0}}\big(a_i, \cdot \big)$ depends on ${s}=b^{i-1}$ only through the control object $g_i(b^{i-1})\tri {\bf P}(da_i|b^{i-1})\equiv \pi_i^{A.0}(da_i|b^{i-1})$, and not the channel distribution.\\
%, then formal application of   Markov Decision theory, implies  the optimal channel input conditional distribution of memoryless channels satisfies $g_i(b^{i-1})  \equiv  {\bf P}(da_i)-a.a.b^{i-1}$, for $i=0, \ldots, n$, i.e., it is independent of the conditioning information. \\ 
By an application of the variational equality, repeating the steps, starting with (\ref{equation15a_VE_300_b}) and leading to (\ref{equation102_VE_31}), with the corresponding upper bound obtained by using    
\bea
V_i(db_i|b^{i-1}) = \Pi_i^{\pi}(db_i)\tri \int_{{\mb A}_i} Q_i(db_i|a_i)\otimes {\pi}_i(da_i), \hso i=0, \ldots, n
\eea
i.e., corresponding to $\pi_i^{A.0}(da_i| b^{i-1})=\pi_i(da_i)\equiv{\bf P}(da_i), i=0, \ldots, n$, then the following upper bound is obtained. 
\begin{align}
&\sup_{\overline{\cal P}_{[0,n]}^{A.0}  }  \sum_{i=0}^n {\bf E}^{ \pi^{A.0}}\Big\{
\log\Big(\frac{dQ_i(\cdot|A_i)}{d\Pi_i^{ \pi^{A.0}}(\cdot|B^{i-1})}(B_i)\Big)
\Big\}  \leq \sup_{\big\{\pi_i(da_i): i=0,\ldots,n\big\}} \sum_{i=0}^n {\bf E}^{ \pi}\Big\{
\log\Big(\frac{dQ_i(\cdot|A_i)}{d\Pi_i^{ \pi}(\cdot)}(B_i)\Big)
\Big\}. \label{equation102_VE_31_new}
\end{align}
Further, the reverse inequality holds, by restricting the channel input distributions to the smaller conditional independent set $\big\{\pi_i^{A.0}(da_i|b^{i-1})=\pi_i(da_i), i=0, \ldots, n\big\}$, and from (\ref{equation15a_VE_40}) (with $P$ replaced by $\pi^{A.0}$) the following lower bound is obtained.
\begin{align}
 \sup_{  \overline{\cal P}_{[0,n]}^{A.0}   }  \sum_{i=0}^n {\bf E}^{ \pi^{A.0}}\Big\{
\log\Big(\frac{dQ_i(\cdot|A_i)}{d\Pi_i^{\pi^{A.0} }(\cdot|B^{i-1})}(B_i)\Big)
\Big\}  \geq  \sup_{\big\{\pi_i(da_i): i=0,\ldots,n\big\}} \sum_{i=0}^n {\bf E}^{ \pi}\Big\{
\log\Big(\frac{dQ_i(\cdot|A_i)}{d\Pi_i^{ \pi}(\cdot)}(B_i)\Big)
\Big\}  \label{equation15a_VE_40_a}
\end{align}
i.e., the upper bound is achievable, when the process $\big\{(A_i, B_i): i=0, \ldots, n\big\}$ is jointly independent. \\
Note that for memoryless channels with feedback, the standard method often applied to derive the capacity achieving distribution, is via the converse coding theorem, which pre-supposes that it is shown   that feedback does not increase capacity, compared to the case without feedback \cite{cover-thomas2006}.  
 As pointed out by Massey \cite{massey1990}, for channels with feedback, it will be a mistake to use  mutual information $I(A^n; B^n)$,  because by Marko's bidirectional information \cite{marko1973}, mutual information is not a tight bound on any achievable rate for channels with feedback.
% , due to the following decomposition. 
%\begin{align}
%I(A^n ; B^n)=& \sum_{i=0}^n {\bf E}^{ P}\Big\{
%\log\Big(\frac{dQ_i(\cdot|A_i)}{d\Pi_i^{ P}(\cdot|B^{i-1})}(B_i)\Big) 
%\Big\}+ \sum_{i=0}^n {\bf E}^{ P}\Big\{
%\log\Big(\frac{dP_i(\cdot|A^{i-1}, B^{i-1})}{d {\bf P}^{ P}(\cdot|A^{i-1})}(A_i)\Big) 
%\Big\} \\
%=& \sum_{i=0}^n {\bf E}^{ P}\Big\{
%\log\Big(\frac{dQ_i(\cdot|A_i)}{d\Pi_i^{ P}(\cdot|B^{i-1})}(B_i)\Big) 
%\Big\}, \hso \mbox{if and only if the channel is used without feedback.}  \label{nofed}
%\end{align}
%Since (\ref{nofed}) holds if and only if the channel is used without feedback, and this is equivalent to the conditional independence condition,  $P_i(da_i|a^{i-1}, b^{i-1})={\bf P}(da_i|a^{i-1}), i=0, \ldots, n$, then 
 Strictly speaking, for memoryless channels, any derivation of capacity achieving distribution for channels with feedback,  which applies the bound  $I(A^n; B^n) \leq \sum_{i=0}^n I(A_i; B_i)$, pre-supposes that it is already shown that feedback does not increase capacity, i.e., that ${\bf P}(da_i|a^{i-1}, b^{i-1})={\bf P}(da_i)-a.a.(a^{i-1}, b^{i-1}), i=0, \ldots, n$ (see \cite{cover-thomas2006}).
%  because in this case  only the rate through the forward channel $\big\{Q_i(db_i|a_i), i=0, \ldots, n\big\}$ is   accounted for. 
\end{remark}

Next, we give  examples to illustrate the dependence of  the information structures of optimal channel input distributions on $L, N$. \\

\begin{example}(Channel Class A and Transmission Cost Class A)\\
Case 1: $I\neq 0$.  Consider a channel $\big\{Q_i(b_i|b^{i-1},a_i, a_{i-1}): i=0,1, \ldots, n\big\}$. \\
(a) Without Transmission Cost.  By Theorem~\ref{cor-ISR}, (a) (since there is no transmission cost constraint)   the optimal channel input conditional distribution occurs in the subset 
\bea
\overline{\cal P}_{[0,n]}^{A.1} \tri \big\{ \pi_i^{A.1}(da_i|a_{i-1}, b^{i-1}): i=0, \ldots, n\big\}  \subset   {\cal P}_{[0,n]}
\eea
and   the characterization of the FTFI capacity is 
\begin{align}
C_{A^n \rar B^n}^{FB, A.1} \tri & \sup_{ \overline{\cal P}_{[0,n]}^{A.1}} \sum_{i=0}^n {\bf E}^{\pi^{A.1}}\Big\{ \log \Big(   \frac{d Q_i(\cdot|B^{i-1}, A_i, A_{i-1})}{ d\Pi_i^{\pi^{A.1}}(\cdot|B^{i-1})}(B_i)\Big)\Big\} \\
=& \sup_{ \overline{\cal P}_{[0,n]}^{A.1}} \sum_{i=0}^n I(A_{i-1},A_i; B_i| B^{i-1})
\end{align}
where 
\begin{align}
  {\Pi}_i^{\pi^{A.1}}(db_i | b^{i-1})=& \int_{  {\mb A}_{i-1}^i } Q_i(db_i |b^{i-1}, a_i, a_{i-1}) \otimes  \pi_i^{A.1}(da_i |a_{i-1}, b^{i-1})  \otimes {\bf P}^{\pi^{A.1}}(a_{i-1}| b^{i-1})  , \hso i=0, \ldots, n, \label{C-A.2.1}\\
  {\bf P}^{\pi^{A.1}}(db^i,da^i)=& \otimes_{i=0}^n \Big(Q_i(db_i|b^{i-1}, a_i, a_{i-1}) \otimes \pi_i^{\pi^{A.1}}(da_i|a_{i-1}, b^{i-1})\Big), \hso  i=0, \ldots, n, \label{C-A.2.2} \\
{\bf P}^{\pi^{A.1}}(a_{i-1}| b^{i-1})& \hst \mbox{satisfy recursions  (\ref{A_POST_1_n}) and (\ref{A_POST_1_nn}) with $L=1$.}
 \end{align}
{\it (b) With Transmission Cost Function  $\big\{ \gamma_i^{A.2}(a_{i-2}^i, b^{i}), i=0,1, \ldots, n\big\}$, that is, $L=1, N=2$.}  The characterization of the FTFI  capacity is given by the following expression.
\begin{align}
C_{A^n \rar B^n}^{FB, A.2}(\kappa) =& \sup_{ \overline{\cal P}_{[0,n]}^{A.2}(\kappa) } \sum_{i=0}^n I(A_{i-1}, A_i; B_i|B^{i-1}) \label{C-A.2.3}
\end{align}
where 
\begin{align}
 &  \overline{\cal P}_{[0,n]}^{A.2}(\kappa)\tri  \Big\{ \pi_i^{A.2}(da_i|a_{i-1},a_{i-2}, b^{i-1}), i=0, \ldots, n: \frac{1}{n+1} {\bf E}^{\pi^{A.2}} \Big(\sum_{i=0}^n \gamma_i^{A.2}(A_i, A_{i-1},A_{i-2}, B^{i})\Big)\leq \kappa      \Big\}, \\
 & {\Pi}_i^{\pi^{A.2}}(db_i | b^{i-1})= \int_{  {\mb A}_{i-1}^i } Q_i(db_i |b^{i-1}, a_i, a_{i-1}) \otimes  \pi_i^{A.2}(da_i |a_{i-1},a_{i-2}, b^{i-1})  \otimes {\bf P}^{\pi^{A.2}}(a_{i-1}, a_{i-2}| b^{i-1}), \label{C-A.2.1.2}\\
&  {\bf P}^{\pi^{A.2}}(db^i,da^i)= \otimes_{i=0}^n \Big(Q_i(db_i|b^{i-1}, a_i, a_{i-1}) \otimes \pi_i^{A.2}(da_i|a_{i-1},a_{i-2}, b^{i-1})\Big), \hso  i=0, \ldots, n \label{C-A.2.2.2} \\
\\
&{\bf P}^{\pi^{A.2}}(a_{i-1}, a_{i-2}| b^{i-1}) \hst \mbox{satisfy recursions  (\ref{A_POST_1_n}) and (\ref{A_POST_1_nn}) with $L=2$.}
\end{align}
Since, $N=2$ and $L=1$, the dependence of the optimal channel input distribution on past channel input symbols is determined from the dependence of the instantaneous transmission cost  on past channel input symbols. Moreover, although, in both cases, with and without transmission cost, the pay-off $\sum_{i=0}^n I(A_{i-1}, A_i; B_i|B^{i-1})$ is the same, the channel output transition probability distributions  and joint distributions, are different, because these are induced by different optimal channel input conditional distributions.\\
Case 2: $I=0$.  Consider any channel and transmission cost function of Remark~\ref{SC_REM}, (a). Clearly, this is much simpler compared to Case 1, because the characterization of FTFI capacity is not a functional of the \^a posteriori distribution of $A_{i-1}$ or $(A_{i-1}, A_{i-2})$ given $B^{i-1}$, for $i=0, \ldots, n$.  
\end{example}

\subsubsection{\bf Channel Class A and Transmission Cost Class B and Vice-Versa}  
From Theorem~\ref{cor-ISR}, we can also deduce the information structures of optimal channel input conditional distributions for channels of Class A and transmission cost functions of Class B, and vice-versa. These  are stated as a corollary.

\ \

\begin{corollary}(Class A  channels and Class  B transmission cost functions and vice-versa)\\
(a) Suppose the channel distribution is of Class A, as in Theorem~\ref{cor-ISR}, i.e., $\{Q_i(db_i|b^{i-1}, a_{i-L}^i):  i=0, \ldots, n\}$, the transmission cost function is of Class B, specifically, $\{\gamma_i^{B.K}(a^i, b_{i-K}^{i}): i=0, \ldots, n\}$, and the corresponding average transmission cost constraint is   defined by
\begin{align}
{\cal P}_{[0,n]}^{B}(\kappa) \tri     \Big\{   P_i(da_i | a^{i-1}, b^{i-1}),  i=0, 1, \ldots, n:   \frac{1}{n+1} {\bf E}^{P}\Big(\sum_{i=0}^n \gamma_i^{B.K}(A^i, B_{i-K}^{i})\Big) \leq \kappa  \Big\}\subset {\cal P}_{[0,n]}. \label{cor-ISR_29_c_gene} 
\end{align}
Then the optimal channel input conditional distribution, which   maximizes  $I(A^n \rar B^n)$ defined by (\ref{cor-ISR_25a_new})  over ${\cal P}_{[0,n]}^B(\kappa)$, is of the form  $\big\{ P_i(da_i | a^{i-1}, b^{i-1}),  i=0, 1, \ldots, n\big\}$ (i.e., there is no reduction in the  information structure of the optimal channel input distribution).\\
(b) Suppose the channel distribution is of Class B, defined by
\bea
{\bf P}_{B_i|B^{i-1}, A^i}(db_i|b^{i-1}, a^i)=Q_i(db_i|b_{i-M}^{i-1}, a^i)-a.a. (b^{i-1}, a^i), \hso i=0, \ldots, n \label{cor-ISR_29_NR}
\eea
   and the average transmission cost constraint is ${\cal P}_{0,n}^A(\kappa)$ defined by (\ref{cor-ISR_29_cc_CC}) (i.e., it corresponds to a transmission cost function of Class A). Then directed information is given by  
   \begin{align}
I(A^n\rightarrow B^n)= \sum_{i=0}^n {\bf E}^{ P}\Big\{
\log\Big(\frac{dQ_i(\cdot|B_{i-M}^{i-1},A^i)}{d\Pi_i^{ P}(\cdot|B^{i-1})}(B_i)\Big)
\Big\} \label{cor-ISR_25a_new_BA}
\end{align}
where 
 \begin{align}
\Pi_i^{ P} (db_i| b^{i-1})=&  \int_{{\mb A}^i} Q_i(db_i|b_{i-M}^{i-1},  a^i)\otimes P_i(da_i|a^{i-1}, b^{i-1})\otimes {\bf P}^{ P}(da^{i-1}|b^{i-1}), \hso  i=0, \ldots, n \label{CIS_9a_new_BA}
\end{align}
and $\big\{{\bf P}^{ P}(da^{i-1}|b^{i-1}): i=0, \ldots, n\big\}$ satisfies a recursion.
Moreover the optimal channel input distribution, which maximizes (\ref{cor-ISR_25a_new_BA}) over  ${\cal P}_{0,n}^A(\kappa)$ is   of the form  $\big\{ P_i(da_i | a^{i-1}, b^{i-1}),  i=0, 1, \ldots, n\big\}$ (i.e., there is no reduction in information structure).

\end{corollary} 
\begin{proof} This follows from the derivation of Theorem~\ref{cor-ISR}.
% (a), (b) By repeating the derivation of Theorem~\ref{cor-ISR}, (b),  if necessary, we can verify that if either the pay-off or the channel conditional distribution, depends on the entire history of the channel input process, then the augmented pay-off is a functional of the entire past of channel input symbols. This implies there is no reduction in the information structure of the optimal channel input conditional distribution, which  maximizes directed information. 
\end{proof}

\subsection{Channels Class C and Transmission Cost  Class C, A or B} 
\label{class_C}
In this section, we consider channel  distributions of Class C and transmission cost functions Class C, A or B.\\
Clearly, channel distributions of Class C and transmission cost functions of Class C, depend only on finite channel input and output symbols, when compared to any of the ones treated in previous sections.\\ Since  any channel of Class C is a special case of Channels of Class A, and any transmission cost  of Class C is a special case of transmission costs of Class A, then we can invoke Theorem~\ref{cor-ISR} to conclude that the maximizing channel input conditional distribution occurs in the  subset $\overline{\cal P}_{[0,n]}^{A,I}(\kappa) \subset {\cal P}_{[0,n]}(\kappa)$.

\subsubsection{\bf Channel Class C with Transmission Costs  Class C}
\label{Class_CC}
 Consider a channel distribution  of Class C, i.e., $\big\{Q_i(db_i|b_{i-M}^{i-1},a_{i-L}^i): i=0,1, \ldots, n\big\}$, and an average  transmission cost constraint corresponding to a transmission cost function of Class C, specifically, $\{\gamma_i^{C.N, K}(a_{i-N}^i, b_{i-K}^{i}): i=0, \ldots, n\}$, defined as follows.  
\begin{align}
{\cal P}_{[0,n]}^C(\kappa) \tri 
  \Big\{   P_i(da_i|a^{i-1}, b^{i-1}), i=0, 1, \ldots, n: \frac{1}{n+1} {\bf E}^{P}\Big( \sum_{i=0}^n \gamma_i^{C.N, K}(A_{i-N}^i, B_{i-K}^{i}) \leq \kappa  \Big)\Big\}. \label{cor-ISR_29_cc_C4_CC}
\end{align}
%Since a channel of Class C is a special case of channel of Class A,  and a transmission cost function of Class C,  is a special case of  a transmission cost function of Class A,  
%
From the preliminary discussion above, 
then  Theorem~\ref{cor-ISR}, (b) is directly applicable, hence we obtain the following   characterization of FTFI capacity. 
 \begin{align}
{C}_{A^n \rar B^n}^{FB,C}(\kappa) \tri & \sup_{ {\cal P}_{[0,n]}^C(\kappa) } \sum_{i=0}^n {\bf E}^{P}\Big\{
\log\Big(\frac{dQ_i(\cdot|B_{i-M}^{i-1}, A_{i-L}^i)}{d\Pi_i^{ P}(\cdot|B^{i-1})}(B_i)\Big)\Big\}\label{CM-TC-C_1_CC} \\
=& \sup_{\overline{\cal P}_{[0,n]}^{C.I}(\kappa) } \sum_{i=0}^n {\bf E}^{ \pi^{A.I}}\Big\{
\log\Big(\frac{dQ_i(\cdot|B_{i-M}^{i-1},A_{i-L}^i)}{d\Pi_i^{ \pi^{A.I}}(\cdot|B^{i-1})}(B_i)\Big) \label{CM-TC-C_1_CC_C} 
\Big\}  \equiv {C}_{A^n \rar B^n}^{FB,C.I}(\kappa)
\end{align}
where 
\begin{align}
\overline{\cal P}_{[0,n]}^{C.I}(\kappa) \tri & \Big\{\pi_i^{A.I}(da_i |a_{i-I}^{i-1}, b^{i-1}), i=0,\ldots,n: \frac{1}{n+1} {\bf E}^{\pi^{A.I}}\big( \sum_{i=0}^n \gamma_{i}^{C.N, K}(A_{i-N}^i, B_{i-K}^{i})\big) \leq \kappa   \Big\} \\
  \Pi_i^{\pi^{A.I}}(db_i | b^{i-1}) =& \int_{  {\mb A}_{i-I}^{i} }   Q_i(db_i |b_{i-M}^{i-1}, a_{i-L}^i) \otimes  \pi_i^{A.I}(da_i |a_{i-I}^{i-1}, b^{i-1}) \otimes {\bf P}^{\pi^{A.I}}(da_{i-I}^{i-1} | b^{i-1}),\hso I\tri \max\{L, N\}, \label{CM-TC-C_3_CC}\\
  {\bf P}^{\pi^{A.I}}(da^i,  d b^{i})=& \otimes_{j=0}^i \Big(Q_j(db_j|b_{j-M}^{j-1}, a_{j-L}^j) \otimes \pi_j^{A.I}(da_j|a_{j-I}^{j-1}, b^{j-1})\Big), \hso  i=0, \ldots, n. \label{CM-TC-C_4_cc}
  \end{align}
  The \'a posteriori distribution satisfies the following recursion.
  \begin{align}
  {\bf P}^{\pi^{A.I}}(da_{i-I}^{i-1}|b^{i-1}) =& \tilde{T}_{i-1}\Big(b_{i-1}, b_{i-1-M}^{i-2}, \pi_{i-1}^{A.I}(\cdot|\cdot, b^{i-2}), {\bf P}^{\pi^{A.I}}(\cdot|b^{i-2}) \Big)(da_{i-I}^{i-1}),\hso i=1, \ldots, n, \label{A_POST_2_n}\\
{\bf P}^{\pi^{A.I}}(da_{-I}^{-1}|b^{-1})=&\mbox{given}  \label{A_POST_2_nn}
 \end{align}

{\it Special Case:  $L=N=0$,  $I=0$ and Initial Data $b_{-M \wedge K}^{-1} \tri \max\{b_{-M}^{-1}, b_{-K}^{-1}\}$ Known to the Encoder and Decoder.}   In this case,  we can further apply the variational equalities of directed information and  stochastic optimal control theory (as in Theorem~\ref{cor-ISR}), or invoke \cite{kourtellaris-charalambousIT2015_Part_1},  to  deduce that the supremum over the set of channel input conditional distributions $\overline{\cal P}_{[0,n]}^{C.0}(\kappa)$  in (\ref{CM-TC-C_1_CC_C}), occurs in a smaller subset $\sr{\circ}{\cal P}_{[0,n]}^J(\kappa) \subset  \overline{\cal P}_{[0,n]}^{C.0}(\kappa) \subset {\cal P}_{[0,n]}(\kappa)$, 
%Alternatively, we can apply the variational equalities directly, without using the results of Theorem~\ref{cor-ISR}. This is not pursued any further because it is done in  
%For the other combinations of channels of Class C and transmission cost functions of Class A or B, the characterizations follow directly from the previous characterizations. 
which   satisfies the conditional independence condition, $\pi^{A.I}_i(da_i|a_{i-I}^{i-1}, b^{i-1})= {\bf P}(da_i|b_{i-J}^{i-1})-a.a. (a_{i-I}^{i-1}, b^{i-1}), i=0, \ldots, n, J\tri \max\{M, K\}$. This follows from the fact that, for each $i$, the pay-off functional, i.e.,   $\ell_i^{\pi^{A.L}}(a_i, \cdot)\equiv \ell_i^{\pi^{A.0}}(a_i, \cdot)$,  depends on $b_{i-J}^{i-1}$ via the channel distribution and the cost function, and on the additional symbols $b^{i-1-J}$ only via the control object $g_i(b^{i-1})\tri \pi_i^{A.0}(da_i| b^{i-1})$, for $i=0, \ldots, n$. For completeness we state the main theorem without derivation, since this is given in \cite{kourtellaris-charalambousIT2015_Part_1}.

\ \

\begin{theorem}
\label{cor-ISR_C5_C_MK}
(Channel class C    transmission cost class C, $L=N=0$ and initial data known to the encoder and decoder)\\
 Suppose the channel conditional distribution is of Class C with $L=0$, i.e., $\big\{Q_i(db_i|b_{i-M}^{i-1},a_i): i=0,1, \ldots, n\big\}$,  the transmission cost constraint  is defined by  (\ref{cor-ISR_29_cc_C4_CC}) with $N=0$ defined by 
 \begin{align}
{\cal P}_{[0,n]}^{C.0,K}(\kappa) \tri 
  \Big\{   P_i(da_i|a^{i-1}, b^{i-1}), i=0, 1, \ldots, n: \frac{1}{n+1} {\bf E}^{P}\Big( \sum_{i=0}^n \gamma_i^{C.0, K}(A_i, B_{i-K}^{i}) \leq \kappa  \Big)\Big\} \label{cor-ISR_29_cc_C4_CC_SC}
\end{align}
 the initial data $b_{-J}^{-1} \tri \max\{b_{-M}^{-1}, b_{-K}^{-1}\}$ is known to the encoder and decoder, 
 and the  following  condition holds.
\begin{align}
\sup_{ {\cal P}_{[0,n]}^{C.0,K}(\kappa)} I(A^{n}
\rar {B}^{n}) = \inf_{\lambda\geq 0} \sup_{  {\cal P}_{[0,n]}   } \Big\{ I(A^{n}
\rar {B}^{n})  - \lambda \Big\{ {\bf E}^P\Big(\sum_{i=0}^n \gamma_i^{C.0, K}(A_i,B_{i-K}^{i})\Big)-\kappa(n+1) \Big\}\Big\}. \label{ISDS_6CC_MK}
\end{align}   
Then the following hold.\\
 The characterization of the FTFI capacity is given by the following expression.
\begin{align}
{C}_{A^n \rar B^n}^{FB,C.0,J}(\kappa)
= \sup_{\sr{\circ}{\cal P}_{[0,n]}^{C.0,J}(\kappa)} \sum_{i=0}^n {\bf E}^{ \pi^{C.0,J}}\Big\{
\log\Big(\frac{dQ_i(\cdot|B_{i-M}^{i-1},A_i)}{d\Pi_i^{ \pi^{C.0,J}}(\cdot|B_{i-J}^{i-1})}(B_i)\Big)
\Big\},  \hso   J \tri \max\{M, K\}  \label{cor-ISR_25a_a_c_New_a-n_MK}
\end{align}
 where the  maximizing channel input conditional distribution  occurs in  the subset    
\begin{align}
&\sr{\circ}{\cal P}_{[0,n]}^{C.0,J}(\kappa) \tri 
  \Big\{   \pi_i^{C.0,J}(da_i| b_{i-J}^{i-1}), i=0, 1, \ldots, n: \frac{1}{n+1} {\bf E}^{\pi^{C.0,J}}\Big( \sum_{i=0}^n \gamma_i^{C.0, K}(A_i, B_{i-K}^{i})\Big) \leq \kappa \Big\} \label{cor-ISR_29_cc_C4_a_CCC-n_MK}
\end{align}
and the joint and channel output distributions are given by 
\begin{align}
  \Pi_i^{\pi^{C.0,J}}(db_i | b_{i-J}^{i-1}) =& \int_{  {\mb A}_i }   Q_i(db_i |b_{i-M}^{i-1}, a_i) \otimes   {\pi}_i^{C.0,J}(da_i | b_{i-J }^{i-1}),  \label{cor-ISR_31_c_C5_CC_MK}\\
  {\bf P}^{\pi^{C.0,J}}(da^i,  d b^i)=& \otimes_{j=0}^i \Big(Q_j(db_j|b_{j-M}^{j-1}, a_i) 
\otimes \pi_j^{C.0,J}(da_j| b_{j-J}^{j-1})\Big),  \hso  i=0, \ldots, n. \label{cor-ISR_32_c_C5_CC_MK} 
 \end{align}
\end{theorem}
 \begin{proof} The derivation when the transmission cost function is $\{\overline{\gamma}_i^{C.0, K}(A_i, B_{i-K}^{i-1}): i=0, \ldots, n\}$ is given in \cite{kourtellaris-charalambousIT2015_Part_1}. This is easily modified to account  for the transmission cost function $\{\gamma_i^{C.0, K}(A_i, B_{i-K}^{i}): i=0, \ldots, n\}$. 
\end{proof}

\begin{remark}(Initial data unknown to the encoder and decoder)\\
We note that if the initial data $b_{-J}^{-1}\tri \max\{b_{-M}^{-1}, b_{-K}^{-1}\}$ is not available to the encoder and decoder, then these become additional variables which need to be estimated, 
%. In this case, although the channel distribution and transmission cost function are class C, i.e., they depend on limited past channel output symbols, 
and hence Theorem~\ref{cor-ISR_C5_C_MK} is no longer valid.  
%the optimal channel input distribution is  of the form $\{\pi_i^{A.0}(da_i|b^{i-1}): i=0, \ldots, n\}$ and consequently the transition probability distribution of the channel output process is not limited memory Markov; the characterization of FTFI capacity is given as in Remark~\ref{SC_REM}. 
\end{remark}

Next, we present  examples. \\

\begin{example}(Channel Class C) \\ 
\label{ex_c}
(a) Consider a channel $\big\{Q_i(db_i|b_{i-1},a_{i}) : i=0,1, \ldots, n\big\}$, i.e., $M=1, L=0$.\\
  By Theorem~\ref{cor-ISR_C5_C_MK},  the optimal channel input conditional distribution  occurs in the subset 
\bea
\sr{\circ}{ {\cal P}}_{[0,n]}^{C.0,1} \tri     \Big\{ \pi_i^{C.0,1}(da_i| b_{i-1}):i=0, 1, \ldots, n\Big\} 
\eea
and the characterterization of the FTFI capacity is 
\begin{align}
C_{A^n \rar B^n}^{FB, C.0,1} =&\sup_{  \sr{\circ}{ {\cal P}}_{[0,n]}^{C.0,1}   } \sum_{i=0}^n
{\bf E}^{\pi^{C.0,1}}\Big\{\log\big(\frac{dQ_i(\cdot| B_{i-1},A_i)}{d\Pi^{\pi^{c.1}}_i(\cdot| B_{i-1})}(B_i)\big)\Big\}\\
 \equiv & \sup_{  \sr{\circ}{ {\cal P}}_{[0,n]}^{C.0,1}   }   \sum_{i=0}^n I(A_i; B_i|B_{i-1}). \label{YKT-4_C}
\end{align}
where 
\begin{align}
  \Pi_i^{\pi^{C.0,1}}(db_i | b_{i-1}) =& \int_{  {\mb A}_i }   Q_i(db_i |b_{i-1}, a_i) \otimes   {\pi}_i^{C.0,1}(da_i |b_{i-1}), \label{cor-ISR_31_c_C5_C}\\
  {\bf P}^{\pi^{C.0,1}}(da^i,  d b^i)=&\otimes_{j=0}^i \Big( Q_j(db_j|b_{j-1}, a_{j}) 
\otimes \pi_j^{C.0,1}(da_j|b_{j-1})\Big),  \hso  i=0, \ldots, n. \label{cor-ISR_32_c_C5_C} 
 \end{align}
The above characterization of FTFI capacity implies
\begin{description}
\item[(a.i)] the joint process $\{(A_i,B_i): i=0, \ldots, n\}$ is first-order Markov; 

\item[(a.ii)] the channel output process  $\{B_i: i=0, \ldots, n\}$ is  first-order Markov.
\end{description}     
%\begin{align}
%{\bf P}(da_i, db_i|a^{i-1}, b^{i-1})={\bf P}^{\pi^{C.1,1}}(da_i, db_i|a_{i-1}, b_{i-1}), \hso
%{\bf P}(db_i|b^{i-1})={\bf P}^{\pi^{C.1,1}}(db_i|b_{i-1}), \hso i=0,1, \ldots, n.
%\end{align}
(b) Consider a channel $\big\{Q_i(db_i|b_{i-1},b_{i-2}, a_{i} ) : i=0,1, \ldots, n\big\}$, i.e., $M=2, L=0$.\\
  By Theorem~\ref{cor-ISR_C5_C_MK},  the optimal channel input conditional distribution  occurs in the subset 
\bea
\sr{\circ}{ {\cal P}}_{[0,n]}^{C.0,2} \tri     \Big\{ \pi_i^{C.0,2}(da_i|b_{i-1}, b_{i-2}):i=0, 1, \ldots, n\Big\} 
\eea
and the characterization of the FTFI capacity is 
\begin{align}
C_{A^n \rar B^n}^{FB, C.0,2} =&\sup_{  \sr{\circ}{ {\cal P}}_{[0,n]}^{C.0,2}   } \sum_{i=0}^n
{\bf E}^{\pi^{C.0,2}}\Big\{\log\big(\frac{dQ_i(\cdot| B_{i-1},B_{i-2}, A_i)}{d\Pi^{{\pi^{C.0,2}}}_i(\cdot| B_{i-1}, B_{i-2})}(B_i)\big)\Big\}\\
 \equiv & \sup_{  \sr{\circ}{ {\cal P}}_{[0,n]}^{C.0,2}   }   \sum_{i=0}^n I(A_{i-1}, A_i; B_i|B_{i-1}, B_{i-2}). \label{YKT-4_C_2}
\end{align}
The  above characterization of FTFI capacity implies
\begin{description}
\item[(b.i)] the joint process $\{Z_i\tri (B_{i-1}, B_{i-2}, A_{i-1}): i=0, \ldots, n\}$ is first-order Markov; 

\item[(b.ii)] the channel output process  $\{S_i \tri (B_{i-1}, B_{i-2}): i=0, \ldots, n\}$ is  first-order Markov,
\end{description} 
%or equivalently the following hold.     
%\begin{align}
%{\bf P}(dz_i|z^{i-1})={\bf P}^{\pi^{C.2}}(dz_i|z_{i-1}), \hso
%{\bf P}(ds_i|s^{i-1})={\bf P}^{\pi^{C.2}}(ds_i|s_{i-1}), \hso i=0,1, \ldots, n.
%\end{align}
The optimizations of characterizations of FTFI capacity expressions in (a) and (b) over the channel input distributions can be solved by applying dynamic programming, in view of the Markov property of the channel output processes. 
\end{example}

Let us illustrate, in the next  example,  the difference of the information structure of optimal channel input distribution, when the channel or the cost function depend on past channel input symbols, compared to Example~\ref{ex_c}.  \\

\begin{example}(Channel Class C and Transmission Cost Class C with $L\neq 0$ and or $N \neq 0$)\\ 
Consider a channel $\big\{Q_i(db_i|b_{i-1},a_i, a_{i-1}) : i=0,1, \ldots, n\big\}$ and transmission cost function $\big\{\gamma_i^{C.2,1}(a_{i-2}^i, b_{i-1}^i): i=0, \ldots, n\big\}$, i.e., , $M=1, L=1, N=2, K=1, I=\max\{L, N\}=2$.\\ By (\ref{CM-TC-C_1_CC})-(\ref{CM-TC-C_4_cc}),  the optimal channel input conditional distribution  occurs in the subset 
\bea
\overline{ {\cal P}}_{[0,n]}^{A.2}(\kappa) \tri     \Big\{ \pi_i^{A.2}(da_i|a_{i-2}^{i-1},  b^{i-1}), i=0, \ldots, n: \frac{1}{n+1} {\bf E}^{\pi^{A.2}} \big\{\sum_{i=0}^n \gamma_i(A_{i-2}^i, B_{i-1}^i)\big\} \leq \kappa   \Big\} 
\eea
and the characterterization of the FTFI capacity is 
\begin{align}
C_{A^n \rar B^n}^{FB, C.2}(\kappa) =&\sup_{  \overline{ {\cal P}}_{[0,n]}^{A.2}(\kappa)   } \sum_{i=0}^n {\bf E}^{\pi^{A.2}}\Big\{   \log\big(\frac{dQ_i(\cdot| B_{i-1},A_i, A_{i-1})}{d\Pi^{\pi^{A.2}}_i(\cdot| B^{i-1})}(B_i)\big)\Big\}
\end{align}
where 
\begin{align}
  \Pi_i^{\pi^{A.2}}(db_i | b^{i-1}) =& \int_{  {\mb A}_{i-2}^i }   Q_i(db_i |b_{i-1},a_i, a_{i-1}) \otimes   {\pi}_i^{A.2}(da_i |a_{i- 2}^{i-1}, b^{i-1})   \otimes {\bf P}^{\pi^{A.2}}(da_{i-2}^{i-1}| b^{i-1}), \label{cor-ISR_31_c_C5_C}\\
  {\bf P}^{\pi^{A.2}}(da^i,  d b^i)=&\otimes_{j=0}^i \Big( Q_j(db_j|b_{j-1},a_j, a_{j-1}) 
\otimes \pi_j^{A.2}(da_j|a_{j-2}^{j-1},b^{j-1})\Big). \label{cor-ISR_32_c_C5_C} \\
{\bf P}^{\pi^{A.2}}(a_{i-2}^{i-1}| b^{i-1})& \hst \mbox{satisfy recursions  (\ref{A_POST_2_n}) and (\ref{A_POST_2_nn}) with $M=1, I=2$.}
 \end{align}
This example illustrates that the  dependence of the transmission cost function, for each $i$, on $a_{i-2}$ in addition to symbols $\{a_{i-1}, a_i\}$ (i.e, the ones  the channel depends on),   implies the information structure of the optimal channel input conditional distribution is ${\cal I}_i^P \tri \{a_{i-1}, a_{i-2},  b^{i-1}\}$, for $i=0, \ldots, n$.\\
{\bf  Special Case.} If the channel is $\big\{Q_i(db_i|b_{i-1},a_i): i=0,1, \ldots, n\big\}$ then the optimal channel input conditional distribution  occurs in the subset 
$\overline{ {\cal P}}_{[0,n]}^{A.2}(\kappa)$,  which is fundamentally different from the information structures of Example~\ref{ex_c}, (a) (although the channels are identical).
\end{example}

\subsubsection{\bf Channel Class C with Transmission Costs  Class A} Consider a channel distribution  of Class C, i.e., $\big\{Q_i(db_i|b_{i-M}^{i-1},a_{i-L}^i): i=0,1, \ldots, n\big\}$,  and an average  transmission cost constraint ${\cal P}_{[0,n]}^{A}(\kappa)$ defined by (\ref{cor-ISR_29_cc_CC}),  and  corresponding to a transmission cost function of Class A, $\{\gamma_i^{A.N}(a_{i-N}^i, b^{i}): i=0, \ldots, n\}$, with $L \neq 0, N \neq 0$. We can repeat the derivation of Theorem~\ref{cor-ISR}, to obtain   the following characterization of FTFI capacity.  
 \begin{align}
&{C}_{A^n \rar B^n}^{FB,C,A}(\kappa) \tri  \sup_{\big\{P_i(da_i| a^{i-1},  b^{i-1})  :i=0,\ldots,n\big\} \in {\cal P}_{[0,n]}^A(\kappa) } \sum_{i=0}^n {\bf E}^P\Big\{ 
\log\Big(\frac{dQ_i(\cdot|B_{i-M}^{i-1}, A_{i-L}^i)}{d\Pi_i^{ P}(\cdot|B^{i-1})}(B_i)\Big)\Big\}  \label{CM-TC-C_1_CC_A} \\
=& \sup_{\big\{\pi_i^{A.I}(da_i |a_{i-I}^{i-1}, b^{i-1}), i=0,\ldots,n: \frac{1}{n+1} {\bf E}^{\pi^{A.I}}\big( \sum_{i=0}^n \gamma_{i}^{A.N}(A_{i-N}^i, B^{i})\big) \leq \kappa   \big\}} \sum_{i=0}^n {\bf E}^{ \pi^{A.I}}\left\{
\log\Big(\frac{dQ_i(\cdot|B_{i-M}^{i-1},A_{i-L}^i)}{d\Pi_i^{ \pi^{A.I}}(\cdot|B^{i-1})}(B_i)\Big) \label{CM-TC-C_1_CC_C_A}
\right\}  
\end{align}
where  $I\tri \max\{L,N \}$ and 
\begin{align}
  \Pi_i^{\pi^{A.I}}(db_i | b^{i-1}) =& \int_{  {\mb A}_{i-I}^{i} }   Q_i(db_i |b_{i-M}^{i-1}, a_{i-L}^i) \otimes  P_i(da_i |a_{i-I}^{i-1}, b^{i-1}) \otimes {\bf P}^{\pi^{A.I}}(da_{i-I}^{i-1} | b^{i-1}), \label{CM-TC-C_3_CC_A}\\
  {\bf P}^{\pi^{A.I}}(da^i,  d b^{i})=& \otimes_{j=0}^i \Big(Q_j(db_j|b_{j-M}^{j-1}, a_{j-L}^j) \otimes \pi_j^{A.I}(da_j|a_{j-I}^{j-1}, b^{j-1})\Big), \hso  i=0, \ldots, n. \label{CM-TC-C_4_cc-A}  \\
{\bf P}^{\pi^{A.I}}(a_{i-I}^{i-1}| b^{i-1})& \hst \mbox{satisfy recursions  (\ref{A_POST_2_n}) and (\ref{A_POST_2_nn}).}
 \end{align}

\subsubsection{\bf Channel Class C with Transmission Costs  Class B} Consider a channel distribution  of Class C, i.e., $\big\{Q_i(db_i|b_{i-M}^{i-1},a_{i-L}^i): i=0,1, \ldots, n\big\}$, and an average  transmission cost constraint ${\cal P}_{[0,n]}^{B}(\kappa)$ defined by (\ref{cor-ISR_29_c_gene}),  and  corresponding to a transmission cost function of Class B, $\{\gamma_i^{B.K}(a^i, b_{i-K}^{i}): i=0, \ldots, n\}$. \\
Similarly as above,  we can repeat the derivation of Theorem~\ref{cor-ISR}, to obtain   the following  characterization of FTFI capacity.  
 \begin{align}
{C}_{A^n \rar B^n}^{FB,C,B}(\kappa) \tri & \sup_{\big\{P_i(da_i| a^{i-1},  b^{i-1})  :i=0,\ldots,n\big\} \in {\cal P}_{[0,n]}^B(\kappa) } \sum_{i=0}^n {\bf E}^P \Big\{ 
\log\Big(\frac{dQ_i(\cdot|B_{i-M}^{i-1}, A_{i-L}^i)}{d\Pi_i^{ P}(\cdot|B^{i-1})}(B_i)\Big) \Big\} \label{CM-TC-C_1_CC_B} 
\end{align}
where  
\begin{align}
  \Pi_i^{P}(db_i | b^{i-1}) =& \int_{  {\mb A}^{i} }   Q_i(db_i |b_{i-M}^{i-1}, a_{i-L}^i) \otimes  P_i(da_i |a^{i-1}, b^{i-1}) \otimes {\bf P}^P(da^{i-1} | b^{i-1})  , \hso i=0, \ldots, n, \label{CM-TC-C_3_CC_B}\\
  {\bf P}^P(da^i,  d b^{i})=& \otimes_{j=0}^i \Big(Q_j(db_j|b_{j-M}^{j-1}, a_{j-L}^j) \otimes P_j(da_j|a^{j-1}, b^{j-1})\Big), \hso  i=0, \ldots, n \label{CM-TC-C_4_cc-B} 
 \end{align}
and $\big\{{\bf P}^P(da^{i-1} | b^{i-1}):  i=0, \ldots, n\big\}$ satisfies a recursion.

The  characterizations of FTFI capacity presented in this section cover many  channel distributions and transmission cost functions of practical interest.

\ \

\begin{conclusion}(Concluding comments)\\
%We  applied the {\it Two-Step Procedure}, which involves stochastic optimal control and  the variational equality of directed inequality, and we  determined  the optimal channel input conditional distribution set,  which characterizes FTFI capacity,   for many different classes of channels with memory and many different classes of transmission cost functions. 
%The characterizations of the FTFI capacity   are generalizations of many related resu two-letter feedback capacity characterization of  DMC and memoryless continuous alphabet channels.\\
(a) For specific channel distributions and transmission cost functions, it is possible to derive  closed form expressions for the optimal channel input conditional distributions and   corresponding characterizations of  FTFI capacity,  via dynamic programming.\\
% since the the theory of Markov decision can be used, by appropriately defining augmented processes (see \cite{...} for results in this  direction). \\  
(b) Whether feedback increases the characterizations of FTFI capacity compared to that of no feedback can be determined  by investigating  whether there exists a channel input distribution without feedback which induces the joint distribution of the joint process $\{(A_i, B_i): i=0, \ldots,n\}$ and the marginal distribution of the output process $\{B_i: i=0, \ldots, n\}$, corresponding to the characterization of FTFI capacity (see \cite{kourtellaris-charalambous-boutros:2015} for specific application examples).\\
(c) The characterizations of the feedback capacity, are obtained from the per unit time limit of the characterization of  FTFI capacity, provided the optimal channel input distributions induce information stability of the directed information density \cite{pinsker1964}. This  can be shown following \cite{ihara1993,tatikonda-mitter2009} with appropriate modifications.\\
%(d) With the aid of information structures of optimal channel input distributions, then the material of Section~\ref{LCM_G} are easily extended to the recursive NCMs given in Definition~\ref{R-exa_A_D}. We do this is subsequent sections.

%\\
%(d) We can apply the main theorems of this section to identify tighter bounds on feedback capacity expressions given in the literature (i.e., Theorem 1 in  \cite{yang-kavcic-tatikonda2005})). We illustrate this point shortly.
\label{concl-1}
\end{conclusion}

%\section{Dynamic Programming and Sufficient Statistics}

\section{Separation Principle of MIMO Gaussian Recursive Linear Channel Models}
\label{exa_gen}
In this section, we show that the maximization of FTFI capacity over channel input distributions exhibits a  separation principle. We show this separation principle by using the orthogonal decomposition of the realizations of optimal channel input distributions.   \\
The application examples we consider are   general Multiple-Input Multiple-Output Gaussian channel with arbitrary memory on past channel input and output symbols, and quadratic cost constraint, i.e., class C.  Via the separation principle, we derive an expression for the optimal channel input distribution, and we  relate the characterization of FTFI capacity to the well-known Linear-Quadratic-Gaussian partially observable stochastic optimal control problem \cite{kumar-varayia1986}. 

 \subsection{Multiple Input Multiple Output Gaussian Linear Channel Models with Memory One}
 \label{C-LCM_C_one}
 First, we treat a special  case, since the general case with arbitrary memory is handled similraly, with additional notation complexity.

\noi{\bf  Gaussian-Linear Channel Model.} Consider a recursive model, called Gaussian-Linear Channel Model (G-LCM) with quadratic cost function, defined as follows.
 \begin{align}
&B_i  =C_{i,i-1} \; B_{i-1} +D_{i,i} \; A_i +D_{i,i-1} A_{i-1}  + V_{i}, \hst  B_{-1}=b_{-1}, \hso A_{-1}=a_{-1},   \hso i= 0, \ldots, n, \label{G-LCM_1_G} \\
& {\bf P}_{V_i|V^{i-1}, A^i}={\bf P}_{V_i}, \;  V_i \sim N(0, K_{V_i}),\label{G-LCM_2_G}  \\
 &\gamma_i(a_i, b_{i-1}) \tri \langle a_i, R_{i,i} a_i \rangle + \langle b_{i-1}, Q_{i,i-1} b_{i-1} \rangle, \label{G-LCM_4_G}  \\
& C_{i,i-1} \in {\mb R}^{p\times p}, \hso  (D_{i,i}, D_{i,i-1}) \in  {\mb R}^{p\times q} \times {\mb R}^{p\times q}, \hso R_{i,i} \in {\mb S}_{++}^{q\times q}, \hso Q_{i,i-1} \in {\mb S}_{+}^{p\times p}, \hso i=0, \ldots, n \label{G-LCM_6_G} 
\end{align}
where the initial data $(b_{-1}, a_{-1})$ are known to the encoder and decoder.\\
By Section~\ref{class_C},   the optimal strategy is of the form $\{\pi_i^{A.1}(da_i|a_{i-1}, b^{i-1}): i=0, \ldots, n\}$, and the \^a posteriori distribution satisfies the following recursion.
% Moreover, 
%%by invoking the Wolf decomposition of Hilbert space random processes \cite{caines}  
%the process $\{A_i: i=0, \ldots, n\}$ admits the decomposition 
%\begin{IEEEeqnarray}{rCl}
%&&A_i= e_i^1(A_{i-1}, Y^{i-1}) + Z_i, \hso i=0,\ldots, n, \label{DOOB_1} \\
%&&Z_i \perp e_i^1(A_{i-1}, Y^{i-1}), \hst \{Z_i: i=0, \ldots, n\} \hso \mbox{orthogonal} \IEEEeqnarraynumspace \label{DOOB_11} 
%\end{IEEEeqnarray}
%i.e, $\{Z_i: i=0, \ldots, n\}$ is an independent process and independent of $\{V_i: i=0, \ldots, n\}$. 
\begin{align}
{\bf P}^{\pi^{A.1}}(da_{i-1}|b^{i-1}) =& \tilde{T}_{i-1}\Big(b_{i-1}, b_{i-2}, \pi_{i-1}^{A.1}(\cdot|\cdot, b^{i-2}), {\bf P}^{\pi^{A.1}}(\cdot|b^{i-2}) \Big)(da_{i-1}),\hso i=1, \ldots, n, \label{A_POST_1_G}\\
{\bf P}^{\pi^{A.1}}(da_{-1}|b^{-1})=&\mbox{given}.  \label{A_POST_1_Gn}
\end{align}
Next, we show that the presence of $\{C_{i,i-1}: i=0, \ldots, n\}$ in (\ref{G-LCM_1_G})  destoys the Markov property of process $\big\{\xi_{i-1}^{\pi^{A.1}}(B^{i-1})\tri {\bf P}^{\pi^{A.1}}(da_{i-1}|b^{i-1}): i=0, \ldots, n\big\}$.
% hence this \'a posteriori distribution is not a sufficient statistic, in the sense of generating the same inf but an augmented version of it is a sufficient statistic. 

%and that the joint process $\big\{\Big(\xi^{\Pi^{A.1}}(B^{i-1}),B_{i-1}\Big): i=0, \ldots, n\big\}$ is Markov.  

\ \

\begin{theorem}(Markov structure of augmented process)\\
\label{suf_st}
The process $\big\{\xi_{i-1}^{\pi^{A.1}}(B^{i-1})\tri {\bf P}^{\pi^{A.1}}(da_{i-1}|b^{i-1}): i=0, \ldots, n\big\}$ satisfying the recursion (\ref{A_POST_1_G}) and (\ref{A_POST_1_Gn}) is not Markov and furthermore the following holds.\\
For any test function $\varphi: {\mathbb A}_{i-1} \longmapsto {\mathbb R}$, which is continuous and bounded then 
\begin{align}
{\bf E}^{\pi^{A.1}}\Big\{ \int_{{\mb A}_{i-1}} \varphi(a) {\bf P}^{\pi^{A.1}}(da|b^{i-1})\Big|B^{i-2}\Big\} ={\bf E}^{\pi^{A.1}}\Big\{ \int_{{\mb A}_{i-1}} \varphi(a) {\bf P}^{\pi^{A.1}}(da|b^{i-1})\Big|B_{i-2}, \xi_{i-2}^{\pi^{A.1}}(B^{i-2}) \Big\},  \hso i=1, \ldots, n. \label{Mar_REC}
\end{align}
\end{theorem}
\begin{proof} For each $i$, the right hand side of  the recursion (\ref{A_POST_1_G}) depends in addition to $\{b_{i-1}, {\bf P}^{\pi^{A.1}}(da_{i-2}|b^{i-2})\}$ on $b_{i-2}$, hence the claim of the non-Markov property. From (\ref{A_POST_1_G}) we deduce (\ref{Mar_REC}). 
\end{proof}

By the channel definition,  directed information pay-off is expressed using conditional entropies  as follows.
\begin{align}
I(A^n \rar B^n)= \sum_{i=0}^n \Big\{ H(B_i|B^{i-1})- H(B_i|A_i, A_{i-1}, B^{i-1})\Big\}  \label{DI_P1}
\end{align}
where 
%Recall that the entropy of an $p-$dimensional Gaussian distributed RV $X\sim N(m, Q), Q=Q^T \succ 0 $ is given by $H(X)=\frac{1}{2}\log(2\pi e)^p|Q|$. Using this or  substituting (\ref{DOOB_1}) into (\ref{G-LCM_1_G}), and then using one of the well-known  properties of conditional entropy \cite{ihara1993} then
\begin{align}
H(B_i|B^{i-1}, A_i, A_{i-1}) =  H(V_i)= \frac{1}{2}\log(2\pi e)^p|K_{V_i}|, \hso i=0, \ldots, n  . \label{entr_11}
\end{align}
Let $\{(A_i^g, B_i^g, V_i): i=0, \ldots, n\}$ denote a jointly Gaussian  process. There are many ways to show that the optimal channel input distribution is Gaussian, i.e., $\{\pi_i^{A.1}(da_i|a_{i-1}, b^{i-1}) =\pi_i^{g,A.1}(da_i|a_{i-1}, b^{i-1}) : i=0, \ldots, n\}$ , which then implies the joint process $\{(A_i, B_i, V_i)=(A_i^g, B_i^g, V_i): i=0, \ldots, n\}$ is Gaussian, and   the average constraint is satisfied. One approach is to assume a Gaussian channel input conditional distribution and to verify the \'a posteriori recursion is satisfied by a Gaussian distribution $\{{\bf P}^{\pi^{A.1}}(da_{i-1}| b^{i-1}) ={\bf P}^{g,\pi^{g,A.1}}(da_{i-1}| b^{i-1}) : i=0, \ldots, n\}$. An alternative approach is to apply the maximum entropy property of Gaussian distributions, as in \cite{cover-pombra1989} (see also \cite{ihara1993}) which states 
\bea
\sum_{i=0}^n H(B_i|B^{i-1})= H(B^n)  \leq H(B^{g,n})
\eea
and this upper bound is achieved if $\{(A_i, B_i, V_i)=(A_i^g, B_i^g, V_i): i=0, \ldots, n\}$,  the average constraint is satisfied, and (\ref{G-LCM_2_G}) also holds.  Since any  linear combination of RVs is Gaussian if and only if each one of them is Gaussian, then we deduce from (\ref{G-LCM_1_G})-(\ref{G-LCM_6_G}),  and the information structures derived in Section~\ref{Class_CC}, 
%(\ref{DOOB_1}), (\ref{DOOB_11}) 
that any realization of the optimal strategies is linear. \\
{\bf Orthogonal Decomposition.} The orthogonal realization of optimal channel input conditional distribution  is  given by 
\begin{align}
  &A_i^g = U_i^g+ \Lambda_{i,i-1} A_{i-1}^g + Z_i^g \equiv e_i^{A.1}(B^{g,i-1}, A_{i-1}^g, Z_i^g), \hso U_i^g \tri \Gamma^{i-1} B^{g,i-1}, \label{OP_STR_1} \\
 &\hst  \equiv \overline{e}_i^{A.1}(B^{g,i-1})+ \Lambda_{i,i-1} A_{i-1}^g + Z_i^g, \label{OP_STR_1_1}
  \\ 
%&A_0^g=U_0^g+Z_0=g_0^{A.1}(y^{-1})+Z_0^g, \\
 &\overline{e}_i^{A.1}(b^{i-1}) \tri  \Gamma^{i-1} b^{i-1}, \hso i=0, \ldots, n,\\
& Z_i^g  \:\:  \mbox{is  independent of}\:\:  \Big(A^{g, i-1}, B^{g,i-1}\Big), \; Z^{g,i} \hso \mbox{is independent of} \hso V^i, i=0, \ldots, n, \label{new_11}\\
& \Big\{Z_i^g \sim N(0, K_{Z_i}) : i=0,1, \ldots, n\Big\} \: \: \mbox{is an independent Gaussian process} \label{new_12}
  \end{align}
%\begin{align}
%  &A_i^g = U_i^g+ \Lambda_{i,i-1} A_{i-1}^g + Z_i^g, \hso U_i^g \tri \Gamma^{i-1} B^{g,i-1},\hso i=0, \ldots, n, \label{OP_STR_1} \\
% &\hst  \equiv g_i^1(B^{g,i-1})+ \Lambda_{i,i-1} A_{i-1}^g + Z_i^g, \label{OP_STR_1_1} \\ 
%%&A_0^g=U_0^g=g_0^1(B^{-1})+Z_0^g, \\
% &g_i^1(b^{i-1}) \tri  \Gamma^{i-1} b^{i-1}, \hso i=0, \ldots, n
%  \end{align}
for some deterministic matrices $\{ (\Gamma^{i-1}, \Lambda_{i,i-1}): i=0, \ldots, n\}$ of appropriate dimensions,  where (\ref{new_11}) follows from the independence condition (\ref{G-LCM_2_G}). \\ 
Next, we shall show a separation principle between the computation of the strategies $\{\overline{e}_i^{A.1}(\cdot)\equiv \Gamma^{i-1}: i=0, \ldots, n\}$ and $\{ (\Lambda_{i,i-1}, K_{Z_i}): i=0, \ldots, n\}$.\\
To determine the expression of   the directed information pay-off  
(\ref{DI_P1}), we need    the conditional density of $B_i^g$ given $B^{g,i-1}$ for $i=0, \ldots, n$. We determine this by  using  (\ref{G-LCM_1_G}) and strategy  (\ref{OP_STR_1}). Since the conditional density is characterized by the conditional mean and the conditional covariance, we introduce the following quantities.
\begin{align*}
&\widehat{B}_{i|i-1} \tri  {\bf E}^{e^{A.1}} \Big\{ B_i^g \Big| B^{g,i-1}\Big\}, \quad \widehat{A}_{i|i} \tri  {\bf E}^{e^{A.1}} \Big\{ A_i^g \Big| B^{g,i}\Big\},  \\
& { K}_{B_i|B^{i-1}} \tri  {\bf E}^{e^{A.1}} \Big\{ \Big(B_i^g -  \widehat{B}_{i|i-1} \Big)  \Big(B_i^g -  \widehat{B}_{i|i-1} \Big)^T  \Big| B^{g,i-1}\Big\}   \nonumber \\
&P_{i|i} =  {\bf E}^{e^{A.1}}\Big(A_i^g- \widehat{A}_{i|i} \Big)\Big (A_i^g-  \widehat{A}_{i|i}   \Big)^T,\quad i=0, \ldots, n .
\end{align*}
Using the properties of the noise processes, i.e., (\ref{G-LCM_2_G}), (\ref{OP_STR_1})-(\ref{new_12}),
 we obtain the following recursive Kalman-filter estimates and recursions \cite{caines1988}.   
\begin{IEEEeqnarray}{rCL}
&& \widehat{A}_{i|i}=\Lambda_{i,i-1} \widehat{A}_{i-1|i-1}+  U_i^g  + \Delta_{i|i-1}\Big(B_i^g{-}\widehat{B}_{i|i-1}\Big), \hso  \widehat{A}_{-1|-1}=\mbox{given}, \hso i=0, \ldots, n, \label{fil_1} \IEEEeqnarraynumspace  \\
&&\widehat{B}_{i|i-1}  =C_{i-1}B_{i-1}^{g}  +D_{i,i} U_i^g +\overline{\Lambda}_{i,i-1} \widehat{A}_{i-1|i-1}, 
 \\
&&{ K}_{B_i|B^{i-1}} =  \overline{\Lambda}_{i,i-1}  P_{i-1|i-1} \overline{\Lambda}_{i,i-1}^T  + D_{i,i} K_{Z_i}D_{i,i}^T+ K_{V_i} \label{A.2_KF_4} 
 \end{IEEEeqnarray}
where $\{\widehat{A}_{-1|-1}, P_{-1|-1}\}$ are the initial data and
\begin{IEEEeqnarray*}{rCl}
&&\overline{\Lambda}_{i,i-1}\tri D_{i,i}\Lambda_{i,i-1}+ D_{i,i-1}, \hso  i=0, \ldots, n, \nonumber \\
&&P_{i|i}=\Lambda_{i,i-1} P_{i-1|i-1} \Lambda_{i,i-1}^T + K_{Z_{i}}  - \Big( K_{Z_{i}}  D_{i,i}^T + \Lambda_{i,i-1} P_{i-1|i-1} \overline{\Lambda}_{i,i-1}^T\Big) \nonumber \\
&&. \Phi_{i|i-1}   \Big( K_{Z_{i}}  D_{i,i}^T+ \Lambda_{i,i-1} P_{i-1|i-1} \overline{\Lambda}_{i,i-1}^T \Big)^T, \hso P_{-1|-1}=\mbox{given},  \\
&&\Phi_{i|i-1} \tri   \Big[ D_{i,i} K_{Z_{i}}D_{i,i}^T 
+ K_{V_{i}} +\overline{\Lambda}_{i,i-1}  P_{i-1|i-1} \overline{\Lambda}_{i,i-1}^T \Big]^{-1}, \\
&&\Delta_{i|i-1}\tri \Big( K_{Z_{i}}  D_{i,i}^T + \Lambda_{i,i-1} P_{i-1|i-1} \overline{\Lambda}_{i,i-1}^T\Big) \Phi_{i|i-1}
 \end{IEEEeqnarray*}
Note that the above recursions are driven by the innovations process defined by $\big\{\nu^{e^{A.1}}\tri B_i^g{-} \widehat{B}_{i|i-1} : i=0, \ldots, n\}$, which is an  an orthogonal process. Moreover, it is easily verified that the innovations process is independent of the strategy  $\{\overline{e}_i^{A.1}(\cdot): i=0, \ldots, n\}$, and satisfies the following identities.
\begin{align}
\nu_i^{e^{A.1}} \tri &  B_i^g{-} \widehat{B}_{i|i-1}=\overline{\Lambda}_{i,i-1}\Big(A_{i-1}^g{-}  \widehat{A}_{i-1|i-1}\Big)+D_{i,i}Z_i^g+V_i \nonumber \\
=& \nu_i^{e^{A.1}}\Big|_{\overline{e}_i^{A.1}=0}\equiv \nu_i^0, \hst \nu_i^0 \sim N(0, K_{\nu_i^0}), \hso K_{\nu_i^0}=K_{B_i|B^{i-1}} \hso i=0, \ldots, n \label{inn_1}
\end{align}
where the notation $\big\{\nu_i^0: i=0, \ldots, n\big\}$ indicates that  the innovations process is independent of the strategy $\big\{\overline{e}_i^{A.1}(\cdot)\equiv \Gamma^{i-1} : i=0, \ldots, n\}$, i.e., it follows from  (\ref{OP_STR_1}) and   (\ref{fil_1}).\\
From the above equations, we deduce that the conditional covariance $K_{B_i|B^{i-1}}$ is independent of $B^{g,i-1}$ for $i=0, \ldots,n$. Hence, the conditional distribution ${\bf P}^{e^{A.1}}(B_i^g\leq b_i|B^{g,i-1})\sim N(\widehat{B}_{i|i-1}, K_{B_i|B^{i-1}}), i=0,\ldots, n$.  Applying the above two observations  to (\ref{DI_P1}), using (\ref{entr_11}) we obtain 
\begin{align}
I(A^{g,n} \rar B^{g,n})=\sum_{i=0}^n I(A_i^g; B_i^g|B^{g,i-1}) =  \frac{1}{2} \sum_{i=0}^n\log \frac{ | K_{B_i|B^{i-1}}|}{|K_{V_i}|}
\equiv  \sum_{i=0}^n \big\{H(\nu_i^0)-H(V_i)\big\}. \label{DI_CC}
\end{align}
Next, we state the main theorem, which establishes a  separation principle between  the computation of optimal strategies, $\{\Gamma^{i-1}: i=0, \ldots, n\}$ and $\{(\Lambda_{i,i-1}, K_{Z_i}): i=0, \ldots, n\}$, and it  is a generalization of the separation principle of  LQG stochastic optimal control problems with partial information \cite{kumar-varayia1986}, when $C_{i,i-1}=0, i=0, \ldots, n$. 
%By repeating the derivation of Theorem~\ref{CC_CP_A}, we show an analogous converse coding theorem, and that the maximization of (\ref{DI_CC}) over all linear strategies defined by (\ref{OP_STR_1})-(\ref{new_12}) and satisfying the average transmission cost constraint  corresponds to the characterization of FTFI capacity. \\
%Below, we state the main theorem, which establishes a separation principle between the computation of strategies $\{\overline{e}_i^{A.1}(b^{i-1}): i=0, \ldots, n\}$ and $\{(\Lambda_{i,i-1}, K_{Z_i}): i=0, \ldots, n\}$.

\ \

\begin{theorem}(Separation principle)\\
\label{gen_exa}
Consider the G-LCM (\ref{G-LCM_1_G})-(\ref{G-LCM_6_G}). 
Then the following hold.\\
(a) {\it Equivalent Extremum Problem.} The  joint process $\{(A_i, B_i, V_i)=(A_i^g, B_i^g, V_i): i=0, \ldots, n\}$, is jointly Gaussian and satisfies the following equations. 
\begin{IEEEeqnarray}{rCl}
 &&A_i^g=  e_i^{A.1}(B^{g,i-1}, A_{i-1}^g)  + Z_i^g, \hst i=0, \ldots, n, \label{LCM-A.2_7_new_n} \\
 &&\hst =U_i^g+ \Lambda_{i,i-1} A_{i-1}^g   + Z_i^g,\hst U_i^g=\overline{e}_i^{A.1}(B^{g,i-1})=\Gamma^{i-1}B^{g,{i-1}},\label{LCM-A.2_88_b} \\
&& B_i^g   =  C_{i,i-1}B_{i-1}^{g}+  D_{i,i}  U_i^g+    \overline{\Lambda}_{i,i-1} A_{i-1}^g    + D_{i,i} Z_i^g + V_{i},  \label{LCM-A.2_8_a_a} \\
 &&i) \hso Z_i^g  \:\:  \mbox{is  independent of}\:\:  \Big(A^{g, i-1}, B^{g,i-1}\Big), i=0, \ldots, n, \nonumber\\
 && ii) \hso Z^{g,i} \hso \mbox{is independent of} \hso V^i, i=0, \ldots, n, \nonumber\\
&& iii) \hso  \Big\{Z_i^g \sim N(0, K_{Z_i}) : i=0,1, \ldots, n\Big\} \: \: \mbox{is an independent Gaussian process}.  \label{LCM-A.2_7_new_n1}
%&& {\bf E}^{e^1}\Big\{\gamma_i(A_i^{g}, Y_{i-1}^{g})\Big\}.
%&&= {\bf E}^{e^1} \Big\{ \langle U_i^g, R_i U_i^g  \rangle  + 2 \langle \Lambda_{i,i-1} {A}_{i-1}^g, R_i U_i^g\rangle  + tr\Big(K_{Z_i} R_i\Big)  \nonumber \\
% &&+ \langle \Lambda_{i,i-1} {A}_{i-1}^g, R_i \Lambda_{i,i-1} {A}_{i-1}^g  \rangle  + \langle Y_{i-1}^{g}, Q_i Y_{i-1}^{g}\rangle  \Big\} \label{cost_12}
\end{IEEEeqnarray}
Moreover, the average cost is given by
\begin{align}
&\frac{1}{n+1} \sum_{i=0}^n {\bf E}^{e^{A.1}}\Big\{\gamma_i(A_i^{g}, B_{i-1}^{g})\Big\}=\frac{1}{n+1} \sum_{i=0}^n {\bf E}^{e^{A.1}} \Big\{ \langle U_i^g, R_{i,i} U_i^g  \rangle  + 2 \langle \Lambda_{i,i-1} \widehat{A}_{{i-1}|{i-1}}, R_{i,i} U_i^g\rangle    \nonumber \\
 &+ \langle \Lambda_{i,i-1} \widehat{A}_{{i-1|i-1}}, R_{i,i} \Lambda_{i,i-1} \widehat{A}_{{i-1}|{i-1}}  \rangle + tr\Big(K_{Z_i} R_{i,i}\Big) +tr\Big(\Lambda_{i,i-1}^T R_i \Lambda_{i,i-1} P_{i-1|i-1}\Big)+ \langle B_{i-1}^{g}, Q_{i,i-1} B_{i-1}^{g}\rangle  \Big\} \nonumber \\
 & \equiv \frac{1}{n+1} \sum_{i=0}^n {\bf E}^{e^{A.1}}\Big\{\overline{\gamma}_i(U_i, \widehat{A}_{i-1|i-1}, B_{i-1}^{g}, \Lambda_{i,i-1}, K_{Z_i})\Big\}. \label{cost_13}
 \end{align}
 The characterization of FTFI capacity given by 
 \begin{align}
&C_{A^n \rar B^n}^{C.1}(\kappa) =\sup_{\overline{\cal P}_{[0,n]}^{A.1}(\kappa) } \frac{1}{2} \sum_{i=0}^n\log \frac{ | K_{B_i|B^{i-1}}|}{|K_{V_i}|} 
  \label{opt_A.1_1} 
\end{align}
where $\{K_{B_i|B^{i-1}}: i=0, \ldots, n\}$ is given by  (\ref{A.2_KF_4}) and  the average constraint set is defined by 
 \begin{align}
  \overline{\cal P}_{[0,n]}^{A.1}(\kappa)\tri \Big\{e_i^{A.1}(\cdot) \tri \big(\overline{e}_i^{A.1}(\cdot, \cdot), \Lambda_{i,i-1}, K_{Z_i}\big), i=0, \ldots, n:   \frac{1}{n+1}\sum_{i=0}^n {\bf E}^{e^{A.1}} \Big(\gamma_i(A_i^g, B^{g,i-1}) \Big) \leq \kappa \Big\}.
 \end{align}
% For the rest of the statements we let $Q_{i,i-1}=0: i=0, \ldots, n\}$.\\
   (b) {\it Separation of Strategies.} If there exists an interior point to the constraint set $\overline{\cal P}_{[0,n]}^{A.1}(\kappa)$ then the optimal strategy denoted by $\{e^{A.1,*}(\cdot)\equiv (\overline{e}_i^{A.1,*}(\cdot), \Lambda_{i,i-1}^*, K_{Z_i}^*): i=0, \ldots,n\}$ is the solution of the following dual optimization problem.
\begin{align}
J_{A^n \rar B^n}(e^{A.1,*}) =& \inf_{\lambda\geq 0} \sup_{ \big\{  e^{A.1}(\cdot) \tri \big(\overline{e}_i^{A.1}(\cdot, \cdot), \Lambda_{i,i-1}, K_{Z_i}\big), i=0,\ldots, n\big\} } \Bigg\{ \frac{1}{2} \sum_{i=0}^n\log \frac{ | K_{B_i|B^{i-1}}|}{|K_{V_i}|}  \nonumber \\
&- \lambda \Big\{ \sum_{i=0}^n {\bf E}^{e^{A.1}}\Big( \gamma_i(A_i^g,B_{i-1}^g)\Big)-\kappa(n+1) \Big\}\Bigg\} \label{ISDS_6cc_New}
\end{align}  
where  $\lambda$ is the Lagrange multiplier associated with the transmission cost constraint.\\
   % \begin{IEEEeqnarray}{rCl}
%\kappa_{0,n}(C\mathrlap{) \tri} &&   \inf_{ \big(g_i^1(\cdot), \Lambda_{i,i-1}, K_{Z_i}\big), i=0, \ldots, n   : \frac{1}{2} \sum_{i=0}^n\log \frac{ | K_{B_i|B^{i-1}}|}{|K_{V_i}|}     \geq (n+1)C}    {\bf E}^{e^1}\Big\{\sum_{i=0}^n  \gamma_i(A_i^{g}, B_{i-1}^{g})\Big\}   . \label{cap_fb_1_TC_1n_NN}
%\end{IEEEeqnarray} 
Moreover, the following separation holds. \\
(i) The optimal strategy $\{\overline{e}_i^{A.1,*}(\cdot): i=0, \ldots, n\}$ is the solution of the optimization problem 
\bea
J_{A^n \rar B^n}(\overline{e}^{A.1,*}(\cdot), \lambda, \Lambda, K_Z) \tri \inf_{\overline{e}_i^{A.1}(\cdot): i=0, \ldots, n} \lambda {\bf E}^{e^{A.1}}\Big\{\sum_{i=0}^n  \gamma_i(A_i^{g}, B_{i-1}^{g})  \Big\}
\eea
for a fixed $\lambda, \{\Lambda_{i,i-1}, K_{Z_i}: i=0, \ldots, n\}$.\\
(ii) The optimal strategy $\{\Lambda_{i,i-1}^*, K_{Z_i}^*: i=0, \ldots, n\}$ is the solution of (\ref{ISDS_6cc_New}) for  $\{\overline{e}_i^{A.1}(\cdot)=\overline{e}_i^{A.1,*}(\cdot): i=0, \ldots, n\}$.\\
(c) {\it Optimal Strategies.} Define the augmented state variable as follows.
\begin{align}
\overline{B}_{i-1}^g \tri \left[ \begin{array}{c} B_{i-1}^g \\ \widehat{A}_{i-1|i-1}\end{array} \right], 
\; i=0, \ldots, n. \nonumber
\end{align}
 Any candidate of the  strategy $\{\overline{e}_i^{A.1}(B^{g,i-1}):i=0, \ldots, n\}$ is of the form 
\begin{align}
\overline{e}_i^{A.1}(B^{g,i-1})  \equiv & \overline{e}_i^{A.1}(\overline{B}_{i-1}^{g})  \tri  \Gamma_{i,i-1}^1 B_{i-1}^g + \Gamma_{i,i-1}^{2}\widehat{A}_{i-1|i-1}, \label{sep_SOC_1} \\
=& \overline{\Gamma}_{i,i-1} \overline{B}_{i-1}^g,  
\; i=0, \ldots, n. \nonumber
\end{align}
where the components of $\{\overline{B}_i^g: i=0, \ldots, n\}$  satisfy (\ref{fil_1}), (\ref{inn_1}), the augmented system is 
\begin{align}
\overline{B}_i^g= &\overline{F}_{i,i-1} \overline{B}_{i-1}^g + \overline{E}_{i,i-1}U_i^g + \overline{G}_{i,i-1}\nu_i^{e^{A.1}}, \hso i=0, \ldots, n, \label{SEP1_SOC} \\
\overline{F}_{i,i-1}\tri & \left[ \begin{array}{cc} C_{i,i-1} & \overline{\Lambda}_{i,i-1} \\ 0 & \Lambda_{i,i-1} \end{array}\right], \; \overline{E}_{i,i-1}\tri  \left[ \begin{array}{c} D_{i,i} \\ I\end{array}\right], \hso 
\overline{G}_{i,i-1}\tri \left[ \begin{array}{c} I \\ \Delta_{i|i-1} \end{array}\right]
\end{align}
the average cost is
\begin{align}
&{\bf E}^{e^{A.1}}\Big\{\sum_{i=0}^n\gamma_i(A_i^{g}, B_{i-1}^{g})\Big\}\equiv {\bf E}^{e^{A.1}}\Big\{\sum_{i=0}^n \overline{\gamma}_i(U_i^{g}, \overline{B}_{i-1}^{g}, \Lambda_{i,i-1}, K_{Z_i})\Big\} \nonumber \\
&  
 \tri  {\bf E}^{e^{A.1}} \Big\{ \sum_{i=0}^n\Big(\left[ \begin{array}{c} \overline{B}_{i-1}^g \\ U_i^g \end{array} \right]^T \left[\begin{array}{cc} \overline{M}_{i,i-1} & \overline{L}_{i,i-1} \\ \overline{L}_{i,i-1}^T & \overline{N}_{i,i-1}\end{array} \right] \left[ \begin{array}{c}\overline{B}_{i-1}^g \\ U_i^g \end{array} \right]  + tr\big(K_{Z_i} R_{i,i}\big) +tr\big(\Lambda_{i,i-1}^T R_{i,i} \Lambda_{i,i-1} P_{i-1|i-1}\big) \Big) \Big\},
\\
& \overline{M}_{i,i-1}\tri \left[ \begin{array}{cc} Q_{i,i-1} & 0 \\ 0 & \Lambda_{i,i-1}^T R_{i,i} \Lambda_{i,i-1}\end{array}\right], \hso \overline{L}_{i,i-1}\tri \left[\begin{array}{c} 0 \\ \Lambda_{i,i-1}^T R_{i,i} \end{array} \right],\hso \overline{N}_{i,i-1} \tri R_{i,i}. \label{CISC}
 \end{align}
and  the following hold. \\
(1) For a fixed $\lambda, \{\Lambda_{i,i-1}, K_{Z_i}: i=0, \ldots, n\}$, the optimal strategy $\{U_i^{g,*}=\overline{e}_i^{A.1,*}(B^{g,i-1})\equiv \overline{e}_i^{A.1,*}(\overline{B}_{i-1}^{g}) :i=0, \ldots, n\}$  is the solution of the fully observable classical stochastic optimal control problem  
\begin{align}
J_{A^n \rar B^n}^1(\overline{e}^{A.1,*}(\cdot), \lambda, \Lambda, K_Z) =   \inf_{\overline{e}_i^{A.1}(\cdot): i=0, \ldots, n } \lambda   {\bf E}^{e^{A.1}}\Big\{\sum_{i=0}^n  \overline{\gamma}_i(U_i^{g}, \overline{B}_{i-1}^{g}, \Lambda_{i,i-1}, K_{Z_i})   \Big\} \label{SOCP_1}
\end{align} 
where $\{\overline{B}_i^g: i=0, \ldots, n\}$ satisfy recursion (\ref{SEP1_SOC}). Moreover, the optimal strategy $\{U_i^{g,*}=\overline{e}_i^{A.1,*}(\overline{B}_{i-1}^{g}):i=0, \ldots, n\}$ is  given by  the following equations.
\begin{align}
&\overline{e}_i^{A.1,*}(\overline{b}_{i-1})=\overline{\Gamma}_{i,i-1} \overline{b}^{i-1}
={-}\Big( \lambda\overline{N}_{i,i-1} {+} \overline{E}_{i,i-1}^T \Sigma(i{+}1)\overline{E}_{i,i-1} \Big)^{-1} \overline{E}_{i,i-1}^T \Sigma(i{+}1) \overline{F}_{i,i-1}\overline{b}_{i-1},  \; i=0, \ldots, n-1,\label{opt_con_1_1_n}  \\ 
 &\overline{e}_n^{A.1,*}(\overline{b}_{n-1})=0 \label{opt_con_1_1_n_n}
\end{align}
where  the symmetric positive semi-definite matrix $\{\Sigma(i): i=0,\ldots, n\}$ satisfies the matrix difference Riccati equation 
\begin{align}
\Sigma(i) {=}& \overline{F}_{i,i{-}1}^T  \Sigma(i{+}1) \overline{F}_{i,i{-}1}  {-} \Big( \overline{F}_{i,i{-}1}^T \Sigma(i{+}1)\overline{E}_{i,i-1}  {+} \lambda\overline{L}_{i,i-1}\Big) \nonumber \\
& .\Big(\lambda \overline{N}_{i,i-1}+ \overline{E}_{i,i-1}^T \Sigma(i+1)\overline{E}_{i,i-1} \Big)^{-1} \Big(\overline{E}_{i,i-1}^T \Sigma_{i,i-1}\overline{F}_{i,i-1} { +}\lambda \overline{L}_{i,i-1}^T\Big)^T{+} \lambda \overline{M}_{i,i-1}^T,\hso i=0, \ldots, n-1, \\ \Sigma(n)=&\lambda \overline{M}_{n,n-1}^T.
\end{align}
and the optimal pay-off is given by
\begin{align}
&J_{A^n \rar B^n}^1(\overline{e}^{A.1,*}(\cdot), \lambda, \Lambda, K_Z)= \sum_{j=0}^n \Big\{ tr\big(\Phi_{j|j-1}\Sigma(j)\big)+ \lambda tr\big(K_{Z_j}R_{j,j}\big) +\lambda tr \big(\Lambda_{j,j-1}^T R_{j,j} \Lambda_{j,j-1} P_{j-1|j-1}\big) \nonumber  \\
&+ tr\big(\Delta_{j|j-1}K_{B_j|B^{j-1}} \Delta_{j,j-1} \Sigma(j)\big) \Big\}+\langle \overline{B}_{0|-1}, \Sigma(0) \overline{B}_{0|-1}\rangle.
\end{align}
(2) The optimal strategies  $\{(\Lambda_{i,i-1}^*, K_{Z_i}^*): i=0, \ldots, n\}$ are the solutions of the optimization problem
\begin{align}
J_{A^n \rar B^n}(e^{A.1,*}) =  \inf_{\lambda\geq 0} \sup_{ \big\{  \big(\Lambda_{i,i-1}, K_{Z_i}\big), i=0,\ldots, n\big\} } \Bigg\{ \frac{1}{2} \sum_{i=0}^n\log \frac{ | K_{B_i|B^{i-1}}|}{|K_{V_i}|}  
- \lambda \Big\{J_{A^n \rar B^n}^1(\overline{e}^{A.1,*}(\cdot), \Lambda, K_Z) - \kappa(n+1) \Big\}\Bigg\}. \label{SEC_S}
\end{align} 
 \end{theorem}
\begin{proof}(a) Equations  (\ref{LCM-A.2_7_new_n})-(\ref{LCM-A.2_7_new_n1})   follow from the statements prior to the theorem.  The average constraint (\ref{cost_13}) follows from  (\ref{G-LCM_4_G}) and (\ref{LCM-A.2_7_new_n})-(\ref{LCM-A.2_7_new_n1}), using  the reconditioning property of expectation.   (\ref{opt_A.1_1})     is due to  (\ref{DI_CC}).  (b) 
(\ref{ISDS_6cc_New}) follows from duality theory, in view of the convexity of the optimization problem.   (i), (ii) These  follow from the observation that directed information expressed in terms of the logarithm in the right hand side of (\ref{ISDS_6cc_New})  depends on $\{(\Lambda_{i,i-1}, K_{Z_i}): i=0, \ldots, n\}$ and  not on $\{\overline{e}_i^{A.1}(\cdot): i=0, \ldots, n\}$, which then  implies the optimization problem in (\ref{ISDS_6cc_New}) over $\{e_i^{A.1}(\cdot): i=0, \ldots, n\}$ can be decomposed as stated.  (c) Consider the output process   $\{B_i^g: i=0, \ldots, n\}$,   expressed in terms of the innovations process as follows.
\begin{align}
B_i^g=\widehat{B}_{i|i-1}+ \nu_i^{e^{A.1}}=C_{i,i-1}B_{i-1}^g + \overline{e}_i^{A.1}(B^{g,i-1})+ \Lambda_{i,i-1} \widehat{A}_{i-1|i-1}+ \nu_i^{e^{A.1}}, \hso i=0, \ldots, n. \label{aug_1}
\end{align}
Also recall that $\{\widehat{A}_{i|i}: i=0, \ldots, n\}$ satisfies (\ref{fil_1}) and it is driven by the innovations process (\ref{inn_1}), as follows. 
\begin{align}
\widehat{A}_{i|i}=\Lambda_{i,i-1} \widehat{A}_{i-1|i-1}+  U_i^g  + \Delta_{i|i-1} \nu_i^{A.1}, \hso  \widehat{A}_{-1|-1}=\mbox{given}, \hso i=0, \ldots, n.
\end{align}
Since the following Markov property holds
\begin{align}
{\bf P}^{e^{A.1}}&(db_i, d\widehat{a}_{i|i}| \{B_j^{g}= b_j, \widehat{A}_{j|j}= \widehat{a}_{j|j}: j=0, \ldots, i-1\},\{U_j^{g}= u_j: j=0, \ldots, i\}) \nonumber \\
&={\bf P}^{e^{A.1}}(db_i, d\widehat{a}_{i|i}|B_{i-1}^{g}= b_{i-1}, \widehat{A}_{i-1|i-1}, U_i^g=u_i), \hso i=0, \ldots, n
\end{align}
and the average cost   in (\ref{cost_13}), specifically, $\overline{\gamma}_i(U_i, \widehat{A}_{i-1|i-1}, B_{i-1}^{g}, \Lambda_{i,i-1}, K_{Z_i})\Big\}$ is a function of $\{U_i, \widehat{A}_{i-1|i-1}, B_{i-1}^{g}\}$,  
 then  $\{(B_i^g, \widehat{A}_{i|i}): i=0, \ldots, n\}$ is a sufficient statistics for strategy $\{\overline{e}_i^{A.1}(\cdot): i=0, \ldots, n\}$, that is,  for each $i$, then $U_i^g=\overline{e}_i^{A.1}(B_{i-1}^g,\widehat{A}_{i-1|i-1})$, for $i=0, \ldots, n$. This also follows from Theorem~\ref{suf_st}. Consequently, (\ref{sep_SOC_1})-(\ref{CISC}) are obtained by simple algebra. Finally, by (\ref{inn_1})  the innovations process $\{\nu_i^{\overline{e}^{A.1}}\equiv \nu_i^0: i=0, \ldots, n\}$,  is independent of strategy $\{\overline{e}_i^{A.1}(\cdot): i=0, \ldots, n\}$, and hence the rest of the equations follow directly from the solution of  LQG partially obervable stochastic optimal control problems  \cite{caines1988} and (b), that is, the statements under (1) follow from the fact that  $\{(B_i^g, \widehat{A}_{i|i}): i=0, \ldots, n\}$ is the state process of a completely observable stochastic optimal control problem (\ref{SOCP_1}).  
\end{proof}

\ \

\begin{remark}(Comments on the separation principle of Theorem~\ref{gen_exa})\\
(a) The solution methodology given in Theorem~\ref{gen_exa} states that there is a separation principle, in the sense  that the optimal strategy $\{\overline{e}_i^{A.1,*}(\cdot)=: i=0, \ldots, n\}$ is obtained for fixed strategies $\{\Lambda_{i,i-1},K_{Z_i}: i=0, \ldots, n\}$, while the optimal  strategies $\{\Lambda_{i,i-1}^*,K_{Z_i}^*: i=0, \ldots, n\}$ are found from the  optimization problem  (\ref{SEC_S}), evaluated at the optimal strategy $\{\overline{e}_i^{A.1,*}(\cdot)=: i=0, \ldots, n\}$. This is analogous to the concept of Person-by-Person (PbP) optimality of decentralized stochastic optimal control or decision theory in cooperative optimization \cite{charalambousGFIS_2013,charalambous-ahmedFIS_Parti2012}.  Moreover,  this separation principle is a generalization of the separation principle between estimation and control of LQG partially observable stochastic optimal control problems. \\
(b)  If $C_{i,i-1}=0, Q_{i,i=0}=0, i=0, \ldots, n$ then the sample path pay-off in  (\ref{cost_13}) and (\ref{aug_1}) reduce to the following expressions.
\begin{align}
&\overline{\gamma}_i(U_i, \widehat{A}_{i-1|i-1}, B_{i-1}^{g}, \Lambda_{i,i-1}, K_{Z_i})= \tilde{\gamma}_i(U_i, \widehat{A}_{i-1|i-1}, \Lambda_{i,i-1}, K_{Z_i}), \hso i=0, \ldots, n, \\
 &B_i^g=\widehat{B}_{i|i-1}+ \nu_i^{e^{A.1}}=\overline{e}_i^{A.1}(B^{g,i-1})+ \Lambda_{i,i-1} \widehat{A}_{i-1|i-1}+ \nu_i^{e^{A.1}}
\end{align} 
  hence $\{\widehat{A}_{i|i}: i=0, \ldots, n\}$ is a sufficient statistics for strategy $\{\overline{e}_i^{A.1}(\cdot): i=0, \ldots, n\}$, that is, 
$U_i^g=\overline{e}_i^{A.1}(B^{g,i-1})=\tilde{e}_i^{A.1}(\widehat{A}_{i-1|i-1})$, i.e., $\overline{B}_{i-1}^g = \widehat{A}_{i-1|i-1}$, for  $i=0, \ldots, n$ (see (\ref{SEP1_SOC}).
\end{remark}

\subsection{Multiple Input Multiple Output Gaussian Linear Channel Models with Arbitrary Memory}
\label{G-LCM_CP_G}
The methodology and results obtained in Section~\ref{C-LCM_C_one} admit generalizations to the following models.

 %We show claim (ii), by considering  the following Gaussian-LCM-C, which is a  degenerate version of the NCM-D of Example~\ref{gne_exa}\footnote{$L=M=0$ corresponds to a memoryless channel.}.  
\noi{\bf  Gaussian-Linear Channel Model with Arbitrary Memory.} Consider a  (G-LCM) with quadratic cost function, and arbitrary memory, defined as follows.
\begin{align}
&B_i   = \sum_{j=1}^M C_{i,i-j} B_{i-j} + \sum_{j=0}^L D_{i,i-j}  A_{i-j} + V_{i}, \hso B_{-M}^{-1}=b_{-M}^{-1}, \hso A_{-L}^{-1}=a_{-L}^{-1},  \hso i= 0, \ldots, n, \label{G-LCM-C_1} \\
& \hso \equiv C_{M}(i)B_{i-M}^{i-1}+ D_{L}(i)A_{i-L}^i+ V_i, \nonumber \\
&\frac{1}{n+1} \sum_{i=0}^n  {\bf E} \Big\{  \gamma_i^{C.L,M}(A_{i-L}^i, B_{i-M}^{i})     \Big\} \tri  \frac{1}{n+1}\sum_{i=0}^n {\bf E} \Big\{ \langle A_{i-L}^i, R_{L}(i) A_{i-L}^i  \rangle+ \langle B_{i-M}^{i},Q_{M}(i)B_{i-M}^{i} \rangle  \Big\}\leq \kappa,  \label{LCM-A.1_a_B} \\
& R_L(i)=R_L^T(i)\succ 0 \in {\mb R}^{(L+1) q \times (L+1) q},\; Q_{M}(i)=Q_M^T(i)\succeq 0 \in {\mb R}^{(M+1) p \times (M+1) p},  \\
& {\bf P}_{V_i|V^{i-1}, A^i}={\bf P}_{V_i}, \;  V_i \sim N(0, K_{V_i}). \label{LCM-C_4} 
\end{align} 
By (\ref{LCM-C_4}),  the channel distribution   is Gaussian given by 
\begin{align}
&{\mb P}\Big\{B_i \leq b_i | B^{i-1}=b^{i-1}, A^i=a^i\Big\} \sim N(C_{M}(i) \;b_{i-M}^{i-1}+ D_L(i)\; a_{i-L}^i, \: K_{V_i}),  \hso  i=0,1, \ldots, n \label{LCM-A.2_1}
\end{align}
From Section~\ref{class_C},  we directly obtain that the optimal channel input conditional distribution occurs in the following set.
\begin{align}
\overline{\cal P}_{[0,n]}^{A.L}(\kappa) \tri & \Big\{ \pi_i^{A.L}(da_i|a_{i-L}^{i-1}, b^{i-1}), i=0, \ldots, n: \frac{1}{n+1}\sum_{i=0}^n {\bf E}^{\pi^{A.L}} \Big(\gamma_i^{C.L,M}(A_{i-L}^i, B_{i-M}^{i}) \Big) \leq \kappa \Big\} \label{NCM-L2_1}
\end{align}
and that the characterization of FTFI capacity is given by the following expression.
\begin{align}
{C}_{A^n \rar B^n}^{FB,C.L} (\kappa) 
 = \sup_{\overline{\cal P}_{[0,n]}^{A.L}(\kappa)}  \Big\{ \sum_{i=0}^n H(B_i|B^{i-1})\Big\} - H(V^n) \label{LCM_L.2_3}
\end{align}
where 
\begin{align}
{\mb P}\Big\{B_i \leq b_i | B^{ i-1}=&b^{i-1}\Big\}=\int_{{\mathbb A}_{i-L}^{i}} {\mb P}\Big\{V_{i}\leq b_i - C_{M}(i) \;b_{i-M}^{i-1}+ D_L(i)\; a_{i-L}^i  \Big\} \nonumber \\
& \otimes \pi_i^{A.L}(da_i|a_{i-L}^{i-1}, b^{i-1})\otimes {\bf P}^{\pi^{A.L}}(da_{i-L}^{i-1}|b^{i-1}),  \hso  i=0,1, \ldots, n. \label{LCM-A.2_2}
\end{align}
We can   show, as in Section~\ref{C-LCM_C_one}, that the optimal channel input distribution satisfying the average transmission cost constraint is Gaussian.\\

% i.e., $\Big\{ \pi_i^{A.L}(da_i | a_{i-L}^{i-1}, b^{i-1})=\pi_i^{g, A.L}(da_i | a_{i-L}^{i-1}, b^{i-1}): i=0,\ldots,n\Big\}$, and the corresponding joint process $\Big\{(A_i, B_i, V_i)=(A_i^g, B_i^g, V_i): i=0, \ldots, n\Big\}$ is jointly Gaussian. \\ 

\begin{theorem}(MIMO G-LCM with arbitrary memory)\\
\label{C-LCM_C}
 Consider the G-LCM defined by (\ref{G-LCM-C_1})-(\ref{LCM-C_4}). \\
 Then the following hold.\\
 (a) The optimal channel input conditional distribution is Gaussian distributed, denoted by $\Big\{ \pi_i^{A.L}(\cdot | a_{i-L}^{i-1}, b^{i-1})=\pi_i^{g, A.L}(\cdot | a_{i-L}^{i-1}, b^{i-1}): i=0,\ldots,n\Big\}$, and the corresponding joint process is jointly Gaussian distributed, denoted by $\Big\{(A_i, B_i, V_i)=(A_i^g, B_i^g, V_i): i=0, \ldots, n\Big\}$. \\
Moreover, the characterization of FTFI capacity is given by 
\begin{align}
{C}_{A^n \rar B^n}^{FB,G-A.L} (\kappa) 
 =&  \sup_{  \overline{\cal P}_{[0,n]}^{G-A.L}(\kappa)  }  \Big\{ \sum_{i=0}^n H(B_i^g|B^{g,i-1})\Big\} - H(V^n)    \label{LCM_A.2_3} 
 \end{align}
 where 
\begin{align}
&\overline{\cal P}_{[0,n]}^{G-A.L}(\kappa) \tri  \Big\{ \pi_i^{g,A.L}(da_i|a_{i-L}^{i-1}, b^{i-1}), i=0, \ldots, n: \frac{1}{n+1}\sum_{i=0}^n {\bf E}^{\pi^{g,A.L}} \Big(\gamma_i^{C.L,M}(A_{i-L}^{g,i}, B_{i-M}^{g,i}) \Big) \leq \kappa \Big\} \\ 
& {\mb P}\Big\{B_i^g \leq b_i | B^{g, i-1}=b^{i-1}\Big\}=\int_{{\mathbb A}_{i-L}^{i}} {\mb P}\Big\{V_{i}\leq b_i - C_{M}(i) \;b_{i-M}^{i-1}+ D_L(i)\; a_{i-L}^i  \Big\} \nonumber \\
& \hst \otimes \pi_i^{g, A.L}(da_i|a_{i-L}^{i-1}, b^{i-1})\otimes {\bf P}^{g, \pi^{A.L}}(da_{i-L}^{i-1}|b^{i-1}),  \hso  i=0,1, \ldots, n. \label{LCM-A.2_4} 
\end{align}
(b) An equivalent FTFI characterization is given by the following expressions.
\begin{align} 
 &A_i^g= \Gamma^{i-1} B^{g,i-1} + \sum_{j=1}^L \Lambda_{i, i-j} A_{i-j}^g+ Z_i^g,   \hst  i=0,1, \ldots, n  \label{LCM-A.2_7} \\
 &\hst   = \Gamma^{i-1} B^{g,i-1} + \Lambda_{L}(i) A_{i-L}^{g,i-1}+ Z_i^g,  \label{LCM-A.2_7_new} \\
 &B_i^g   =\sum_{j=1}^{M} C_{i,i-j} B_{i-j}^g +\sum_{j=0}^L D_{i,i-j}A_{i-j}^g + V_{i},     \hso i= 0, \ldots, n, \label{LCM-A.2_8}  \\
 & \hst =  C_{M}(i) B_{i-M}^{g,i-1} +D_{L}(i) A_{i-L}^{g,i} + V_{i} \\ 
 &i) \hso Z_i^g  \:\:  \mbox{is  independent of}\:\:  \Big(A^{g, i-1}, B^{g,i-1}\Big), i=0, \ldots, n, \\
 & ii) \hso Z^{g,i} \hso \mbox{is independent of} \hso V^i, i=0, \ldots, n \\
& iii) \hso  \Big\{Z_i^g \sim N(0, K_{Z_i}), K_{Z_i} \succeq 0 : i=0,1, \ldots, n\Big\} \: \: \mbox{is a Gaussian process}\\
&{C}_{A^n \rar B^n}^{FB,G-A.L} (\kappa) 
 = \sup_{\big\{  \big(\Gamma^{i-1}, \Lambda_{L}(i), K_{Z_i}\big), i=0,\ldots, n    \big\} \in {\cal E}_{[0,n]}^{IL-G-A.L}(\kappa)  }   \sum_{i=0}^n H(B_i^g|B^{g,i-1}) - H(V^n)  \label{opt_A.1}\\
 & {\cal E}_{[0,n]}^{IL-G-A.L}(\kappa)\tri \Big\{\big(\Gamma^{i-1}, \Lambda_{L}(i), K_{Z_i}\big), i=0, \ldots, n:    \frac{1}{n+1}\sum_{i=0}^n {\bf E} \Big(\gamma_i^{C.L, M}(A_{i-L}^{g,i}, B_{i-M}^{g,i}) \Big) \leq \kappa \Big\}, \\
& \gamma_i^{C.L,M}(A_{i-L}^{g,i}, B_{i-M}^{g,i})
 \tri  \langle A_{i-L}^{g,i}, R_L(i) A_{i-L}^{g, i} \rangle +  \langle B_{i-M}^{g,i}, Q_M(i) B_{i-M}^{g,i}\rangle, \hso i=0, \ldots, n .
\end{align}
\end{theorem}
\begin{proof} The derivation is done precisely as in Section~\ref{C-LCM_C_one}, hence it is omitted.
\end{proof}

We can proceed further to derive the anlog of the material in Section~\ref{C-LCM_C_one}.

In the next remark we relate the above theorem to the Cover and Pombra \cite{cover-pombra1989} characterization, and illustrate some of the fundamental differences.   \\

\begin{remark}(MIMO G-LCM and relation to Cover and Pombra \cite{cover-pombra1989})\\
By Theorem~\ref{C-LCM_C}, and assuming,  without loss of generality  the initial data are $B_{-M}^{-1}=b_{-M}^{-1}=0, A_{-L}^{-1}=a_{-L}^{-1}=0$, and $P_0(da_0|a_{-L}^{-1}, b^{-1})=P_0(da_0), P_1(da_1|a_0, a_{1-L}^{-1},b_0, b^{-1})=P_1(da_1|a_0, b_0)$, etc., we can express the decomposition   (\ref{LCM-A.2_7}) in terms of $\{(V_i, Z_i^g): i=0, \ldots, n\}$, by simple recursive substitution, as follows. 
\begin{align}
A_i^g=& \sum_{j=0}^{i-1}{\Gamma}_{i,j} B_j^g + \sum_{j=1}^L \Lambda_{i,i-j}A_{i-j}^g+ Z_{i}^g, \hso A_0^g=Z_0^g,   \hst i=1, \ldots, n, \label{R_CP1989_1} \\
=& \sum_{j=0}^{i-1}\overline{\Gamma}_{i,j} V_j + \overline{Z}_{i}, \hso \overline{Z}_i \tri \sum_{j=0}^i \overline{\Delta}_{i,j} Z_{j}^g \label{R_CP1989}
\end{align}
for appropriately chosen matrices $\{\overline\Gamma_{i,j}: i=0, \ldots, n, j=0, \ldots, i-1\}, \{\overline{\Delta}_{i,j}: i=0, \ldots, n, j=0, \ldots, i\}$. \\
Clearly, in (\ref{R_CP1989_1}) the process $\{Z_i^g: i=0, \ldots, n\}$ is an orthogonal or independent innovations process, while in     the alternative equivalent expression (\ref{R_CP1989}),  the process $\{\overline{Z}_i:i=0, \ldots, n\}$  is not an independent innovations process. In fact, the realization of the process $\{A_i^g: i=0, \ldots, n\}$ given by (\ref{R_CP1989}), in terms of $\{\overline{Z}_i: i=0, \ldots, n\}$,  is analogous to  the realization derived by Cover and Pombra \cite{cover-pombra1989} (see (\ref{c-p1989})-(\ref{cp1989})), where  the process $\{\overline{Z}_i: i=0, \ldots, n\}$ is not an orthogonal process. \\ 
Since the objective is to compute ${C}_{A^n \rar B^n}^{FB,G-A.L} (\kappa) $ given by   (\ref{opt_A.1}), as in Theorem~\ref{gen_exa}, then decomposition   (\ref{LCM-A.2_7}) or (\ref{R_CP1989_1}), where the process $\{Z_i^g: i=0, \ldots, n\}$ is an orthogonal process, is much simpler to analyze, compared to  decomposition  (\ref{R_CP1989}), where $\{\overline{Z}_i: i=0, \ldots, n\}$ is correlated. This is possibly one of the main  reason,  which prevented many of the past attempts to solve the Cover and Pombra \cite{cover-pombra1989} non-stationary non-ergodic characterization explicitly, or any of its stationary ergodic  variants \cite{yang-kavcic-tatikonda2007,kim2010}.    \\
\end{remark}

Finally, we note that although, the emphasis  is to illustrate applications in MIMO G-LCM, the methodology applies  to arbitrary channel models, irrespectively of the type of alphabet spaces and channel noise distributions.

\section{General Discrete-Time Recursive Nonlinear Channel Models}
\label{DTRM}
In this section, we show the following.

\begin{description}
\item[(i)] The information structures of channel distributions (\ref{CD_C2})-(\ref{CD_C5}) are sufficient to derive the information structures of  Nonlinear Channel Models (NCMs) driven by arbitrary distributed and correlated noise processes.  
 
 \item[(ii)] The optimal channel input conditional  distributions of Multiple-Input Multiple Output (MIMO) Gaussian Linear Channel Models (G-LCM), driven by correlated Gaussian noise processes, can be dealt with as in Section~\ref{exa_gen}.
\end{description}
Claim (i) illustrates that  many of the existing channels investigated in the literature, for example,  \cite{cover-pombra1989,kim2008,chen-berger2005,yang-kavcic-tatikonda2005,permuter-cuff-roy-weissman2010,permuter-asnani-weissman2013,elishco-permuter2014,kourtellaris-charalambous2015}, induce channel distributions of Class A, B or C.
Claim  (ii) generalizes the Cover and Pombra \cite{cover-pombra1989} characterization (\ref{cp1989_a}) of feedback capacity of  non-nstationary non-ergodic Additive Gaussian channels driven by correlated noise.

\subsection{Nonlinear Discrete-Time Recursive Channel Models}
\label{NRCM_G}
First, we illustrate that channel distributions of  Class A, B or C, i.e.,  (\ref{CD_C2})-(\ref{CD_C5}), are induced  by various nonlinear channel models (NCM), and include  nonlinear and linear time-varying  Autoregressive models,  nonlinear and linear  channel models expressed in state space form \cite{caines1988}, and many of the existing channels investigated in the literature, for example,  \cite{cover-pombra1989,kim2008,chen-berger2005,yang-kavcic-tatikonda2005,permuter-cuff-roy-weissman2010,permuter-asnani-weissman2013,elishco-permuter2014,kourtellaris-charalambous2015}, and non-nstationary non-ergodic Additive Gaussian channels driven by correlated noise \cite{cover-pombra1989}.

\ \

\begin{definition}(Nonlinear channel models and transmission costs) \\
\label{R-exa_A_D}
(a) NCM-A. Nonlinear Channel Models A (NCM-A) are defined by   nonlinear recursive models and   transmission cost functions, as follows.
\begin{align}
&B_i   =h_i^{A}(B^{i-1}, A_{i-L}^i, V_i),  \hso  B^{-1}=b^{-1},\; A_{-L}^{-1}=a_{-L}^{-1}, \hso  i=0, \ldots, n, \label{NCM-A_D} \\
&\frac{1}{n+1}\sum_{i=0}^n {\bf E} \Big\{ \gamma_i^{A.N}(A_{i-N}^i, B^{i}) \Big\}  \leq \kappa
 \label{NCM-A.D-IC}
\end{align} 
where $\{V_i: i=0,1, \ldots, n\}$ is the noise process, and the following assumption holds.\\
 { Assumption A.(i).}  The alphabet spaces include any of the following.
\begin{align}
&\mbox{Continuous Alphabets:} \;
{\mb B_i}\tri {\mb R}^p, \; {\mb A_i}\tri {\mb R}^q, \;  {\mb V}_i \tri {\mb R}^r, \hso i=0, 1, \ldots, n; \label{CA_A.D} \\
&\mbox{Finite Alphabets:} \;
{\mb B_i}\tri \big\{1, \ldots, p\big\}, \; {\mb A_i}\tri \big\{1, \ldots, q\},  \; {\mb V}_i \tri \big\{1, \ldots, r\big\}, \; i=0, 1, \ldots, n;  \label{DA_A.D} \\
&\mbox{Combinations of Continuous and Discrete (Finite or Countable) Alphabets.}
\end{align}
{Assumption A.(ii).} $h_i^{A}: {\mb B}^{i-1} \times {\mb A}_{i-L}^i  \times {\mb V_i} \longmapsto {\mb B}_i, \gamma_i^{A.N}:  {\mathbb A}_{i-N}^i\times  {\mb B}^{i} \longmapsto {\mb A}_i$  and $h_i^A(\cdot, \cdot, \cdot), \gamma_i^{A.N}(\cdot, \cdot)$ are  measurable functions, for $i=0, 1, \ldots, n$; \\
{Assumption A.(iii).}  The noise process $\{V_i: i=0, \ldots, n\}$  satisfies  conditional independence condition
\bea
{\bf P}_{V_i|V^{i-1}, A^i}(dv_i|v^{i-1}, a^i)=  {\bf P}_{V_i}(dv_i)-a.a.(v^{i-1}, a^i),\hso i=0, \ldots, n. \label{CIC_A}
\eea

(b) Nonlinear Channel Models A.B and B.A are as follows.\\
(b.1) {NCM-A.B.}  Nonlinear Channel Models A.B  (NCM-A.B) correspond to nonlinear recursive models NCM-A, with $\gamma_i^A(A_{i-N}^i, B^{i})$ in  (\ref{NCM-A_D})   replaced by  $\gamma_i^{B.K}(A^i, B_{i-K}^{i}), i=0, \ldots, n$, and Assumptions A.(i)-A.(iii) hold with appropriate changes. \\
(b.2) {NCM-B.A.}  Nonlinear Channel Models B.A (NCM-B.A) correspond to nonlinear recursive models NCM-A, with $h^A(B^{i-1}, A_{i-L}^i, V_i)$ in (\ref{NCM-A_D})    replaced by  $h_i^B(B_{i-M}^{i-1}, A^i, V_i), i=0, \ldots, n$, and Assumptions A.(i)-A.(iii) hold with appropriate changes.

(c) {NCM-C.} Nonlinear Channel Models C (NCM-C) are defined as follows.
\begin{align}
& B_i   = h_i^{C}(B_{i-M}^{i-1}, A_{i-L}^i, V_i), \hso B_{-M}^{-1}=b_{-M}^{-1},\; A_{-L}^{-1}=a_{-L}^{-1}, \; i=0, \ldots, n,  \label{NCM-C_D_C} \\
 &\frac{1}{n+1}\sum_{i=0}^n {\bf E} \Big\{ \gamma_i^{C.N,K}(A_{i-N}^i, B_{i-K}^{i}) \Big\}  \leq \kappa \label{NCM-A.D-IC_C}   %hst I \tri \max\{L, N\}, \: J\tri \max\{M, L, N, K\}, 
\end{align} 
where $\{V_i: i=0,1, \ldots, n\}$ is the noise process,  and  Assumptions A.(i)-A.(iii) hold with appropriate changes.

(d) {NCM-D.}  Nonlinear Channel Models D (NCM-D) correspond to any one of NCM-A, NCM-A.B, NCM-B.A,  NCM-C, with  recursive function  $h_i^D(\cdot, \cdot, \cdot)$ for  $D \in \{A,B, C\}$,  Assumptions A.(i)-A.(ii) hold with appropriate changes, and the noise noise process $\{V_i: i=0, \ldots, n\}$ is correlated with Assumption A.(iii) replaced by the following assumptions. \\
{Assumption  D.(iii)}. The noise process $V^n \tri \{V_i: i=0, \ldots, n\}$ distribution   satisfies conditional independence 
\bea
{\bf P}_{V_i|V^{i-1}, A^i}(dv_i|v^{i-1}, a^i)=  {\bf P}_{V_i|\overline{V}^{i-1}}(dv_i|\overline{v}^{i-1}), \; \overline{v}^{i-1} \in \big\{v_{i-T}^{i-1}, v^{i-1}\big\}-a.a.(v^{i-1}, a^i),\; i=0, \ldots, n \label{CI_M}
\eea
where $T$ is nonnegative and finite. \\
Assumption D.(iv). The inverse of the map 
\bea
h_i^D(\overline{b}^{i-1},\overline{a}^i, \cdot) : {\mathbb V}_i \longmapsto  h_i^D(\overline{b}^{i-1},\overline{a}^i, v_i),\hso  \overline{a}^i \in \{a^i, a_{i-L}^i\}, \hso \overline{b}^{i-1} \in \{b^{i-1}, b_{i-M}^{i-1}\}, \hso i=0, \ldots, n
\eea
 exists and it measurable, i.e., the inverse is $\overline{h}_i^D(b_i, \overline{b}^{i-1}, \overline{a}^i)$, for $i=0, \ldots, n$.
\end{definition}

\ \

Clearly, by (\ref{CIC_A}) the noise process  distribution satisfies  ${\bf P}_{V^n}(dv^n)= \otimes_{i=0}^n {\bf P}_{V_i}(dv_i)$, and the following  consistency condition holds.
\begin{align}
{\mb P}\Big\{B_i \in \Gamma \Big| B^{i-1}=b^{i-1}, A^i=a^i\Big\}=\; & {\bf P}_{V_i}\Big(\big\{V_i: h_i^{A}(b^{i-1}, a_{i-L}^i, V_i) \in \Gamma\big\} \Big), \hso \Gamma \in {\cal B}({\mb B}_i)\label{SCNCMD-A.D} \\
=\; & Q_i(\Gamma |b^{i-1}, a_{i-L}^i),  \hso  i=0,1, \ldots, n.   \label{NCM-A.D_CD}
\end{align}
We use  the convention  that  transmission  starts at time $i=0$,   the  initial data $B^{-1}\tri b^{-1}, A_{-L}^{-1}=a_{-L}^{-1}$ are specified and  known to the encoder and decoder,  and their distribution is fixed. Alternatively, we can assume no information is available for $i \in \{-1, -2, \ldots, \}$, i.e., $\sigma\{B^{-1}, A^{-1}\}= \{\Omega, \emptyset\}$,  which then implies $B_0=h_0^A(A_0, V_0), B_1=h_1^A(B_0, A_1, A_0, V_1), \ldots, B_n=h_n^A(B_{n-1}, \ldots, B_0, A_n, \ldots, A_0, V_n)$.
% For example, these can be taken to   be the null set of data $\emptyset$ or any  available data prior to transmission times $i\in \{-\infty, \ldots, -2,-1\}$. 
% 

Any NCM-A.B induces channel distribution $\{Q_i(db_i|b^{i-1}, a_{i-L}^i): i=0, \ldots, n\big\}$ and any  NCM-B.A induces channel distribution $\{Q_i(db_i|b_{i-M}^{i-1}, a^i): i=0, \ldots, n\big\}$ (i.e., they satisfy  a   consistency conditions as in (\ref{NCM-A.D_CD})).\\
Any NCM-C induces channel distribution
\begin{align}
{\mb P}\Big\{B_i \in \Gamma \Big| B^{i-1}=b^{i-1}, A^i=a^i\Big\}=\; & {\bf P}_{V_i}\Big(\big\{V_i: h_i^{C}(b_{i-M}^{i-1}, a_{i-L}^i, V_i) \in \Gamma \big\}\Big), \hso \Gamma \in {\cal B}({\mb B}_i)\label{SCNCMD-A.D_C} \\
=\; & Q_i(\Gamma |b_{i-M}^{i-1}, a_{i-L}^i),  \hso  i=0,1, \ldots, n.   \label{NCM-A.D_CD_C}
\end{align} 
Since  any NCM-A, NCM-A.B, NCM-B.A., NCM-C, induces a channel distribution of Class A, B, C, then by the converse to the coding theorem the characterization of  FTFI is  the one given in Definition~\ref{def-gsub2sc}, in terms of directed information $I(A^n \rar B^n)$.

However, any NCM-D induces a channel distribution, which depends on past noise symbols. To gain insight into the distribution induced by any  NCM-D with correlated  noise sequence $\{V_i:i=0, \ldots, n\}$  consider the following case.\\
Any NCM-D, with $D=C$ with correlated noise satisfying  Assumptions D.(iii), (iv) induces the following distribution.  
\begin{align}
{\mb P}\Big\{B_i \in \Gamma \Big| B^{i-1}=b^{i-1}, A^i=a^i\Big\}=&{\mb P}\Big\{B_i \in \Gamma \Big| B^{i-1}=b^{i-1}, A^i=a^i, V^{i-1}=v^{i-1}\Big\} \nonumber \\
=& {\bf P}_{V_i| \overline{V}^{i-1}}\Big(\big\{V_i: h_i^{D}(b_{i-M}^{i-1}, a_{i-L}^i, V_i) \in \Gamma\big\}  \Big| \overline{v}^{i-1}\Big), \hso \Gamma \in {\cal B}({\mb B}_i)\label{SCNCMD-A.D_C_CD} \\
=\; & Q_i(\Gamma |b_{i-M}^{i-1}, a_{i-L}^i, \overline{v}^{i-1}), \hso \overline{v}^{i-1} \in \big\{v_{i-T}^{i-1}, v^{i-1}\big\}, \hso  i=0,1, \ldots, n.   \label{NCM-A.D_CD_C_CD}
\end{align}
Clearly, in general, if the  noise  is correlated then  the channel (\ref{NCM-A.D_CD_C_CD}) depends on past noise sequences, and hence we need to identify the  characterization of  FTFI from the converse to the coding theorem.   \\
%However, since  the  inverse of the map $
%v_i \in {\mathbb V}_i \longmapsto h_i^D({b}_{i-M}^{i-1},{a}_{i-L}^i, v_i), \hso i=0, \ldots, n$ 
% exists and it measurable, i.e., the inverse is $\overline{h}_i^D(b_i, {b}_{i-M}^{i-1}, {a}_{i-L}^i)$, for $i=0, \ldots, n$), then in (\ref{NCM-A.D_CD_C_CD}) we can express the conditioning in terms of the inverse function.  For example, if the noise is Markov and $B_i= h_i^D({B}_{i-1},A_i, {A}_{i-1}, v_i), \hso i=0, \ldots, n$, i.e., 
%$T=1, M=L=1$, then $v_i=\overline{h}_i^D(b_i, {b}_{i-1},a_i, {a}_{i-1})$, and (\ref{NCM-A.D_CD_C_CD}) reduces to 
%\begin{align}
%Q_i(db_i|b_{i-1}, a_i, a_{i-1},v_{i-1} ) = Q_i(db_i|b_{i-1}, a_i, a_{i-1},\overline{h}_{i-1}^D(b_{i-1}, {b}_{i-2},a_{i-1}, {a}_{i-2}) ),  \hso  i=0, \ldots, n. \label{Gen_CH}
%\end{align}
%The last identity implies the channel is transformed to another channel described by (\ref{Gen_CH}). 
%In this case, we show that the characterization of FTFI capacity is similar to  the one given in Definition~\ref{def-gsub2sc}. \\
For general NCM-D, in the next section, we identify via the converse coding theorem, the characterization of FTFI capacity and the information structures of optimal channel input distributions.

%In view of these observations, we show via running  application examples  
%
%
%
%
%
%
%
%
%In Theorem~\ref{thm_NCM}, we show under mild conditions, that any NCM driven by correlated noise process $\{V_i: i=0, \ldots, n\}$, can be       equivalently transformed, to one of the NCMs, NCM-A, NCM-A.B, NCM-B.A, or NCM-C,  which then implies channel distributions of Class A, B, C are sufficient to deal with NCMs,  driven by noise processes, which are not necessarily independent.    Indeed, an application of Theorem~\ref{thm_NCM} illustrates that the \cite{cover-pombra1989} nonnstationary nonergodic Additive Gaussian channels driven by correlated noise defined (\ref{c-p1989}) and  (\ref{CR_1}) can be equivalently transformed into  a degenerate  version of NCM-B.A,  while if the noise is finite memory then it becomes a degenerate version of  NCM-C. 

\subsection{Arbitrary Distributed NCM-D: Converse Coding Theorem \& Information Structures}
In this section we illustrate the application of information structures of optimal channel input distributions to NCM-D, with arbitrary alphabets, and correlated noise process.  Subsequently,  we apply the results to generalizations of  the the Cover and Pombra non-stationary non-ergodic Additive Gaussian Noise (AGN) Channel  \cite{cover-pombra1989} defined by (\ref{c-p1989}) and (\ref{CR_1}), including the case when the noise is finite memory.  

{\bf Nonlinear Channel Model-D.} 
Consider the specific NCM-D (see Definition~\ref{R-exa_A_D}) defined as follows.
\begin{align}
& B_i   = h_i^{D}(B_{i-M}^{i-1}, A_{i-L}^i, V_i), \hso B_{-M}^{-1}=b_{-M}^{-1},\; A_{-L}^{-1}=a_{-L}^{-1}, \hso i=0, \ldots, n,  \label{NCM-D_IS_1} \\
 &{\cal P}_{[0,n]}^{C}(\kappa) \tri 
  \Big\{   P_i(da_i|a^{i-1}, b^{i-1}), i=0, 1, \ldots, n:\frac{1}{n+1}\sum_{i=0}^n {\bf E} \Big( \gamma_i^{C.L,M}(A_{i-L}^i, B_{i-M}^{i}) \Big)  \leq \kappa\Big\}, \label{NCM-D_IS_2} \\  %hst I \tri \max\{L, N\}, \: J\tri \max\{M, L, N, K\}, 
& {\bf P}_{V_i|V^{i-1}, A^i}(dv_i|v^{i-1}, a^i)=  {\bf P}_{V_i|V_{i-T}^{i-1}}(dv_i|v_{i-T}^{i-1})-a.a.(v^{i-1}, a^i),\; i=0, \ldots, n, \label{NCM-D_IS_3}\\
&  \big\{{\mb B^i},  {\mb A}^i,  {\mb V}^i\big\}, \hso i=0, \ldots, n \hso \mbox{are arbitrary}, \\
&\mbox{the inverse of the map $ 
v_i \in {\mathbb V}_i \longmapsto h_i^D(b_{i-M}^{i-1},a_{i-L}^i, v_i), \hso i=0, \ldots, n$ exists and it measurable} \label{NCM-D_IS_4}
\end{align}
i.e., the inverse is $\overline{h}_i^D(b_i, {b}_{i-M}^{i-1}, a_{i-L}^i)$, for $i=0, \ldots, n,$, where $\{V_i: i=\ldots, -1, 0,1, \ldots, n\}$ is the noise process, 
%and for each $i$, $V^i \tri \{\ldots, V_{-1},  V_0, \ldots, V_i\}, A^i\tri \{\ldots, A_{-1}, A_0, \ldots, A_i\}, B^i\tri \{\ldots, B_{-1}, B_0, \ldots, B_i\}$, for $i=0,  \ldots, n,$ and  $P_0(da_0|a^{-1}, b^{-1})=P_0(da_0)$, ${\bf P}_{V_0|V^{-1}, A^0}={\bf P}_{V_0}$, i.e.,
 the past   information available  to the encoder and decoder at time $i=0$ is the null set, or the initial data $(b^{-1}, a^{-1})$ are known to the encoder and decoder (an alternative convention can be used  as in \cite{kim2008,kim2010}).

First, we prove  the converse to the coding theorem, as discussed in Section~\ref{NRCM_G}. Then we show that the information structures of maximizing distributions, follow directly either from Section~\ref{class_A} or Section~\ref{class_C}.   \\
The channel distribution is obtained as follows. 
\begin{align}
{\mb P}\Big\{B_i \leq b_i  \Big| B^{i-1}=b^{i-1}, A^i=a^i\Big\}=& {\mb P}\Big\{h_i^{D}(B_{i-M}^{i-1}, A_{i-L}^i, V_i) \leq b_i  \Big| B^{i-1}=b^{i-1}, A^i=a^i, V^{i-1}=v^{i-1}\Big\} \label{CH_D1}  \\  
=&{\bf P}_{V_i|B^{i-1}, A^i, V^{i-1}}\Big( \big\{V_i: V_i\leq \overline{h}_i^{D}(b_i,b_{i-M}^{i-1}, a_{i-L}^i) \big\}\Big) \label{A-NCD-3_n} \\
=& {\bf P}_{V_i|A^i, V^{i-1}}\Big(\big\{V_i:  \overline{h}_i^{D}(b_i,b_{i-M}^{i-1}, a_{i-L}^i)   \leq b_i\big\} \big\}\Big) \label{CH_D2}\\
=& {\bf P}_{V_i|V_{i-T}^{i-1}}\Big(\big\{V_i:  \overline{h}_i^{D}(b_i,b_{i-M}^{i-1}, a_{i-L}^i)   \leq b_i \big\}\Big) \hso \mbox{by (\ref{NCM-D_IS_3})} \label{A-NCD-4_n} \\
\equiv & Q_i\big((-\infty, b_i]  \Big| b_{i-M}^{i-1},a_{i-L}^i, v_{i-T}^{i-1}\big), \hso  i=0,\ldots, n \label{A-NCD-4_nn}
\end{align}
where the first and second identities (\ref{CH_D1}), (\ref{A-NCD-3_n})  follow from (\ref{NCM-D_IS_4}), the third identity (\ref{CH_D2}) follows from (\ref{NCM-D_IS_1}), the fourth identity (\ref{A-NCD-4_n}) is due to (\ref{NCM-D_IS_3}), and the last identity states that the channel distribution $\Big\{{Q}_i\big(db_i| a_{i-L}^{i}, b_{i-M}^{i-1},v_{i-T}^{i-1}\big): i=0, \ldots, n\Big\}$ is induced by the noise distribution and the channel.\\
Next, we use the assumption that the initial data are known to the encoder and decoder. Let ${\bf P}_{A_i|A^{i-1}, B^{i-1}, V^{i-1}}\equiv \overline{P}_i(da_i|a^{i-1}, b^{i-1}, v^{i-1}): i=0, \ldots,n\}$. Then by (\ref{NCM-D_IS_4}) we have $\{P_i(da_i |a^{i-1}, b^{i-1})=\overline{P}_i(da_i|a^{i-1}, v^{i-1}, b^{i-1}): i=0, \ldots, n\}$. In fact, if the initial data are known to the encoder then by knowing $\{a^{i-1}, b^{i-1}\}$ the encoder also knows $v^{i-1}$ for $i=0, \ldots, n$.  \\
Moreover, the induced joint  distributions is  given as follows\footnote{$\delta_{x}(dy)$ is the delta measure concentrated at $y=x$.}.
\begin{align}
{\bf P}^{{P}}&(dv^{i-1}, da^i,   d b^{i})=\otimes_{j=0}^i \Big({\bf P}(db_j| b^{j-1}, a^j, v^{j-1}) \otimes {\bf P}(da_j|a^{j-1},v^{j-1}, b^{j-1})\otimes \delta_{\overline{h}_{j-1}^D(b_{j-1-M}^{j-1}, a_{j-1-L}^{j-1})}(dv_{j-1})\Big) \nonumber  \\  
=& \otimes_{j=0}^i \Big(Q_j(db_j|b_{j-M}^{j-1},a_{j-L}^{j}, v_{j-T}^{j-1}) \otimes \overline{P}_j(da_j|a^{j-1},v^{j-1}, b^{j-1})\otimes \delta_{\overline{h}_{j-1}^D(b_{j-1-M}^{j-1},a_{j-1-L}^{j-1})}(dv_{j-1})\Big)  \label{CM-TC-C_4_cc-B_n} \\
\equiv & {\bf P}^{\overline{P}}(dv^{i-1}, da^i,   d b^{i}), \hso i=1, \ldots, n.
\end{align}
The channel output transition  distribution is thus, given by
\begin{align}
  \Pi_i^{\overline{P}}(db_i | b^{i-1}) =& \int_{  {\mb A}^i\times {\mb V}^{i-1} }   Q_i(db_i |  b_{i-M}^{i-1},a_{i-L}^i, v_{i-T}^{i-1}) \otimes  \overline{P}_j(da_i |a^{i-1}, v^{i-1}, b^{i-1}) \nonumber \\
  &\otimes {\bf P}^{\overline{P}}(da^{i-1}, dv^{i-1} | b^{i-1})  , \hso i=0, \ldots, n. \label{A-CM-TC-C_3_CC_B_n}
  \end{align}
  If the initial data is the null set then we can set 
  \begin{align}
  \Pi_0^{\overline{P}}(db_0 | b^{-1}) = \Pi_0^{P}(db_0)= \int_{  {\mb A}_{0} }   Q_0(db_0 | a_0) \otimes  P_0(da_0),\hso {\bf P}^{\overline P}(da_0,  d b_{0})=Q_0(db_0|a_0) \otimes P_0(da_0).
   \end{align}
Next, we  derive a converse coding theorem, which  shows that the supremum of all achievable codes is bounded above by the supremum over all channel input distribution,  of the conditional mutual information $\sum_{i=0}^n I(A_{i-L}^i, V_{i-T}^{i-1}; B_i|B^{i-1})$, and then we identify the information structures of optimal channel input distributions, so that this upper bound  is tight. 

\ \

\begin{theorem}(Converse coding theorem \& information structures for arbitrary NCM-D)\\
\label{CC_CP_A}
Consider the NCM-D defined by (\ref{NCM-D_IS_1})-(\ref{NCM-D_IS_4}). Then we have the following.\\
{\it (1)  Converse coding theorem.} If there exists a sequence of feedback codes $\big\{(n, { M}_n, \epsilon_n): n=0, 1, \ldots, \}$ as defined in Section~\ref{introduction}, (a), (b) such that  such that $\lim_{n\longrightarrow\infty} {\epsilon}_n=0$ (i.e., the probability of decoding error goes to zero) then\footnote{The superscript notation $I^g(\cdot; \cdot|\cdot)$ indicates that the distributions depend on encoding strategies.} 
\begin{align} 
 R \leq & \liminf_{n \longrightarrow\infty}\frac{1}{n+1}\log{{M}_n} \label{rate_A-LCM_n} \\
{\leq} \: & \liminf_{n \longrightarrow\infty}  \sup_{  \overline{\cal E}_{[0,n]}^{FB}(\kappa)}  \frac{1}{n+1}  \sum_{i=0}^n I^{\overline{g}}(A_{i-L}^i, V_{i-T}^{i-1}; B_i| B^{i-1}), \hso a_i=\overline{g}_i(w, a^{i-1}, v^{i-1},  b^{i-1})
 \label{CCIS_6a_BC_n}\\
{\leq} \: & \liminf_{n \longrightarrow \infty} \sup_{  \overline{\cal P}_{[0,n]}^{D}(\kappa)}  \frac{1}{n+1}\sum_{i=0}^n I(A_{i-L}^i, V_{i-T}^{i-1}; B_i|B^{i-1}) \label{CCIS_8_BC_n} \\
\equiv & \liminf_{n \longrightarrow \infty}\frac{1}{n+1} {C}_{W ; B^n}^{FB,D}(\kappa) 
\end{align}
where 
\begin{align}
&{C}_{W ; B^n}^{FB,D}(\kappa)\tri  \sup_{\overline{\cal P}_{[0,n]}^{C}(\kappa)}\sum_{i=0}^n  I(A_{i-L}^i, V_{i-T}^{i-1}; B_i|B^{i-1}), \label{CCIS_8_BC_D_n}\\
&I(A_{i-L}^i, V_{i-T}^{i-1}; B_i|B^{i-1})= \int_{{\mb B}^i\times {\mb A}^i \times {\mb V}^{i-1} } 
\log\Big(\frac{dQ_i(\cdot|b_{i-M}^{i-1}, a_{i-L}^{i},v_{i-T}^{i-1})}{d\Pi_i^{\overline P}(\cdot|b^{i-1})}(b_i)\Big) Q_i(db_i|b_{i-M}^{i-1},a_{i-L}^i, v_{i-T}^{i-1}), \nonumber \\
& \hst \hst \hst  \otimes \overline{P}_i(da_i |a^{i-1}, v^{i-1},  b^{i-1})\otimes {\bf P}^{\overline{P}}(da^{i-1},dv^{i-1}|b^{i-1})\otimes {\bf P}^{\overline{P}}(db^{i-1}), \hso i=0, \ldots, n,
\\
&\overline{\cal E}_{[0,n]}^{FB}(\kappa) \tri\Big\{ \overline{g}_i(w, a^{i-1}, v^{i-1},  b^{i-1}), i=0,\ldots,n:  \frac{1}{n+1} {\bf E}^{\overline g}\Big( \sum_{i=0}^n \gamma_i^{C.L,M}(A_{i-L}^i, B_{i-M}^{i}) \leq \kappa  \Big)  \Big\},  \\
&\overline{\cal P}_{[0,n]}^{D}(\kappa) \tri \Big\{ \overline{P}_i(da_i| a^{i-1}, v^{i-1},  b^{i-1}), i=0,\ldots,n:  \frac{1}{n+1} {\bf E}^{\overline P}\Big( \sum_{i=0}^n \gamma_i^{C.L, M}(A_{i-L}^i, B_{i-M}^i) \leq \kappa  \Big)  \Big\} 
\end{align} 
provided the following conditions hold.\\
(a) The supremum of $\sum_{i=0}^n I(A_{i-L}^i, V_{i-T}^{i-1}; B_i|B^{i-1})$  over $\overline{\cal P}_{[0,n]}^{C}(\kappa)$ in       (\ref{CCIS_8_BC_n}) for any finite $n$  in achieved in the set  (i.e., the maximizing distribution exists).\\
(b) The $liminf_{n \longrightarrow \infty}$ in (\ref{CCIS_8_BC_n}) is finite.\\
{(2) Information structures.}  The optimal distribution  in   (\ref{CCIS_8_BC_D_n}) satisfies  conditional independence 
\begin{align}
\overline{P}_i(da_i| a^{i-1}, v^{i-1},  b^{i-1})=\pi_i^{D.L, T}(da_i| a_{i-L}^{i-1}, v_{i-T}^{i-1},  b^{i-1})-a.a.(a^{i-1}, v^{i-1}, b^{i-1}), \hso i=0,\ldots,n \label{IS_D_1}
\end{align}
and the corresponding charactrization of FTFI capacity is given by 
\begin{align}
{C}_{W ; B^n}^{FB,D.L,T}(\kappa)\tri & \sup_{\overline{\cal P}_{[0,n]}^{D.L,T}(\kappa)}\sum_{i=0}^n {\bf E}^{\pi^{D.L,T}}\Big\{ \log\Big(\frac{dQ_i(\cdot|B_{i-M}^{i-1}, A_{i-L}^{i},V_{i-T}^{i-1})}{d\Pi_i^{\pi^{A.L,T}}(\cdot|B^{i-1})}(B_i)\Big)\Big\} \label{CCIS_8_BC_D_nn}
\end{align}
where 
\begin{align}
\overline{\cal P}_{[0,n]}^{D.L,T}(\kappa) \tri& \Big\{ \pi_i^{D.L,T}(da_i| a_{i-L}^{i-1}, v_{i-T}^{i-1},  b^{i-1}), i=0,\ldots,n:  \frac{1}{n+1} {\bf E}^{\pi^{D.L,T}}\Big( \sum_{i=0}^n \gamma_i^{C.L, M}(A_{i-L}^i, B_{i-M}^i) \leq \kappa  \Big)  \Big\} 
\end{align}
and the joint and transition probability distributions are given as follows.
\begin{align}
&{\bf P}^{\pi^{D.L,T}}(dv^{i-1}, da^i,   d b^{i})=
 \otimes_{j=0}^i \Big(Q_j(db_j|b_{j-M}^{j-1},a_{j-L}^{j}, v_{j-T}^{j-1}) \otimes \pi^{A.L,T}_j(da_j|a_{j-L}^{j-1},v_{j-T}^{j-1}, b^{j-1}) \nonumber \\
& \hst \hst \hst \hst  \otimes {\bf P}^{\pi^{D.L,T}}(da_{j-L}^{j-1}, dv_{j-T}^{j-1}|b^{j-1}) \otimes \delta_{\overline{h}_{j-1}^D(b_{j-1-M}^{j-1},a_{j-1-L}^{j-1})}(dv_{j-1})\Big), \hso i=0, \ldots, n, \label{CM-TC-C_4_cc-B_nn}  \\
&  \Pi_i^{\pi^{D.L,T}}(db_i | b^{i-1}) = \int_{  {\mb A}_{i-L}^i\times {\mb V}_{i-T}^{i-1} }   Q_i(db_i |  b_{i-M}^{i-1},a_{i-L}^i, v_{i-T}^{i-1}) \otimes  \pi^{D.L,T}_i(da_i |a_{i-L}^{i-1}, v_{i-T}^{i-1}, b^{i-1}) \nonumber \\
  & \hst \hst \hst \hst \otimes {\bf P}^{\pi^{D.L,T}}(da_{i-L}^{i-1}, dv_{i-T}^{i-1} | b^{i-1}) \label{A-CM-TC-C_3_CC_B_nn}
%  \\
%&  \Pi_0^{\pi^{A.L,T}}(db_0 | b^{-1}) = \Pi_0^{P}(db_0)= \int_{  {\mb A}_{0} }   Q_0(db_0 | a_0) \otimes  P_0(da_0). \label{IS_D_2}
   \end{align} 
 where the \'a posteriori distribution $\big\{{\bf P}^{\pi^{D.L,T}}(da_{i-L}^{i-1}, dv_{i-T}^{i-1} | b^{i-1}): i=0, \ldots, n\big\}$ satisfies a recursion.    
 \end{theorem}
\begin{proof}
(1) Consider any sequence of feedback codes as defined in Section~\ref{introduction}, (a), (b).  Suppose $R$ is achievable, so  there exists an  $(n, { M}_n, \epsilon_n)$ block code ${\cal C}_n=(u_0, u_1,u_2,\dots,u_{{M}_n})$ such that $\lim_{n\longrightarrow\infty} {\epsilon}_n=0$ and $\liminf_{n \longrightarrow\infty}\frac{1}{n+1}\log{{M}_n}\geq R$. Then for each $n$, since $W \in {\cal M}_n$ is uniformly distributed and  in view of  Fano's inequality   \cite{cover-thomas2006},  the following inequalities hold\footnote{The superscript $g$ on $H^g(\cdot), I^g(\cdot; \cdot), {\bf E}^g\{\cdot\}$ indicates that the distributions are induced by the channel and $\{g_i(\cdot, \cdot): i=0, \ldots, n\} \in {\cal E}_{[0,n]}^{FB}(\kappa)$.}.
\begin{align}
\log  M_n =& H(W) =H^g(W|B^n)+ I^g(W;B^n), \hso  \forall  \{g_i(\cdot, \cdot): i=0, \ldots, n\} \in {\cal E}_{[0,n]}^{FB}(\kappa) \nonumber \\
\leq & h(\varepsilon_n) + \varepsilon_n \log M_n + I^g(W; B^n), \hso h(z) \tri -z\log z -(1-z)\log (1-z), \hso z \in [0,1]  \nonumber \\
= & h(\varepsilon_n) + \varepsilon_n \log M_n  + \sum_{i=0}^n I^g(W; B_i|B^{i-1})
\end{align}
where 
\begin{align}
I^g(W; B^n) =& \sum_{i=0}^n I^g(W; B_i|B^{i-1}),  \hso a_i=g_i(w, a^{i-1}, b^{i-1}), \; i=0, \ldots, n  \label{EQ_1} \\
=&  \sum_{i=0}^n {\bf E}^g \Big\{ 
\log\Big(\frac{d{\bf P}^g(\cdot|B^{i-1}, W)}{{\bf P}^{g}(\cdot|B^{i-1})}(B_i)\Big) \Big\} \label{EQ_2} \\
=& \sum_{i=0}^{n} {\bf E}^g \Big\{ 
\log\Big(\frac{d{\bf P}^g(\cdot|B^{i-1}, \{g_j(W, A^{j-1}, B^{j-1}): j=0, \ldots, i\}, W)}{{\bf P}^{g}(\cdot|B^{i-1})}(B_i)\Big) \Big\} \label{EQ_3} \\
=& \sum_{i=0}^{n} {\bf E}^g \Big\{ 
\log\Big(\frac{d{\bf P}^g(\cdot|B^{i-1}, \{g_j(W, A^{j-1}, B^{j-1}): j=0, \ldots, i\}, W, V^{i-1})}{{\bf P}^{g}(\cdot|B^{i-1})}(B_i)\Big) \Big\} \label{EQ_4} \\
=& \sum_{i=0}^n  {\bf E}^{\overline{g}} \Big\{ 
\log\Big(\frac{d{\bf P}^{\overline{g}}(\cdot|B^{i-1}, \{\overline{g}_j(W,A^{j-1}, V^{j-1}, B^{j-1}): j=0, \ldots, i\}, W, V^{i-1})}{{\bf P}^{\overline{g}}(\cdot|B^{i-1})}(B_i)\Big) \Big\} \label{EQ_5} \\
\leq &    \sup_{ \overline{\cal E}_{[0,n]}^{FB}(\kappa)}\sum_{i=0}^n {\bf E}^{\overline{g}} \Big\{ 
\log\Big(\frac{d{\bf P}^{\overline{g}}(\cdot|B_{i-M}^{i-1},\{ \overline{g}_j(W, A^{j-1}, V^{j-1}, B^{j-1}):j=i-L, \ldots, i\}, V_{i-T}^{i-1})}{{\bf P}^{\overline{g}}(\cdot|B^{i-1})}(B_i)\Big) \Big\} \label{EQ_6} \\
\leq &  \sup_{ \overline{\cal P}_{[0,n]}^{D}(\kappa) }\sum_{i=0}^n  {\bf E}^{\overline P} \Big\{ 
\log\Big(\frac{dQ_i(\cdot|B_{i-M}^{i-1},A_{i-L}^i, V_{i-T}^{i-1})}{\Pi_i^{\overline P}(\cdot|B^{i-1})}(B_i)\Big) \Big\} \label{EQ_7} \\
\equiv & \sup_{ \overline{\cal P}_{[0,n]}^{D}(\kappa) }\sum_{i=0}^n I(A_{i-L}^i, V_{i-T}^{i-1}; B_i|B^{i-1}) \label{CCT-EQ2_a}
\end{align}
where (\ref{EQ_1}) is due to chain rule of mutual information,    (\ref{EQ_2}) is by definition, (\ref{EQ_3}) is also by definition, i.e., the distributions are evaluated for a fixed encoding strategy, (\ref{EQ_4}) and (\ref{EQ_5}) are  due to the invertibility condition (\ref{NCM-D_IS_4}), 
(\ref{EQ_6}) follows from the channel definition and by taking the supremum, (\ref{EQ_7}) is due to the fact that $\overline{\cal E}_{[0,n]}^{FB}(\kappa) \subseteq    \overline{\cal P}_{[0,n]}^{C}(\kappa)$, since  $\{\overline{P}_i(da_i|a^{i-1}, v^{i-1}, b^{i-1}): i=0, \ldots, n\}$ are not   necessarily generated by uniform RVs $W$. 
%the data processing inequality, due  to  (\ref{CI_Massey}), i.e., (\ref{G-LCM_2_AG_n}). 
 From the above inequalities we extract the following inequalities.  
\begin{align}
\log M_n \leq  & h(\varepsilon_n) + \varepsilon_n \log M_n + \sup_{ \overline{\cal E}_{[0,n]}^{FB}(\kappa)  }\sum_{i=0}^n I^{\overline{g}}(A_{i-L}^{i}, V_{i-T}^{i-1}; B_i|B^{i-1}), \hso a_i=\overline{g}_i(w, v^{i-1}, b^{i-1}), \; i=0, \ldots, n  \nonumber   \\
\leq  & h(\varepsilon_n) + \varepsilon_n \log M_n +        \sup_{ \overline{\cal P}_{[0,n]}^{D}(\kappa)  }\sum_{i=0}^n I(A_{i-L}^i, V_{i-T}^{i-1}; B_i|B^{i-1})  \label{CCT-EQ3_n}
\end{align} 
By conditions (a), (b), there exists a channel input distribution, which achieves the  supremum in (\ref{CCT-EQ3_n}) and its per unit time limit exists and it is finite, hence    
by dividing both sides of the above inequalities  by $(n+1)$ and  taking the limit, as $n \longrightarrow \infty$, then    $\varepsilon_n \longrightarrow 0$ and  $h(\varepsilon_n) \longrightarrow 0$, and moreover the inequalities 
(\ref{rate_A-LCM_n})-
(\ref{CCIS_8_BC_n}) are obtained.  \\
(2) Next, we show the statements regarding the information structure of the optimal distribution  
(\ref{IS_D_1})-(\ref{A-CM-TC-C_3_CC_B_nn}). These are easily  obtained from   
Section~ \ref{Class_A_SOC} or Section~\ref{class_C} with $\{(a_{i-L}^{i-1}, v_{i-T}^{i-1}): i=0, \ldots, n\}$ replacing $\{a_{i-L}^{i-1}: i=0, \ldots, n\}$. This completes the prove.
\end{proof}

Note that for each $i$, the noise sequence $v^{i-1}$ is available to the encoder, but not at the decoder, and hence,  it is a state variable that needs to estimated at the decoder for $i=0, \ldots, n$. This is analogous to the discussion in  Remark~\ref{SC_REM}, (b). 

In general, the noise distribution may be described via another recursive nonlinear dynamical model. We discuss an example, below.

\ \

\begin{example}(Nonlinear and linear noise models)\\
Consider the NCM-D defined by (\ref{NCM-D_IS_1})-(\ref{NCM-D_IS_4}), with $T=1$. Then we can model the noise process as follows.\\
(a) Nonlinear Noise Model. A Nonlinear Noise Model (NNM) is described by the recursion
\begin{align}
V_i=f_i(V_{i-1}, W_i), \hso V_{-1}=v_{-1}, \hso i=0, \ldots, n
\end{align}  
where ${\mathbb V}_i={\mb R}^p, {\mb W}_i={\mathbb R}^r$, $W_i\in {\mathbb R}^r$, and the noise process $\{W_i: i=0, \ldots, n\}$ is independent and identically distributed, independent of $V_{-1}$.  Then ${\bf P}_{V_i|V^{i-1}}={\bf P}_{V_i|V_{i-1}}, i=0, \ldots, n$, i.e., the noise process is Markov.\\
(b) Linear Noise Model. A Linear Noise Model (LNM) is a special case of NNM, described by 
\begin{align}
V_i=A_i V_{i-1}+ B_i W_i, \hso V_{-1}=v_{-1}, \hso i=0, \ldots, n
\end{align} 
(c) Gaussian Linear Noise Model. A Gaussian Linear Noise Model (G-LNM) is a special case of  LNM with  Gaussian distributed  noise, i.e.,  $\{W_i\sim N(0, \Sigma_{W_i}): i=0, \ldots, n\}$. This implies, that  $\{V_i:i=0, \ldots, n\}$ for fixed $V_{-1}=v_{-1}$ is also Gaussian. 
\end{example}

Next, we discuss generalizations of Theorem~\ref{CC_CP_A} to different NCM-D.\\

\begin{remark}(Alternative NCM-D)\\
\label{NCM_CP}
(a) If  the noise distribution of the NCM-D defined by  (\ref{NCM-D_IS_1})-(\ref{NCM-D_IS_4}) is replaced by 
\begin{align}
{\bf P}_{V_i|V^{i-1}, A^i}(dv_i|v^{i-1}, a^i)=  {\bf P}_{V_i|V^{i-1}}(dv_i|v^{i-1})-a.a.(v^{i-1}, a^i),\; i=0, \ldots, n \label{NCM-D_IS_3_nn}
\end{align}
then the analog of  Theorem~\ref{CC_CP_A} is obtained by substituting $v_{i-T}^{i-1} \longmapsto v^{i-1}, i=0, \ldots,n$ in all equations, i.e., the optimal distribution satisfies conditional independence condition
\begin{align}
\overline{P}_i(da_i| a^{i-1}, v^{i-1},  b^{i-1})=\pi_i^{A.L}(da_i| a_{i-L}^{i-1}, v^{i-1},  b^{i-1})-a.a.(a^{i-1}, v^{i-1}, b^{i-1}), \hso i=0,\ldots,n \label{IS_D_1_1}
\end{align}
(b) If the channel of the NCM-D defined by  (\ref{NCM-D_IS_1})-(\ref{NCM-D_IS_4}) is replaced by  
\begin{align}
B_i= h_i^D(B_{i-M}^{i-1}, A^i, V_i), \hso i=0, \ldots, n
\end{align}
then the analog of  Theorem~\ref{CC_CP_A} is  obtained by substituting $a_{i-L}^{i-1} \longmapsto a^{i-1}, i=0, \ldots,n$ in all equations.
\end{remark}

In view of the main theorems obtained thus far, next we relate our characterizations of FTFI capacity to existing results found in the literature, specifically, \cite{cover-pombra1989,kim2010}.

\subsubsection{Arbitrary Distributed  Additive Channel Noise Models}
\label{LCM_G}
Consider a model  called Arbitrary Distributed Additive Channel Noise  (ACN) model   with transmission cost constraint, defined as follows. 
 \begin{align}
&B_i  =A_i + V_{i},   \hso i= 0, \ldots,n, \label{G-LCM_1_AG} \\
& {\bf P}_{V_i|V^{i-1}, A^i}={\bf P}_{V_i|V^{i-1}}, \hso i=0, \ldots, n,   \label{G-LCM_2_AG}  \\
 & {\cal P}_{[0,n]}^{C.0,0}(\kappa) \tri 
  \Big\{   P_i(da_i|a^{i-1}, b^{i-1}), i=0, 1, \ldots, n: \frac{1}{n+1} {\bf E}^{P}\Big( \sum_{i=0}^n \gamma_i^{C.0, 0}(A_i, B_{i}) \leq \kappa  \Big)\Big\}, \label{G-LCM_4_AG}  \\
&  \big\{{\mb B_i},  {\mb A}_i,  {\mb V}_i\big\}, \hso i=0, \ldots, n \hso \mbox{are arbitrary} \label{G-LCM_6_AG} 
\end{align}
where for each $i$, $V^i \tri \{V_0, \ldots, V_i\}, A^i\tri \{A_0, \ldots, A_i\}$, for $i=0,  \ldots, n,$ and  $P_0(da_0|a^{-1}, b^{-1})=P_0(da_0)$, ${\bf P}_{V_0|V^{-1}, A^0}={\bf P}_{V_0}$, i.e., no information is available at time $i=0$ to the encoder and decoder,  the alphabet spaces are arbitrary, and the noise process $\{V_i: i=0, \ldots, n\}$ is arbitrary distributed. The above model is a generalization of the non-stationary non-ergodic AGN channel investigated by Cover and Pombra \cite{cover-pombra1989}.
%Note also that for each $i$, and fixed $b^{i-1}$ then $a_^{i-1}=b^{i-1}-v_0^{i-1}$  is uniquely defined from knowledge of $v_0^{i-1}$ and vice-versa. Hence, by properties of conditional mutual information the following identity holds. 
%\begin{align}
%I(A_i, V_0^{i-1}; B_i|B^{i-1})= I(A_i, A^{i-1}; B_i|B^{i-1}), \hso i=0, \ldots, n.
%\end{align}

The characterization of the FTFI capacity is obtained from Theorem~\ref{CC_CP_A}, as discussed in Remark~\ref{NCM_CP}, (a). 
Next, we express the characterization of FTFI capacity  in terms of random processes $\{A_i, B_i): i=0, \ldots, n\}$.  To simplify the presentation,   we assume all distributions are absolutely continuous with respect to the Lebesgue measures, i.e.,  ${\bf P}_{X|Z}(dx|z)= {\bf p}_{X|Z}(x|z)dx$, thus, lower case functions denote  probability density functions.\\
By Remark~\ref{NCM_CP} and  (\ref{A-NCD-4_n}),  (\ref{A-CM-TC-C_3_CC_B_n}), and since $a_i=b_i-v_i, i=0, \ldots, n$,  we have the following. 
\begin{align}
&{\bf p}_{B_i|B^{i-1},A^i}(b_i|a^i, b^{i-1})={\bf p}_{V_i|V^{i-1}}(b_i-a_i|v^{i-1}), \hso i=0, \ldots, n,\\
&{\bf p}_{B_i|B^{i-1}}(b_i | b^{i-1}) = \int_{  {\mb A}_{i}\times {\mathbb V}^{i-1} } {\bf p}_{V_i|V^{i-1}}(b_i-a_i|v^{i-1}) {\bf p}_{A_i|V^{i-1}, B^{i-1}}(a_i |v^{i-1}, b^{i-1}) {\bf p}_{V^{i-1}|B^{i-1}}(v^{i-1} | b^{i-1}) da_i dv^{i-1}.
\end{align}
By change of variables of integration we obtain the following.
\begin{align}
{\bf p}_{B_i|B^{i-1}}(b_i | b^{i-1})=&\int_{{\mathbb A}^{i} } {\bf p}_{V_i|V^{i-1}}(b_i-a_i|b^{i-1}-a^{i-1}) {\bf p}_{A_i|A^{i-1}, B^{i-1}}(a_i |a^{i-1}, b^{i-1}) {\bf p}_{V^{i-1}|B^{i-1}}(b^{i-1}-a^{i-1} | b^{i-1}) da^i.
\end{align}
Therefore,  the characterization of FTFI capacity, i.e., the analog of  (\ref{CCIS_8_BC_D_nn}) is  the following.
\begin{align}
&{C}_{W ; B^n}^{FB,D.0}(\kappa) = \sup_{\big\{{\bf p}_{A_i| A^{i-1},  B^{i-1}}, i=0,\ldots,n:  \frac{1}{n+1} {\bf E}\Big( \sum_{i=0}^n \gamma_i^{C.0, 0}(A_i, B_{i}) \leq \kappa  \Big)  \big\}  } \sum_{i=0}^n \int_{ {\mb B}^i \times {\mb A}^{i} } 
\log\Big( \frac{  {\bf p}_{V_i|V^{i-1}}(b_i-a_i|b^{i-1}-a^{i-1})}{ {\bf p}_{B_i|B^{i-1}}(b_i|b^{i-1})}\Big) \nonumber \\
& {\bf p}_{V_i|V^{i-1}}(b_i-a_i|b^{i-1}-a^{i-1})    
 {\bf p}_{A_i|A^{i-1}, B^{i-1}}(a_i| a^{i-1},  b^{i-1}) {\bf p}_{V^{i-1}|B^{i-1}}(b^{i-1}-a^{i-1}|b^{i-1}) {\bf p}_{B^{i-1}}(b^{i-1})da^idb^i. \label{NEW_1}
\end{align} 
Note that by the  additive noise channel  property, the distribution   $\{P_i(da_i |v^{i-1}, b^{i-1}): i=0, \ldots, n\}$ uniquely defines $\{P_i(da_i |a^{i-1}, b^{i-1}): i=0, \ldots, n\}$ and vice-versa, and this also holds for general recursive NCMs, under mild conditions, i.e., the invertibility condition  (\ref{NCM-D_IS_4}). 
% then  we obtain   the following  characterization of FTFI capacity.  
% \begin{align}
%{C}_{W ; B^n}^{FB,LCM}(\kappa) \tri & \sup_{\big\{P_i(da_i| v^{i-1},  b^{i-1}), i=0,\ldots,n:  \frac{1}{n+1} {\bf E}^{P}\Big( \sum_{i=0}^n \gamma_i^{C.0, 0}(A_i, B_{i}) \leq \kappa  \Big)  \big\}  } \sum_{i=0}^n {\bf E}^P \Big\{ 
%\log\Big(\frac{dQ_i(\cdot|V^{i-1}, A_i)}{d\Pi_i^{ P}(\cdot|B^{i-1})}(B_i)\Big) \Big\} \label{CM-TC-C_1_CC_B_A} 
%\end{align}

%Then directed information $I(A^n \rar B^n)$ is given by 
%\begin{align}
%I(A^n \rar B^n) = &  \sum_{i=0}^n {\bf E}^P \Big\{ 
%\log\Big(\frac{dQ_i(\cdot|A_i, V_0^{i-1})}{d\Pi_i^{ P}(\cdot|B^{i-1})}(B_i)\Big) \Big\}\equiv \sum_{i=0}^n I(A_i, V_{0}^{i-1}; B_i|B^{i-1}) \label{A-CM-TC-C_1_CC_B} \\
%=& \sum_{i=0}^n H(B_i|B^{i-1})- \sum_{i=0}^n H(V_i|V_{0}^{i-1}).
%\end{align}

\subsubsection{Non-stationary Non-ergodic AGN Channel: Orthogonal Decomposition \& Relation to Cover and Pombra \cite{cover-pombra1989}} Suppose the channel is the non-stationary  non-ergodic  Additive Gaussian Noise (AGN) channels  with memory,  
defined by (\ref{c-p1989}), i.e., the following hold
\begin{align}
\{V_i: i=0,1, \ldots, n\} \sim  N(\mu_{V^n}, K_{V^n}), \hso   \gamma_i^{C.0.0}(a_i,b_i) \tri |A_i|^2,  \hso i=0, \ldots, n. \label{c-p1989_G}    
\end{align}
Then by the entropy maximizing property of Gaussian processes, as in \cite{cover-pombra1989}, it follows from the definition of ${C}_{W ; B^n}^{FB,D.0}(\kappa)$ given by  (\ref{NEW_1}), that   the maximizing   channel input distribution induces a Gaussian joint distribution for the   joint process $\{(A_i, B_i, V_i)=(A_i^g, B_i^g, V_i): i=0, \ldots, n\}$,   the average constraint is satisfied, and condition (\ref{G-LCM_2_AG}) holds. Since  a linear combination of RVs is Gaussian if and only all RVs are Gaussian, then a 
  realization of the  channel input process corresponding to (\ref{NEW_1}) is the following.\\
  {\bf Orthogonal Decomposition.} 
 \begin{align}
  &A_i^g = \sum_{j=0}^{i-1} \gamma_{i,j}^1 B_j^g + \sum_{j=0}^{i-1} \gamma_{i,j}^2 V_j  + Z_i^g, \hso i=1, \ldots, n, \hso A_0=Z_0^g, \label{OP_STR_1_CP_n} \\
&\hst  \equiv  N_i +    M_i, \hso N_i\tri \sum_{j=0}^{i-1} \gamma_{i,j}^1 B_j^g + \sum_{j=0}^{i-1} \gamma_{i,j}^2 V_j, \hso         M_i\tri Z_i^g,   \label{DM_CP_n}\\
& Z_i^g  \:\:  \mbox{is  independent of}\:\:  \Big(A^{g, i-1}, B^{g,i-1}\Big), \; Z^{g,i} \hso \mbox{is independent of} \hso V^i, i=0, \ldots, n, \label{new_11_n}\\
& \Big\{Z_i^g \sim N(0, K_{Z_i}) : i=0,1, \ldots, n\Big\} \: \: \mbox{is an independent  Gaussian process} \label{new_12_n}
  \end{align}
for some deterministic sequences $\{(\gamma_{i,j}^1, \gamma_{i,j}^2): i=0, \ldots, n, j=0, \ldots, i-1\}$. The decomposition (\ref{DM_CP_n}) is unique due to the orthogonality condition (\ref{new_11_n}). 

{\bf Reduction of Orthogonal Decomposition to \cite{cover-pombra1989}.}  From the above decomposition, by recursive substitution, we can obtain the realization of optimal channel input distribution  derived by Cover and Pombra \cite{cover-pombra1989}, i.e.,  (\ref{cp1989_a}),  (\ref{cp1989}), as follows. 
\begin{align} 
&A_i^g = \sum_{j=0}^{i-1} \overline{\gamma}_{i,j}^1 V_j  + \overline{Z}_i^g, \hso   \overline{Z}_i^g \tri\sum_{j=0}^i \overline{\gamma}_{i,j}^2 Z_j^g, \hso  i=1, \ldots, n, \hso A_0=\overline{Z}_0^g, \label{OP_STR_1_CP_SC}\\
%& \hst \equiv  \sum_{j=0}^{i-1} \overline{\gamma}_{i,j}^1 V_j  + \overline{Z}_i^g, \hst \overline{Z}_i^g \tri\sum_{j=0}^i \gamma_{i,j}^2 Z_j^g, \\
& \overline{Z}^{g,i} \hso \mbox{is independent of} \hso V^i,\hso i=0, \ldots, n,\\
& \Big\{\overline{Z}_i^g: i=0,1, \ldots, n\Big\} \: \: \mbox{zero mean correlated   Gaussian process} \label{new_12_SC}
\end{align}
for some deterministic sequences $\{\overline{\gamma}_{i,j}^1: j=0, \ldots, i-1\}$, $ \{\overline{\gamma}_{i,j}^2: j=0, \ldots, i\}$, $i=0, \ldots, n$.  Thus, the realization via the  orthogonal decomposition  (\ref{DM_CP_n}) is equivalent to the Cover and Pombra \cite{cover-pombra1989} realization, i.e.,  (\ref{OP_STR_1_CP_SC}),  (\ref{new_12_SC}). 

However, as illustrated in Theorem~\ref{gen_exa}, since the objective is to compute the  characterization of FTFI capacity given by (\ref{NEW_1}), then  the orthogonal  decomposition  realized by (\ref{OP_STR_1_CP_n})-(\ref{new_12_n}), in which  the process $\{Z_i^g: i=0, \ldots, n\}$ is an orthogonal process, is more convenient  compared to  the non-orthogonal decomposition realized by  (\ref{OP_STR_1_CP_SC}), (\ref{new_12_SC}), in which $\{\overline{Z}_i^g: i=0, \ldots, n\}$ is correlated. This is possibly one of the main  reason,  which prevented many of the past attempts to solve explicitly,  the non-stationary non-ergodic Cover and Pombra \cite{cover-pombra1989} characterization of FTFI capacity or its stationary variants \cite{yang-kavcic-tatikonda2007,kim2010}, and to    generalize it to  Multiple Input Multiple Output (MIMO) Gaussian channels, with past dependence on past channel input and output. 
%In Section~\ref{exa_gen}, we derive  analogous expressions for MIMO Gaussian recursive  LCMs with  quadratic  transmission cost functions.

% defined as follows.
%\begin{align}
%&\gamma_i^{C.0.0}(a_i, b_{i}) \tri \langle a_i, R_{i,i} a_i \rangle + \langle b_{i}, Q_{i,i} b_{i} \rangle, \hso i=0, \ldots, n, \label{G-LCM_4_G_1}  \\
%& {\mathbb A}_i={\mathbb B}_i={\mathbb V}_i\tri   {\mb R}^{p}, \hso R_{i,i} \in {\mb S}_{++}^{q\times q}, \hso Q_{i,i} \in {\mb S}_{+}^{p\times p}, \hso i=0, \ldots, n \label{G-LCM_6_G_new} 
%\end{align}

\subsubsection{Arbitrary Distributed  Limited Memory Noise}
\label{LMN}
Suppose the noise distribution is limited memory defined by 
\begin{align}
 {\bf P}_{V_i|V^{i-1}, A^i}={\bf P}_{V_i|V_{i-L}^{i-1}}, \hso i=0, \ldots, n. \label{LM_G-LCM_2_AG}  
 \end{align}
 Then  by  Theorem~\ref{CC_CP_A},  the optimal channel input distribution  
  is also limited memory, and satisfies conditional independence $\big\{{\bf P}_{A_i|V^{i-1}, B^{i-1}}(da_i|v^{i-1}, b^{i-1})= \pi^{D.0,L}(da_i|v_{i-L}^{i-1}, b^{i-1}): i=0, \ldots, n\}$. Moreover, the characterization of FTFI capacity is given as follows.
\begin{align}
{C}_{W ; B^n}^{FB,D.0,L}(\kappa) \tri & \sup_{ \overline{\cal P}_{[0,n]}^{D.0,L}(\kappa) } \sum_{i=0}^n {\bf E}^{\pi^{D.0,L}} \Big\{ 
\log\Big(\frac{dQ_i(\cdot|V_{i-L}^{i-1}, A_i)}{d\Pi_i^{\pi^{D.0,L}}(\cdot|B^{i-1})}(B_i)\Big) \Big\} \label{CM-TC-C_1_CC_B_A_LM} 
\end{align}
where 
\begin{align}
  \Pi_i^{\pi^{D.0,L}}(db_i | b^{i-1}) =& \int_{  {\mb A}_{i}\times {\mathbb V}_{i-L}^{i-1} }   Q_i(db_i | a_i, v_{i-L}^{i-1}) \otimes  P_i(da_i |v_{i-L}^{i-1}, b^{i-1}) \otimes {\bf P}^{\pi^{D.0,L}}(dv_{i-L}^{i-1} | b^{i-1})  , \hso i=0, \ldots, n, \label{A-CM-TC-C_3_CC_B_1}\\
\overline{\cal P}_{[0,n]}^{D.0,L}(\kappa) \tri& \Big\{\pi_i^{D.0,L}(da_i| v_{i-L}^{i-1},  b^{i-1}) , i=0,\ldots,n:  \frac{1}{n+1} {\bf E}^{\pi^{D.0,L}}\Big( \sum_{i=0}^n \gamma_i^{C.0, 0}(A_i, B_{i}) \leq \kappa  \Big)     \Big\}.
\end{align}
Note that for each $i$, then $v_{i-L}^{i-1}$ are the state variables, known to the encoder but  unknown to the decoder, hence they  need to be estimated at the decoder. Moreover, it is straight forward to    verify that the \'a posteriori distribution $\big\{ {\bf P}^{\pi^{D.0,L}}(dv_{i-L}^{i-1}|b^{i-1}): i=0, \ldots, n\big\}$ satisfies a recursion similar to   (\ref{rec_g1})-(\ref{rec_g3}). \\
The characterization of FTFI capacity, (\ref{CM-TC-C_1_CC_B_A_LM}) holds for  finite and continuous alphabet spaces, and any  combination of them, for  arbitrary distributed noise. 

\subsubsection{Non-stationary Non-ergodic AGN Channel with Limited Noise Memory \& Orthogonal Decomposition}  Suppose the channel is  
defined by (\ref{c-p1989}) and the noise is Gaussian and  limited memory, i.e., (\ref{c-p1989_G}) and (\ref{LM_G-LCM_2_AG}) hold. Then we can  show using the recursion (\ref{A-CM-TC-C_3_CC_B_1}) that   the optimal channel input distribution is Gaussian. Moreover,   we deduce the following realization of the  channel input process.\\
{\bf Orthogonal Decomposition.}
 \begin{align}
  &A_i^g = \sum_{j=0}^{i-1} \gamma_{i,j}^1 B_j^g + \sum_{j=1}^{L} \gamma_{i,i-j}^2 V_{i-j}  + Z_i^g, \hso i=1, \ldots, n, \hso A_0=Z_0^g, \label{OP_STR_1_CP_LM} \\
&\hst  \equiv  N_i +    M_i, \hso M_i\tri Z_i^g,  \hso \mbox{(\ref{new_11}) and (\ref{new_12}) hold.} \label{DM_CP_LM}
  \end{align} 
The above property of optimal channel input distribution, i.e., its dependence on limited memory on the channel noise,  is new and did not appear in the literature. It compliments similar results obtained   by Kim in \cite{kim2010}, for the stationary ergodic case, where the author applied  frequency domain methods to the  Cover and Pombra \cite{cover-pombra1989} decomposition of optimal channel input process,   to show that if the noise power spectral density corresponds to a stationary Gaussian autoregressive moving-average model of order $K$, then a
%the optimal channel input conditional distribution is also of order $K$,
%and that a 
 $K-$dimensional generalization of the Schalkwijk-Kailath coding scheme achieves feedback capacity.\\
However, our analysis is based on the  information structures derived in this paper, it is  strictly probabilistic, and applies to general channels.

\section{Achievability}
\label{ach}
Many  existing coding  theorems  found in  \cite{ihara1993,kramer1998,kramer2003,
chen-berger2005,kim2008,tatikonda-mitter2009,permuter-weissman-goldsmith2009,kim2010}, are either  applicable or can be generalized to show  the  per unit time limiting versions of the characterizations of FTFI capacity, corresponds to feedback capacity, under appropriate conditions. \\
Next, we provide a short elaboration on technical issues, which need to be resolved, in order to ensure, under relaxed conditions (i.e., without imposing stationarity,  ergodicity, or assuming finite alphabet spaces), that  the per unit time  limiting versions of the characterizations of FTFI capacity correspond to the supremum of all achievable feedback codes. \\ 
For Class A, B, C channel distributions and transmission cost functions, it is shown by Massey in \cite{massey1990}, that directed information $I(A^n \rar B^n)$ gives a tight bounds on any achievable code rate (of feedback codes).  This follows from the converse coding theorem \cite{kramer1998,kim2008,tatikonda-mitter2009,permuter-weissman-goldsmith2009}, similar to the converse coding theorem of NCM-D, given in Theorem~\ref{CC_CP_A}.
Via these tight bounds,  the direct part of the coding theorem can be shown,  by investigating the per unit time limit of the characterizations of FTFI capacity, without unnecessary \'a priori assumptions on  the channel, such as, stationarity, ergodicity, or information stability of the joint process $\{(A_i, B_i): i=0, 1, \ldots\}$.\\
Further, through the characterizations of FTFI capacity, several hidden properties of the role of optimal channel conditional distributions   to affect the channel output transition probability distribution can  be  identified. \\
Next,  we   state  the fundamental conditions, in order to make the transition to the per unit time limiting versions of the characterizations of FTFI capacity, and to give an operational meaning to these characterizations.  

{\bf (C1)} For any source process $\big\{X_i: i=0, \ldots, \big\}$ to be encoded and transmitted over the channel,   the  conditional independence  condition (\ref{CI_Massey_N}) is satisfied    \cite{massey1990}. 
%\begin{align}
%{\bf P}_{B_i|B^{i-1}, A^i, X^k}={\bf P}_{B_i|B^{i-1}, A^i} \hso   \forall k \in \{0,1, \ldots, n\},\hso i=0, \ldots, n \label{CI_Massey} 
%\end{align}
As pointed out by Massey  \cite{massey1990}, conditional independence  condition (\ref{CI_Massey_N}), is a necessary condition for  directed information $I(A^n \rar B^n)$ to give a tight upper bound on the information conveyed by the source to the channel output (Theorem~3 in \cite{massey1990}),    and that directed information reduces to mutual information in the absence of feedback, that is, if ${\bf P}_{A_i|A^{i-1}, B^{i-1}}={\bf P}_{A_i|A^{i-1}}, i=0, \ldots, n$, then $I(A^n \rar B^n)=I(A^n ; B^n)$.

{\bf (C2)} For any of the channels and transmission cost functions investigated, there exist channel input conditional distributions denoted by  $\big\{{\pi}_i^*(da_i|{\cal I}^{\bf P}): i=0, \ldots, n\big\} \in {\cal P}_{[0,n]}(\kappa)$ (if transmission cost is imposed),  which achieve the supremum of the characterizations of FTFI capacity, and their per unit time limits exist and they are finite.\\
%, i.e.,  the supremum in , and  $\liminf_{n \longrightarrow \infty} \frac{1}{n+1} C_{A^n \rar  B^n}^{FB, C.I,J}(\kappa)$ is finite.  \\
For the converse part of the channel coding theorem,  existence, i.e., (C2),  is necessary, because it is often shown by invoking  Fano's inequality, which  requires finiteness of $\liminf_{n \longrightarrow \infty} \frac{1}{n+1} C_{A^n \rar  B^n}^{FB}(\kappa)$. Similarly, the direct part of the coding theorem is often shown  by generating channel codes according to the channel input distributions, which achieve  $\liminf_{n \longrightarrow \infty} \frac{1}{n+1} C_{A^n \rar  B^n}^{FB}(\kappa)$.  Hence, the derivation of coding theorems pre-supposes existence of optimal channel input distributions and finiteness of the limiting expression.     \\
Since, for continuous and countable alphabet spaces, $\big\{({\mb A}_i, {\mb B}_i): i=0, \ldots, n\big\}$, information theoretic measures are not necessarily continuous functions on the space of distributions \cite{ho-yeung2009ieeeit}, and that,  directed information is lower semicontinuous,   as a functional of channel input conditional distributions  $\big\{{\bf P}_{A_i|A^{i-1}, B^{i-1}}: i=0, \ldots, n\big\} \in {\cal P}_{[0,n]}$, sufficient conditions for continuity of directed information should be  identified. Such conditions are given in  \cite{charalambous-stavrou2013aa}. However, for finite alphabet spaces such technicalities do not arize, and hence one can invoke the various  coding theorems derived in \cite{kramer1998,kramer2003,
chen-berger2005,kim2008,tatikonda-mitter2009,permuter-weissman-goldsmith2009,kim2010} are applicable.

{\bf (C3)} The optimal channel input distributions  $\big\{\pi_i^*(da_i| {\cal I}_i^P): i=0,1,\ldots, n\big\} \in  {\cal P}_{[0,n]}(\kappa)$, which achieve the supremum of the characterizations of FTFI capacity,  induce stability in the sense of Dobrushin \cite{pinsker1964}, of the  directed information density, that is,  
\begin{align}
 \lim_{n \longrightarrow \infty} {\bf P}^{\pi^*} \Big\{(A^n ,B^n) \in {\mb A}^n \times {\mb B^n}:\frac{1}{n+1} \Big|{\bf E}^{\pi^*} \big\{ {\bf i}^{\pi^*}(A^n, B^n)\big\} - {\bf i}^{\pi^*}(A^n, B^n) \Big| > \varepsilon \Big\} =0  \label{IS-O_1}
\end{align}
and stability of  the transmission cost constraint, that is,  
  \begin{align}
 \lim_{n \longrightarrow \infty} {\bf P}^{\pi^*} \Big\{ (A^n ,B^n) \in {\mb A}^n \times {\mb B^n}: \frac{1}{n+1} \Big| {\bf E}^{\pi^*} \Big\{\sum_{i=0}^n \gamma_i(T^i A^n, T^iB^{n})\Big\} -  \sum_{i=0}^n \gamma_i(T^i A^n, T^iB^{n}) \Big| > \varepsilon \Big\} =0. \label{IS-O_2}
\end{align}
For example, for any channel distribution of Class C, and any transmission cost of Class C, the directed information density is  
\bea
{\bf i}^{\pi^*}(A^n, B^n)\equiv {\bf i}^{\pi^{*,A.I}}(A^n, B^n)  \tri \sum_{i=0}^n \log\Big(  \frac{{Q}_i(\cdot|B_{i-M}^{i-1}, A_{i-L}^i)}{\Pi_i^{\pi^{*,A.I}}(\cdot|B^{i-1})}(B_i)\Big), \hso i=0, \ldots, n, \hso I\tri \max\{L, N\}
\eea
and  similarly for the rest of the characterizations of FTFI capacity derived in the paper. The important research question of showing (\ref{IS-O_1}) and (\ref{IS-O_2}) requires extensive analysis, especially, for abstract alphabet spaces (i.e., continuous), and this is beyond the scope of this paper. For finite alphabet spaces  various  coding theorems derived in \cite{kramer1998,kramer2003,
chen-berger2005,kim2008,tatikonda-mitter2009,permuter-weissman-goldsmith2009,kim2010} are applicable. 

Condition (C1) implies the well-known data processing inequality, while condition (C2) implies  existence of the optimal channel input distributions and finiteness of the corresponding characterizations of the FTFI capacity and their per unit time limits. Condition (C3) is sufficient to ensure the AEP holds, and hence   standard random coding arguments hold,   i.e., following  Ihara \cite{ihara1993}, by replacing the information density of mutual information by the directed information density.  \\
Finally, we note that, for specific application examples,  it is possible to  invoke  the characterizations of FTFI capacity derived in this paper, to compute the expressions of error exponents derived in \cite{permuter-weissman-goldsmith2009}, and establish coding theorems via this alternative direction.

\section{Conclusion}
We derived structural properties of optimal channel input conditional distributions, which maximize   directed information from channel input RVs to channel output RVs, for general channel distributions with memory, with and without transmission cost constraints, and we obtained the corresponding characterizations of FTFI capacity. These are characterized by channel input distributions, which satisfy  conditional independence. We have also derived similar structural properties for general Nonlinear Channel Models (NCM) driven by correlated noise processes. \\
%These structural properties generalize the structural properties of  Memoryless Channels with feedback,  and Shannon's two-letter characterization of channel capacity,    to  channels with memory.  \\
We have applied one of the characterizations of FTFI capacity  to recursive Multiple Input Multiple Output Gaussian Linear Channel Models, with  limited  memory on channel input and output sequences,  under  general transmission cost constraints, and we have established a separation principle. The separation principle is based on  realizing optimal channel input distributions by randomized strategies, using  orthogonal decompositions. The  feedback capacity can be  obtained via its per unit time limiting version and standard results on  ergodic Markov Decision theory.\\
In future work, it is of interest to understand the role of feedback to control the channel output process, to derive, for specific channel models,  closed form expressions for the characterizations of FTFI capacity  and feedback capacity, and to determine whether feedback increases capacity, and by how much.   \\
Whether the methodology of this paper  can be applied to extremum problems of network information theory, to identify information structures of optimal distributions and achievable upper bounds, remains, however, a subject for further research.

\bibliographystyle{IEEEtran}
\bibliography{Bibliography_capacity}

\end{document}